\documentclass[11pt]{article}
\usepackage[a4paper,margin=1in]{geometry}
\usepackage[T1]{fontenc}
\usepackage[utf8]{inputenc}
\usepackage{amsmath,amssymb,amsthm,mathtools}
\usepackage{stmaryrd}
\usepackage{enumitem}
\usepackage{hyperref}
\usepackage[expansion=false]{microtype}
\usepackage{tikz}
\usepackage{algorithm}
\usepackage{algpseudocode}
\hypersetup{colorlinks=true,linkcolor=blue,citecolor=blue,urlcolor=blue}

\newtheorem{theorem}{Theorem}[section]
\newtheorem{lemma}[theorem]{Lemma}
\newtheorem{proposition}[theorem]{Proposition}
\newtheorem{corollary}[theorem]{Corollary}
\newtheorem{definition}[theorem]{Definition}
\newtheorem{assumption}[theorem]{Assumption}
\newtheorem{example}[theorem]{Example}
\newtheorem{remark}[theorem]{Remark}

\newcommand{\Pre}{\operatorname{Pre}}
\newcommand{\Viab}{\operatorname{Viab}}

\newcommand{\Cert}{\operatorname{Cert}}
\newcommand{\Val}{\operatorname{Val}}

\newcommand{\Sat}{\operatorname{Sat}}
\newcommand{\dist}{\operatorname{dist}}
\newcommand{\Attr}{\operatorname{Attr}}

\newcommand{\Post}{\operatorname{Post}}
\newcommand{\obs}{\operatorname{obs}}
\newcommand{\bpost}{\operatorname{post}}
\newcommand{\rk}{\operatorname{rk}}
\newcommand{\marg}{\operatorname{marg}}
\newcommand{\argmin}{\operatorname*{arg\,min}}
\newcommand{\argmax}{\operatorname*{arg\,max}}
\newcommand{\Bel}{\operatorname{Bel}}

\newcommand{\Rev}{\operatorname{Rev}}
\newcommand{\Just}{\operatorname{Just}}
\newcommand{\JAct}{\operatorname{JAct}}

\newcommand{\Opq}{\operatorname{Opq}}

\newcommand{\CAct}{\operatorname{CAct}}

\newcommand{\Res}{\operatorname{Res}}

\newcommand{\rank}{\operatorname{rank}}

\newcommand{\Der}{\operatorname{Der}}
\newcommand{\Lic}{\operatorname{Lic}}
\newcommand{\Mod}{\operatorname{Mod}}
\newcommand{\NoAdd}{\operatorname{NoAdd}}
\newcommand{\BYSA}{\operatorname{BYSA}}
\newcommand{\Root}{\operatorname{Root}}

\newcommand{\Cell}{\operatorname{Cell}}

\newcommand{\Out}{\operatorname{Out}}

\newcommand{\Loc}{\operatorname{Loc}}

\title{Value-Refined Modal Fixed-Point Semantics with Certified Choice and Public Share-Alike Certificates}
\author{Faruk Alpay\textsuperscript{*}\quad Levent Sar{\i}o\u{g}lu\\
Department of Computer Engineering, Bah\c{c}e\c{s}ehir University, Istanbul, Turkey\\
\texttt{\{faruk.alpay, levent.sarioglu\}@bahcesehir.edu.tr}}
\date{}

\begin{document}
\maketitle

\begin{abstract}
This article studies a finite modal semantics in which truth is closed under admissible continuation before it is refined by value and certified by finite residual tests.  The underlying admissibility kernel is the classical greatest fixed point $K_G=\nu X.(G\cap\Pre_{\exists\Box}X)$, where $[\exists]\Box$ has the standard one-step reading that some choice cell has all compatible successors inside the tested set.  The contribution is the canonical value-certified completion of that fixed point: certified choices are exactly its local witnesses, the discounted value transformer is defined only over those witnesses, and value-refined modal bisimulation is the coarsest local equivalence preserving formulas, kernels, certified choices, Bellman values, greedy sets, residual certificates, and public release certificates.  A canonical pseudometric refinement grades this exact equivalence: it is the unique fixed point of a Hausdorff-lifted choice-matching transformer over certified choices, its zero set is exactly the value-refined modal bisimulation, and the discounted value is $1$-Lipschitz with respect to it, so an approximate quotient incurs only a distance-bounded value error.  Branching choice-cell and locus presentations place the choice-bearing structure inside the model, while the earlier transition presentation is recovered as a conservative retraction preserving every displayed operator.  The same engine is then applied internally to a public share-alike release fragment: attribution is treated as label preservation, same-license propagation as derivative closure, absence of downstream restriction as an admissibility clause, and the resulting BY-SA witness as a residual-stable certificate.  Finite separating examples and a completion theorem show that changing the order of truth, admissibility, value, quotienting, public derivation, and certification changes the semantics rather than merely changing notation.
\end{abstract}


\medskip
\noindent\textbf{Public release certificate.} This version declares the public certificate $\BYSA(M)$ for the manuscript object $M$: attribution to the named authors, a persistent reference to Creative Commons Attribution--ShareAlike 4.0 International, retention of modification traces for derivative manuscripts, and exclusion of additional downstream restrictions.  Section~\ref{sec:bysa-certificate} records this declaration as a semantic certificate generated by the same admissibility, derivative-closure, value-refinement, and residual-certificate components used in the main construction.

\section{Introduction}

A finite modal structure can carry more than truth at a point.  It can carry a structured choice of continuations, an order on the continuations that remain available, and a certificate explaining why a finite computation is enough.  Greatest fixed points, one-step coalition ability, choice-cell semantics, bisimulation, and discounted value equations are classical tools.  The question studied here is which ordered semantic construction preserves all of them at once.  Public share-alike release is used as an internal stress test for the construction, since it requires labels, derivative closure, admissible continuations, quotient preservation, and finite certificate acceptance to agree in one finite fragment.

The preserved order is
\[
\text{modal satisfaction}
\longrightarrow
\text{admissible continuation}
\longrightarrow
\text{value refinement}
\longrightarrow
\text{certificate}.
\]
The first layer is ordinary satisfaction in a labelled modal structure.  The second layer is the greatest fixed point
\[
K_G=\nu X.\bigl(G\cap\Pre_{\exists\Box}(X)\bigr),
\]
where $\Pre_{\exists\Box}$ has the standard one-step shape: there is a choice cell whose every compatible successor remains in $X$.  The third layer compares values only over cells witnessing membership in the fixed point.  The fourth layer records residual, quotient, and stopping certificates as semantic objects rather than as implementation annotations.

The finite core is deliberate.  State sets and choice sets are finite, costs are bounded, and the value equation is discounted, so every displayed operator is exact, every chain stabilizes, and every quotient claim has a finite witness.  Infinite-state, undiscounted, probabilistic, and approximation-heavy variants appear later only as delimited boundary fragments.  They do not enlarge the main theorem; they test how much of its layer order survives when one rigidity is removed.

The main presentation is model-internal.  A point is not only labelled by propositions; it carries cells of histories and, in the locus-indexed refinement, sites at which the model partitions those histories.  A semantic selection chooses a cell.  No external bearer is introduced.  The construction belongs to the line running through branching-time choice semantics, coalition and ATL one-step ability, modal $\mu$-calculus fixed points, bisimulation, and quantitative modal semantics.

The transition calculus from which the construction begins is retained as a conservative retraction.  Every transition choice induces a cell in an unfolded branching presentation, and every cell presentation flattens back to the same predecessor, kernel, certified choice set, value transformer, quotient semantics, and residual certificates.  This is the formal place where the earlier transition-style fixed point keeps its role: it is not the whole semantics, but it is a preserved fragment of the richer semantics.

Several expanded mechanisms are retained only as boundary fragments.  Their role is diagnostic: they mark which clauses survive when a rigidity is loosened, without changing the finite theorem cluster that bears the main proof load.

The vocabulary is kept small.  A \emph{target} names the modal region to be maintained.  A \emph{kernel} is the greatest fixed point of admissible continuation.  A \emph{certified choice} is a witness to the fixed-point equation at the current point.  A \emph{value refinement} ranks only those witnesses.  A \emph{certificate} records when a finite computation is semantically adequate.  Each term is tied below to a displayed operator or relation.

\subsection{Scope and semantic boundary}

The basic fixed point is classical: it is the greatest fixed point generated by an admissible predecessor.  The semantic content lies in the enforced composition around that fixed point.  Modal truth is first closed under admissible continuation; only then is it refined by a discounted value transformer; only value-refined equivalences may quotient the structure; and finite stopping is accepted only through residual certificates.  Weakening or reordering these layers is not a stylistic change: finite witnesses in Section~\ref{sec:worked-discipline} show that the induced semantics changes.

The scope of the main theorem is exact and finite.  It does not introduce additional syntax for the modal $\mu$-calculus, a replacement for ATL, or a full extension of STIT.  It uses their established one-step and fixed-point mechanisms to define a combined semantics in which formulas, admissible continuation, values, quotients, and certificates are preserved by one local relation.  The boundary sections later in the paper record how the construction can be relaxed, but the central result is the finite theorem cluster stated next.

\subsection{Theorem cluster}

The results are organized around four core claims and two auxiliary representation claims.  The first core claim is that the admissible-continuation operator is monotone and its descending Kleene chain stabilizes after at most $|S|$ strict removals; its limit is exactly the largest region admitting certified cells.  The second core claim is that the value transformer restricted to certified cells is a $\gamma$-contraction and has a unique fixed point with stationary value-optimal selectors.  The third core claim is that value-refined modal bisimulation preserves base formulas, the positive $\mu$-calculus fragment generated by $[\exists]\Box$, kernel membership, Bellman values, greedy choice sets, and residual stopping certificates.  The fourth core claim is that the greatest fixed point of the local bisimulation transformer is the coarsest quotient preserving the combined formula-kernel-value-certificate semantics; each clause of the transformer has a finite separating witness.  A fifth, quantitative claim grades this quotient: the value-refined modal bisimulation is the zero set of a canonical pseudometric---the unique fixed point of a Hausdorff-lifted choice-matching transformer over certified choices---with respect to which the optimal value is $1$-Lipschitz, so the exact quotient is the distance-zero case and an approximate quotient carries only a distance-bounded value error.

The two auxiliary representation claims explain why the choice vocabulary is model-internal.  Every finite transition-style modal choice structure has a branching choice-cell presentation with the same predecessor, kernel, certified choices, values, quotients, and certificates.  Conversely, locus-indexed choice cells assign the choice-bearing structure to the model itself while flattening conservatively to the unindexed cell semantics and to the retained transition fragment.

\subsection{Relation to established fixed-point and choice semantics}

The relation to the nearest formalisms is exact.  Relative to the modal $\mu$-calculus, the kernel is the familiar greatest fixed-point formula
\[
\nu Z.(\varphi_G\wedge[\exists]\Box Z),
\]
where $\varphi_G$ defines the target region.  The present construction does not alter the fixed-point principle; it constrains which later value and certificate operations may be placed on top of it.

Relative to coalition logic and ATL, $[\exists]\Box$ is a one-step ability modality.  The construction uses only that one-step ability inside a greatest fixed point; it does not require the full temporal language of ATL or ATL$^*$.  This is why the comparison is local: the semantic burden is not temporal expressiveness but preservation after value refinement and quotienting.

Relative to branching-time choice-cell semantics, cells locate choice inside the model as partitions of histories at a moment.  No additional deontic or intentional vocabulary is required.  The cell structure is used only to give a model-internal semantics to certified choice and to justify the conservative retraction back to the transition presentation.

Relative to quantitative modal and game semantics, the discounted value equation is standard once the certified domain is fixed.  What is not standardly preserved by ordinary modal bisimulation is the joint package of modal truth, certified choices, costs, successor classes, greedy sets, and residual certificates.  The value-refined bisimulation theorem is the exact quotient discipline for that package.

Relative to behavioural pseudometrics, the graded refinement of the quotient layer is a bisimulation metric in the sense of Desharnais--Gupta--Jagadeesan--Panangaden, van~Breugel--Worrell, and Ferns--Panangaden--Precup \cite{desharnais2004,vanbreugel2005,ferns2004}, specialized to the alternating choice-cell setting: the probabilistic Kantorovich lift is replaced by a Hausdorff lift because compatible successors are resolved universally rather than in expectation, and the value-Lipschitz theorem is the choice-refined counterpart of the result that the optimal value is nonexpansive in the bisimulation metric.  The same order also explains the certified-domain map: admissibility fixes the local witnesses first, and value refinement ranks only those witnesses afterward.

\section{Axiomatic semantic core}
\label{sec:axiomatic-core}

The core definitions appear before the longer derivations, examples, and boundary cases.

\begin{definition}[Modal choice structure]
A finite modal choice structure is a tuple
\[
\mathfrak M=(S,U,A,T,L,c),
\]
where $S$ is a finite state set, $U$ is a finite choice alphabet, $A(s)\subseteq U$ is nonempty for every $s$, $T(s,u)\subseteq S$ is nonempty for every $u\in A(s)$, $L:S\to 2^{AP}$ labels states by atomic propositions, and $c(s,u)\in[0,\bar c]$ is a bounded value cost.  For $X\subseteq S$, define
\[
\Pre_{\exists\Box}(X)=\{s\in S: \exists u\in A(s)\; T(s,u)\subseteq X\}.
\]
\end{definition}

\begin{definition}[Admissible continuation kernel and certified choices]
For a modal target $G\subseteq S$, define
\[
F_G(X)=G\cap\Pre_{\exists\Box}(X),\qquad
K_G=\nu X.F_G(X).
\]
For $s\in K_G$, the certified choice set is
\[
\CAct_G(s)=\{u\in A(s):T(s,u)\subseteq K_G\}.
\]
\end{definition}

\begin{theorem}[Fixed-point admissibility]
\label{thm:core-fixedpoint}
The set $K_G$ is the largest subset $K\subseteq G$ such that every $s\in K$ has at least one choice $u\in A(s)$ with $T(s,u)\subseteq K$.  The descending sequence $X_0=S$, $X_{n+1}=F_G(X_n)$ stabilizes at $K_G$ after at most $|S|$ strict inclusions.  Moreover $\CAct_G(s)$ is nonempty for every $s\in K_G$, and any selector $\sigma(s)\in\CAct_G(s)$ preserves $K_G$ under all compatible successors.
\end{theorem}

\begin{proof}
The map $F_G$ is monotone on the finite complete lattice $2^S$.  Tarski's theorem gives a greatest fixed point, and the descending Kleene chain from $S$ reaches it because a strict descent removes at least one state.  If $s\in K_G$, the fixed-point equality $K_G=G\cap\Pre_{\exists\Box}(K_G)$ gives $s\in G$ and a witness $u\in A(s)$ with $T(s,u)\subseteq K_G$, hence $u\in\CAct_G(s)$.  Induction on time then proves that every path generated by a selector taking only certified choices remains in $K_G\subseteq G$.  Conversely, if $Y\subseteq G$ has the same one-step closure property, then $Y\subseteq F_G(Y)$, so $Y$ is post-fixed and must be contained in the greatest fixed point $K_G$.
\end{proof}

\begin{definition}[Value-refinement transformer]
Fix $0<\gamma<1$.  On $\ell_\infty(K_G)$ define
\[
(\mathcal V_G J)(s)=\min_{u\in\CAct_G(s)}\left(c(s,u)+\gamma\max_{s'\in T(s,u)}J(s')\right).
\]
A selector is value-greedy when it realizes the displayed minimum at every $s\in K_G$.
\end{definition}

\begin{theorem}[Unique value refinement and certified optimality]
\label{thm:core-value}
The operator $\mathcal V_G$ is a $\gamma$-contraction on $\ell_\infty(K_G)$.  Hence it has a unique fixed point $J_G^\star$, value iteration converges uniformly to $J_G^\star$, and every value-greedy selector preserves $K_G$ while minimizing the worst-compatible discounted value among all selectors that preserve $K_G$.
\end{theorem}

\begin{proof}
For $J,H\in\ell_\infty(K_G)$ and each fixed $s,u$,
\[
\left|\max_{s'\in T(s,u)}J(s')-\max_{s'\in T(s,u)}H(s')\right|
\leq \|J-H\|_\infty.
\]
Multiplication by $\gamma$ and taking minima over the same nonempty finite set $\CAct_G(s)$ gives
\[
|\mathcal V_GJ(s)-\mathcal V_GH(s)|\leq \gamma\|J-H\|_\infty.
\]
Taking the supremum over $s$ proves contraction.  Banach's fixed-point theorem gives existence, uniqueness, and convergence.  A value-greedy selector uses only certified choices, so invariance follows from Theorem~\ref{thm:core-fixedpoint}.  For any certified selector $\sigma$, its selector-evaluation operator is also a $\gamma$-contraction, and $\mathcal V_GJ\leq \mathcal V_G^\sigma J$ pointwise.  Monotone contraction iteration gives $J_G^\star\leq J^\sigma$ for every certified $\sigma$, while equality holds for a greedy selector because its selector operator agrees with $\mathcal V_G$ at the fixed point.
\end{proof}

\begin{definition}[Value-refined modal bisimulation]
An equivalence relation $R$ on $S$ is value-refined for $\mathfrak M$ when $sRt$ implies $L(s)=L(t)$ and, for every $u\in A(s)$, there exists $v\in A(t)$ such that
\[
c(s,u)=c(t,v),\qquad [T(s,u)]_R=[T(t,v)]_R,
\]
and symmetrically from $t$ to $s$.  Here $[X]_R$ is the set of $R$-classes met by $X$.
\end{definition}

\begin{theorem}[Preservation and coarsest quotient]
\label{thm:core-quotient}
Let $G$ be definable in the base modal language and let $R$ be value-refined.  If $sRt$, then $s$ and $t$ satisfy the same positive $\mu$-calculus formulas generated by Boolean connectives, $[\exists]\Box$, and greatest fixed points; they are either both in $K_G$ or both outside it; and, when they are in $K_G$, $J_G^\star(s)=J_G^\star(t)$.  The union of all value-refined equivalences is the greatest fixed point of the local relation transformer induced by the two matching clauses above, and therefore gives the coarsest quotient preserving the combined formula-kernel-value semantics.
\end{theorem}

\begin{proof}
Formula preservation is the standard bisimulation induction, with the $[\exists]\Box$ step using equality of successor $R$-class sets: a certified existential choice on one side has a matched choice on the other side whose every successor class is the same.  Greatest fixed points are preserved because the descending approximants of a positive $\mu$-formula remain unions of $R$-classes.  Since $G$ is a union of $R$-classes and $\Pre_{\exists\Box}$ maps unions of $R$-classes to unions of $R$-classes, every approximant of $K_G$ is class-saturated; hence $K_G$ is class-saturated.  On $K_G$, certified choices also match certified choices, because matched successor classes are already contained in the class-saturated kernel.  The Bellman backups at related states therefore minimize over cost-equal choices with identical successor-class maxima.  The contraction fixed point is unique, so the class-constant solution on the quotient lifts exactly to $J_G^\star$ on $K_G$.  Finally, the local matching requirements define a monotone transformer on equivalence relations; its greatest fixed point contains exactly the pairs surviving all finite modal/value tests.  Any quotient preserving the displayed one-step semantic data is a post-fixed point of this transformer and is contained in its greatest fixed point.
\end{proof}

\begin{definition}[Residual certificate]
For an approximation $J\in\ell_\infty(K_G)$, define the residual
\[
\Res_G(J)=\|J-\mathcal V_GJ\|_\infty.
\]
\end{definition}

\begin{theorem}[Certified finite approximation]
\label{thm:core-certificate}
For every $J\in\ell_\infty(K_G)$,
\[
\|J-J_G^\star\|_\infty\leq \frac{\Res_G(J)}{1-\gamma}.
\]
If a greedy selector for $J$ is used and has one-step greedy gap at least $2\Res_G(J)/(1-\gamma)$ against every nongreedy certified choice, then it is also greedy for $J_G^\star$.
\end{theorem}

\begin{proof}
Using $J_G^\star=\mathcal V_GJ_G^\star$ and contraction,
\[
\|J-J_G^\star\|_\infty
\leq \|J-\mathcal V_GJ\|_\infty+\|\mathcal V_GJ-\mathcal V_GJ_G^\star\|_\infty
\leq \Res_G(J)+\gamma\|J-J_G^\star\|_\infty,
\]
which gives the residual bound.  The second statement follows by perturbing each one-step backup by at most the value error: if every nongreedy choice is separated by twice that error bound, the ordering of greedy backups cannot be reversed at the exact fixed point.
\end{proof}

\begin{definition}[Exact value-certified extension]
Let $G\subseteq S$ be a modal target.  An exact value-certified extension of the underlying one-step semantics is a tuple
\[
(D,C,\mathcal B,\equiv),
\]
where $D\subseteq G$, $C(s)\subseteq A(s)$ is nonempty for every $s\in D$, $T(s,u)\subseteq D$ for every $u\in C(s)$, $\mathcal B$ is a discounted Bellman-type transformer on $\ell_\infty(D)$ whose backups use only choices in $C(s)$, and $\equiv$ is a local equivalence preserving labels, immediate costs, and successor equivalence classes for every choice used by $C$.  The extension is exact when all these clauses are interpreted on the given finite structure, without projection error, sampling error, or external feasibility predicates.
\end{definition}

\begin{theorem}[Canonical completion of the admissible fixed point]
\label{thm:canonical-completion}
For every exact value-certified extension $(D,C,\mathcal B,\equiv)$ of the one-step semantics, the inclusion $D\hookrightarrow S$ factors through the admissible kernel $K_G$.  More precisely,
\[
D\subseteq K_G,\qquad C(s)\subseteq \CAct_G(s)\quad(s\in D).
\]
Consequently every exact value-certified extension is a restriction of the certified-choice domain generated by $K_G$.  Among exact extensions that use all locally available certified witnesses, the transformer $\mathcal V_G$ is the unique discounted value transformer determined by the displayed costs and successor sets, and the value-refined quotient of Theorem~\ref{thm:core-quotient} is the coarsest quotient preserving the resulting formula-kernel-value-certificate semantics.
\end{theorem}

\begin{proof}
The defining closure condition of an exact extension says that $D\subseteq G$ and that every $s\in D$ has some choice in $C(s)$ whose successor set is contained in $D$.  Hence $D\subseteq G\cap\Pre_{\exists\Box}(D)=F_G(D)$, so $D$ is a post-fixed point of the monotone admissibility transformer.  By greatestness of $K_G=\nu X.F_G(X)$, $D\subseteq K_G$.  If $u\in C(s)$, then $T(s,u)\subseteq D\subseteq K_G$, so $u\in\CAct_G(s)$.  This proves the factorization through the certified-choice domain.  If the extension uses all local witnesses, then its backup set at each $s\in K_G$ is exactly $\CAct_G(s)$; the displayed costs and successor sets therefore force the same Bellman expression as $\mathcal V_G$, and contraction gives uniqueness of the fixed value.  The quotient statement is exactly the coarsest-preservation clause of Theorem~\ref{thm:core-quotient}, with residual preservation added by Theorem~\ref{thm:core-certificate} because the residual is computed from the same transformer.
\end{proof}

\section{Branching choice-cell semantics}
\label{sec:choice-cell-semantics}

The transition presentation is compact, but it hides where choice is located.  In a branching presentation, choice is not an instruction attached from outside the model; it is a partition of histories available at a moment.  The semantics below uses only moments, histories, cells, loci, outcomes, labels, and costs.  The word ``locus'' marks a choice-bearing site in the model and has no extra ontology beyond the partition it indexes.

\begin{definition}[Finite branching choice-cell frame]
A finite branching choice-cell frame is a tuple
\[
\mathfrak C=(M,\prec,H,\Xi,L,\kappa),
\]
where $M$ is a finite set of moments, $\prec$ is an acyclic immediate-successor relation, $H$ is a nonempty set of maximal $\prec$-paths, and $H_m=\{h\in H:m\in h\}$ is nonempty for every $m\in M$.  For each moment $m$, $\Xi(m)$ is a finite nonempty partition of $H_m$ into nonempty choice cells.  The labelling map is $L:M\to2^{AP}$, and $\kappa(m,C)\in[0,\bar c]$ is the value cost of cell $C\in\Xi(m)$.  The one-step outcome of a cell is
\[
\Out(m,C)=\{m'\in M:m\prec m'\text{ and some }h\in C\text{ contains }m'\text{ immediately after }m\}.
\]
Cells with empty one-step outcome are terminal and are omitted from the admissible-continuation construction unless explicitly marked absorbing.
\end{definition}

\begin{definition}[Locus-indexed choice-cell frame]
A locus-indexed choice-cell frame refines the preceding object by replacing $\Xi(m)$ with a finite nonempty set $\Loc(m)$ and, for every $\lambda\in\Loc(m)$, a finite nonempty partition $\Xi_\lambda(m)$ of $H_m$.  A resolved cell is a pair $(\lambda,C)$ with $\lambda\in\Loc(m)$ and $C\in\Xi_\lambda(m)$, and
\[
\Out(m,\lambda,C)=\{m'\in M:m\prec m'\text{ and some }h\in C\text{ contains }m'\text{ immediately after }m\}.
\]
Its flattened cell family is the disjoint union
\[
\Xi^\flat(m)=\{(\lambda,C):\lambda\in\Loc(m),\ C\in\Xi_\lambda(m)\},
\qquad
\kappa^\flat(m,\lambda,C)=\kappa(m,\lambda,C).
\]
A one-locus frame is obtained when $\Loc(m)$ is a singleton for every $m$.
\end{definition}

\begin{proposition}[Flattening of choice loci]
\label{prop:locus-flattening}
For every locus-indexed frame, flattening $\Xi_\lambda(m)$ to $\Xi^\flat(m)$ preserves the one-step modality, the admissible-continuation kernel, certified resolved cells, value refinement, and residual certificates.
\end{proposition}

\begin{proof}
The one-step modality inspects only the existence of a resolved cell whose one-step outcomes are contained in the tested set.  Flattening neither changes the set of resolved cells nor their outcomes or costs; it only forgets the name of the indexing site after the pair $(\lambda,C)$ has been formed.  Therefore the predecessor operator is identical before and after flattening.  The descending fixed-point approximants are generated by the same operator, so the kernels coincide.  Certified resolved cells are the same witnesses to the fixed-point equality.  Since the value transformer and residual norm use the same costs, outcomes, and certified cell family, they are unchanged.
\end{proof}

\begin{definition}[Cellular admissible predecessor]
For $X\subseteq M$, define
\[
\Pre_{\Cell}(X)=\{m\in M: \exists C\in\Xi(m)\; \Out(m,C)\neq\varnothing\text{ and }\Out(m,C)\subseteq X\}.
\]
For a target $G\subseteq M$, the cellular kernel and certified cell set are
\[
K^{\Cell}_G=\nu X.\bigl(G\cap\Pre_{\Cell}(X)\bigr),
\qquad
\Cell_G(m)=\{C\in\Xi(m):\Out(m,C)\subseteq K^{\Cell}_G\}.
\]
The value-refinement transformer on $K^{\Cell}_G$ is
\[
(\mathcal V^{\Cell}_GJ)(m)=
\min_{C\in\Cell_G(m)}\left(\kappa(m,C)+\gamma\max_{m'\in\Out(m,C)}J(m')\right).
\]
\end{definition}

\begin{proposition}[Cell semantics is the same predecessor in internal form]
\label{prop:cell-to-transition}
Every finite branching choice-cell frame induces a modal choice structure
\[
S=M,\qquad A(m)=\Xi(m),\qquad T(m,C)=\Out(m,C),\qquad c(m,C)=\kappa(m,C),
\]
after deleting terminal cells or replacing them by marked absorbing successors.  Under this translation
\[
\Pre_{\Cell}(X)=\Pre_{\exists\Box}(X),
\qquad
K^{\Cell}_G=K_G,
\qquad
\Cell_G(m)=\CAct_G(m),
\]
and $\mathcal V_G^{\Cell}=\mathcal V_G$.
\end{proposition}

\begin{proof}
The equality of predecessors is immediate from the definitions: a moment has a cell whose one-step outcomes all lie in $X$ exactly when the induced modal choice structure has an available choice whose transition set lies in $X$.  The greatest fixed points are therefore the greatest fixed points of the same monotone operator on the same finite powerset lattice.  Certified cells and certified choices are the same witnesses to the fixed-point equality, and the value backups minimize over the same finite family with the same costs and successor sets.
\end{proof}

\begin{proposition}[Transition structures unravel into cell frames]
\label{prop:transition-to-cell}
Every finite modal choice structure $\mathfrak M=(S,U,A,T,L,c)$ has a one-step branching-cell unraveling whose projected predecessor agrees with $\Pre_{\exists\Box}$.  At a finite path $\rho=s_0u_0s_1\cdots s_n$ ending in $s_n$, let $H_\rho$ be the set of infinite continuations compatible with $T$, and let $\Xi(\rho)$ partition $H_\rho$ by the next chosen symbol $u\in A(s_n)$.  The cell for $u$ contains exactly those continuations whose next symbol is $u$, and its projected outcomes are the extensions $\rho u s'$ with $s'\in T(s_n,u)$.  Projection to the last state sends the cellular predecessor back to $\Pre_{\exists\Box}$.
\end{proposition}

\begin{proof}
For a set $X\subseteq S$, lift it to paths by $\widehat X=\{\rho: \mathrm{last}(\rho)\in X\}$.  A path $\rho$ ending in $s$ belongs to $\Pre_{\Cell}(\widehat X)$ iff some cell determined by $u\in A(s)$ has all projected one-step extensions ending in $X$.  This is equivalent to $T(s,u)\subseteq X$, hence to $s\in\Pre_{\exists\Box}(X)$.  Induction on the Kleene approximants then gives equality of projected kernels.  The same argument applies to certified cells and value backups because costs and successor sets are copied from the original structure.
\end{proof}

\begin{theorem}[Internal choice representation theorem]
\label{thm:internal-choice-representation}
For every finite target $G$ and every finite modal choice structure, the layered semantics computed in transition form and the layered semantics computed on its branching-cell presentation agree after projection: satisfaction of base formulas, membership in the admissible-continuation kernel, certified choices/cells, value-refinement fixed points, greedy selector sets, quotient preservation, and residual certificates are invariant under the translation.
\end{theorem}

\begin{proof}
Base satisfaction is preserved because labels are copied by projection.  Proposition~\ref{prop:transition-to-cell} gives equality of the predecessor on every lifted set, so the kernel approximants agree by induction and the greatest fixed points agree by finite stabilization.  Certified witnesses are exactly the cells corresponding to the choices whose successor sets stay inside the kernel.  The Bellman/value transformer is the same functional after projection, hence has the same unique fixed point by contraction.  Value-refined bisimulation only inspects labels, costs, and successor classes; these are preserved by the construction.  Finally, the residual certificate is the norm of $J-\mathcal V_GJ$, and $\mathcal V_G$ is unchanged under projection.
\end{proof}

\begin{theorem}[Conservative retraction of the transition presentation]
\label{thm:conservative-retraction}
Let $\mathsf U$ send a transition-style modal choice structure to the one-step branching-cell unraveling of Proposition~\ref{prop:transition-to-cell}, and let $\mathsf P$ flatten a branching-cell frame by Proposition~\ref{prop:cell-to-transition}.  For every finite target $G$, the composite $\mathsf P\mathsf U$ preserves the predecessor operator, the Kleene approximants of $K_G$, the fixed point $K_G$, certified choices, value-refinement fixed points, value-refined quotients, and residual certificates exactly.  Conversely, for every branching-cell frame, $\mathsf U\mathsf P$ is a refinement whose projection preserves the same data.
\end{theorem}

\begin{proof}
The first statement follows by composing Propositions~\ref{prop:transition-to-cell} and~\ref{prop:cell-to-transition}: $\mathsf U$ replaces a choice by a cell of histories and $\mathsf P$ reads that cell back as a transition choice with the same outcome set and cost.  Thus the one-step predecessor is unchanged on every subset, and the whole descending Kleene chain is unchanged.  Certified choices are defined solely by membership of their outcome sets in the terminal fixed point, so they are unchanged as well.  The value transformer, quotient transformer, and residual map are functions of labels, costs, certified choices, and successor class sets; all of these are preserved by the composite.  For the converse direction, $\mathsf U\mathsf P$ may split histories more finely than the original cell frame, but projection identifies the split cells with their original outcome sets and costs.  All semantic objects named in the theorem are therefore projection-invariant.
\end{proof}

\begin{remark}[Semantic effect of cells]
The cell presentation does not add a second mathematics.  It relocates choice into the model.  The formula $[\exists]\Box\varphi$ says that the present moment has at least one cell whose immediate outcomes all satisfy $\varphi$.  This is stronger than ordinary possibility and weaker than historical necessity.  It is the exact modal position used later: enough internal structure to speak about certified choice, value order, and quotient preservation, without importing an external narrative.
\end{remark}

\begin{example}[Two loci at one moment]
Let $m$ be a moment with four histories $h_{00},h_{01},h_{10},h_{11}$, where the first bit records the successor reached by one locus and the second bit records a later refinement.  Let $\lambda$ partition $H_m$ into
\[
C_0=\{h_{00},h_{01}\},\qquad C_1=\{h_{10},h_{11}\},
\]
and let $\mu$ partition $H_m$ into
\[
D_0=\{h_{00},h_{10}\},\qquad D_1=\{h_{01},h_{11}\}.
\]
If a target $X$ contains exactly the successors occurring along $h_{00}$ and $h_{01}$, then $\lambda$ has a certified cell $C_0$ while $\mu$ has none.  If $X$ contains the successors occurring along $h_{00}$ and $h_{10}$, the situation is reversed.  The choice-bearing site is therefore not a decorative index: changing the locus changes which existential-universal modal step is certified, even though the set of histories is the same.
\end{example}

\section{Worked semantic discipline: unfolding, ranks, and non-commutation}
\label{sec:worked-discipline}

The preceding section gives the compressed theorem layer.  The present section gives the finite semantic mechanics that force the order
\[
\text{satisfaction}\to\text{admissible continuation}\to\text{value refinement}\to\text{certificate}.
\]
Each item below is a small proof obligation.  Removing the order makes one of the finite witnesses fail.

\subsection{Finite unfolding and survival rank}

Define the descending approximants
\[
D^0_G=S,
\qquad
D^{n+1}_G=G\cap\Pre_{\exists\Box}(D^n_G).
\]
Thus $D^1_G=G$ under the standing nonemptiness of choices and successors, while $D^{n+1}_G$ contains the states from which membership in $G$ can be maintained for one more certified step than $D^n_G$.

\begin{definition}[Depth-$n$ certified tree]
A depth-$n$ certified tree rooted at $s$ is a finite tree of height $n$ whose nodes are labelled by states, whose root is labelled by $s$, and such that every non-leaf node labelled by $x$ is also labelled with a choice $u_x\in A(x)$ and has exactly the successor labels in $T(x,u_x)$.  It is $G$-admissible when every node at depth $<n$ is labelled by a state in $G$.
\end{definition}

\begin{lemma}[Approximants are finite proof depths]
\label{lem:depth-tree}
For every $n\geq0$ and every state $s$,
\[
s\in D^n_G
\quad\Longleftrightarrow\quad
\text{there is a $G$-admissible depth-$n$ certified tree rooted at $s$}.
\]
\end{lemma}

\begin{proof}
For $n=0$ the tree consists only of its root and the condition is empty, matching $D^0_G=S$.  Assume the statement for $n$.  A state $s$ belongs to $D^{n+1}_G$ precisely when $s\in G$ and there exists $u\in A(s)$ with $T(s,u)\subseteq D^n_G$.  By the induction hypothesis, each successor $s'\in T(s,u)$ has a depth-$n$ $G$-admissible certified tree.  Placing these trees below a root labelled $s$ and choice-labelled by $u$ gives a depth-$(n+1)$ tree.  Conversely, the first choice at the root of any depth-$(n+1)$ certified tree has all its successors as roots of depth-$n$ certified subtrees, so the induction hypothesis places all those successors in $D^n_G$, and the root belongs to $G\cap\Pre_{\exists\Box}(D^n_G)$.
\end{proof}

The finite tree view is the proof-theoretic reading of the greatest fixed point.  A state is in the kernel exactly when every finite demand for a certified continuation tree can be met; finiteness of $S$ turns this into the stabilized greatest fixed point.

\begin{definition}[Elimination rank]
For $s\in S$, define
\[
\rank_G(s)=\inf\{n\geq0:s\notin D^n_G\},
\]
with $\rank_G(s)=\infty$ when $s\in D^n_G$ for all $n$.  For $s\in G$ and $u\in A(s)$ define the successor rank margin
\[
\rank_G(s,u)=\min_{s'\in T(s,u)}\rank_G(s').
\]
\end{definition}

\begin{proposition}[Rank is the finite trace of the fixed point]
\label{prop:rank-trace}
A state belongs to $K_G$ if and only if $\rank_G(s)=\infty$.  Moreover, for $s\in G$ and $n\geq0$,
\[
s\in D^{n+1}_G
\quad\Longleftrightarrow\quad
\exists u\in A(s)\;\rank_G(s,u)>n.
\]
Consequently, a certified choice in the limit is exactly a choice whose successor rank is infinite.
\end{proposition}

\begin{proof}
The first statement follows because the descending sequence stabilizes at $K_G$: surviving all approximants is the same as belonging to the intersection, and in a finite descending chain the intersection is the stable value.  For the second statement, $s\in D^{n+1}_G$ iff $s\in G$ and some $u$ has $T(s,u)\subseteq D^n_G$.  The latter inclusion is equivalent to every successor having rank greater than $n$, i.e. to the minimum successor rank being greater than $n$.  Letting $n$ range over all finite depths gives the limit statement.
\end{proof}

\subsection{Finite-horizon value as a second, later unfolding}

After $K_G$ and $\CAct_G$ have been fixed, define $J_0\equiv0$ and $J_{n+1}=\mathcal V_GJ_n$.  This is not another way to compute the kernel.  It is a value unfolding over the already certified choice graph.

\begin{lemma}[Finite-horizon value interpretation]
\label{lem:finite-horizon-value}
For each $n$, $J_n(s)$ is the least worst-compatible $n$-step discounted cost obtainable from $s$ when every choice at every step is restricted to $\CAct_G$.  More explicitly,
\[
J_n(s)=\inf_\sigma\sup_{s_0s_1\cdots s_n}
\sum_{k=0}^{n-1}\gamma^k c(s_k,\sigma(s_k)),
\]
where $s_0=s$, $\sigma(x)\in\CAct_G(x)$, and $s_{k+1}\in T(s_k,\sigma(s_k))$.
\end{lemma}

\begin{proof}
For $n=0$ both sides are zero.  For the step from $n$ to $n+1$, the first certified choice $u$ contributes $c(s,u)$, the compatible successor is chosen universally from $T(s,u)$, and the remaining $n$-step tail has value $J_n$.  Taking the minimum over certified first choices gives exactly $\mathcal V_GJ_n$.
\end{proof}

\begin{corollary}[A posteriori and a priori reading of residuals]
\label{cor:residual-reading}
For the finite-horizon iterates,
\[
\|J_n-J_G^\star\|_\infty\leq \frac{\|J_{n+1}-J_n\|_\infty}{1-\gamma}
\quad\text{and}\quad
\|J_n-J_G^\star\|_\infty\leq \gamma^n\frac{\bar c}{1-\gamma}.
\]
The first bound is a certificate carried by the computation; the second is a universal clock bound known before the computation starts.
\end{corollary}

\begin{proof}
The first inequality is Theorem~\ref{thm:core-certificate} applied to $J_n$, since $\Res_G(J_n)=\|J_n-J_{n+1}\|_\infty$.  For the second, $\|J_G^\star\|_\infty\leq \bar c/(1-\gamma)$ and contraction gives $\|J_n-J_G^\star\|_\infty=\|\mathcal V_G^n0-\mathcal V_G^nJ_G^\star\|_\infty\leq\gamma^n\|J_G^\star\|_\infty$.
\end{proof}

\subsection{The layers do not commute}

The layer order is not a stylistic convention.  Each attempted swap below has a two- or three-state countermodel.

\begin{proposition}[Value before admissibility is unsound]
\label{prop:value-before-kernel}
There is a finite modal choice structure in which minimizing value over all choices selects a move that is not in any admissible continuation kernel, while minimizing after the kernel selects a different move.
\end{proposition}

\begin{proof}
Let $S=\{s,g,b\}$, let $G=\{s,g\}$, let $g$ and $b$ have zero-cost self-loops, and let $s$ have two choices
\[
T(s,a)=\{g\},\quad c(s,a)=1,
\qquad
T(s,d)=\{b\},\quad c(s,d)=0.
\]
The unrestricted one-step value prefers $d$, since it has lower immediate cost.  But $b\notin G$, so $d$ cannot witness $s\in G\cap\Pre_{\exists\Box}(G)$ and cannot belong to $\CAct_G(s)$.  The kernel is $K_G=\{s,g\}$ and the only certified choice at $s$ is $a$.  Thus value-first evaluation changes the admissible semantics.
\end{proof}

\begin{proposition}[Cost-blind quotienting is not value preserving]
\label{prop:costblind-quotient-breaks}
Label equivalence and ordinary modal bisimulation do not suffice for the value-refined semantics.
\end{proposition}

\begin{proof}
Let $S=\{p,q\}$ with equal labels and one self-loop choice at each state.  Let $c(p,a)=0$ and $c(q,a)=1$.  Ordinary modal bisimulation identifies $p$ and $q$: both have the same label and the same one-step successor pattern.  With discount $\gamma\in(0,1)$, however,
\[
J^\star(p)=0,
\qquad
J^\star(q)=1+\gamma J^\star(q)=\frac1{1-\gamma}.
\]
A quotient that merges them has no well-defined exact value preserving both equations.  The cost clause in value-refined modal bisimulation is therefore forced by the semantics, not by presentation.
\end{proof}

\begin{proposition}[Residual smallness before certification is meaningless]
\label{prop:residual-before-certification}
A residual certificate for a value equation does not imply admissible continuation unless the equation is solved over the certified kernel.
\end{proposition}

\begin{proof}
Use the previous three-state structure and restrict attention to the bad state $b$.  The self-loop value equation at $b$ with zero cost has the exact solution $J(b)=0$ and residual zero.  Nevertheless $b\notin G$ and $b\notin K_G$.  Hence residual zero certifies only value consistency for the equation being solved; it does not certify membership in the modal fixed point.  The residual layer must come after the admissible-continuation layer.
\end{proof}

\begin{theorem}[Non-commutation of the semantic architecture]
\label{thm:noncommutation}
The four layers cannot in general be permuted without changing the induced semantics.  In particular,
\[
\Val\circ\Viab\neq\Viab\circ\Val,
\qquad
\text{quotient}\circ\Val\neq\Val\circ\text{ordinary quotient},
\qquad
\Cert\circ\Viab\neq\Viab\circ\Cert
\]
when the informal compositions are instantiated by the finite constructions above.
\end{theorem}

\begin{proof}
The first inequality is witnessed by Proposition~\ref{prop:value-before-kernel}; the second by Proposition~\ref{prop:costblind-quotient-breaks}; the third by Proposition~\ref{prop:residual-before-certification}.  The equations are schematic, but each proposition gives an exact finite structure and exact numerical values.  Therefore the fixed order used here is mathematically necessary for the intended semantics.
\end{proof}

\subsection{Clause-by-clause necessity of value-refined bisimulation}

The bisimulation used in the quotient theorem has three local clauses: label equality, cost equality, and equality of successor class sets for matched choices.  The following proposition records why none of them is cosmetic.

\begin{proposition}[Each local clause has a finite separating witness]
\label{prop:clause-necessity}
For the combined formula-kernel-value semantics, removing any one of the following requirements can identify states that must remain distinct:
\begin{enumerate}[label=(\roman*)]
\item label equality;
\item equality of matched choice costs;
\item equality of matched successor class sets.
\end{enumerate}
\end{proposition}

\begin{proof}
If label equality is removed, two one-state self-loop structures with different atomic labels are identified although the atomic formula separates them.  If cost equality is removed, the two-state self-loop structure of Proposition~\ref{prop:costblind-quotient-breaks} is identified although the value equations separate the states.  If equality of successor class sets is weakened to one-sided successor matching, use states $s,t$ with one choice each, $T(s,u)=\{x\}$ and $T(t,v)=\{x,y\}$, where $x\models G$ and $y\not\models G$.  The successor $x$ of $s$ is matched by a successor of $t$, but $s\models[\exists]\Box G$ while $t\not\models[\exists]\Box G$.  Hence the universal successor component of the choice modality forces equality of successor classes, not just existence of some matching successor.
\end{proof}

\section{A minimal reading path}
\label{sec:minimal-reading-path}

The core construction is independent of the boundary fragments.  It begins with the axiomatic core, which defines modal choice structures, the admissible kernel, certified choices, value refinement, and value-refined bisimulation.  The choice-cell section shows that the transition presentation is a conservative retraction of the model-internal cell presentation.  The non-commutation examples explain why the order of the layers cannot be changed.  The quotient and certificate theorems provide the preservation results needed by the main theorem.

This separation matters because several surrounding topics have larger languages.  Full ATL, strategy logics, epistemic temporal logics, probabilistic model checking, and undiscounted games all contain neighboring mechanisms.  The construction uses only the fragment needed for the displayed semantics: one-step cell ability, a greatest fixed point, discounted value over certified cells, cost-aware quotienting, and residual certificates.  The finite laboratory below is included to make that fragment testable without importing extra machinery.

\section{A detailed finite laboratory}
\label{sec:finite-laboratory}

This section gives a compact system whose computations exhibit the fixed-point chain, the value iteration, the residual certificate, and the quotient constraint in one place.  Let
\[
S=\{r,x,y,z,w,\bot\},
\qquad
G=\{r,x,y,z,w\},
\qquad
0<\gamma<1,
\]
and let all states satisfy the same propositional label except $\bot$, which violates the target label.  The choices are
\[
T(r,a)=\{x,y\},\quad T(r,b)=\{\bot\},
\]
\[
T(x,a)=\{z\},\quad T(x,b)=\{\bot\},
\qquad
T(y,a)=\{z,w\},\quad T(y,b)=\{y\},
\]
\[
T(z,a)=\{z\},\quad T(w,a)=\{w\},\quad T(\bot,a)=\{\bot\}.
\]
The costs are
\[
\begin{aligned}
c(r,a)&=4, & c(r,b)&=0, & c(x,a)&=2, & c(x,b)&=0,\\
c(y,a)&=3, & c(y,b)&=5, & c(z,a)&=1, & c(w,a)&=0,\\
c(\bot,a)&=0.&&&&&
\end{aligned}
\]
The commas in this display are semantic punctuation: they are not separate data structures.

The descending chain starts with $D_G^0=S$.  Since $\bot\notin G$, the first nontrivial stage is
\[
D_G^1=G=\{r,x,y,z,w\}.
\]
At the next stage, $r$ survives because of $a$, $x$ survives because of $a$, $y$ survives because both $a$ and $b$ remain inside $G$, and $z,w$ survive by self-loops.  Hence $D_G^2=G$.  The chain has stabilized, so $K_G=G$.  The certified choices are
\[
\CAct_G(r)=\{a\},\quad
\CAct_G(x)=\{a\},\quad
\CAct_G(y)=\{a,b\},\quad
\CAct_G(z)=\{a\},\quad
\CAct_G(w)=\{a\}.
\]
The choices $b$ at $r$ and $x$ are cheap but uncertified; they are removed before value is allowed to speak.

The terminal self-loop equations give
\[
J^\star(z)=1+\gamma J^\star(z)=\frac1{1-\gamma},
\qquad
J^\star(w)=0.
\]
At $x$ there is only the certified choice $a$, so
\[
J^\star(x)=2+\gamma J^\star(z)=2+\frac{\gamma}{1-\gamma}.
\]
At $y$ the two certified choices compete:
\[
Q(y,a)=3+\gamma\max(J^\star(z),J^\star(w))=3+\frac{\gamma}{1-\gamma},
\]
\[
Q(y,b)=5+\gamma J^\star(y).
\]
If $a$ is selected, then $J^\star(y)=3+\gamma/(1-\gamma)$.  If $b$ is selected, then $J^\star(y)=5/(1-\gamma)$.  Since
\[
3+\frac{\gamma}{1-\gamma}\leq\frac5{1-\gamma}
\quad\Longleftrightarrow\quad
3-2\gamma\leq5,
\]
choice $a$ is always no worse for $\gamma\in(0,1)$.  Therefore
\[
J^\star(y)=3+\frac{\gamma}{1-\gamma}.
\]
Finally,
\[
J^\star(r)=4+\gamma\max(J^\star(x),J^\star(y))
=4+\gamma\left(3+\frac{\gamma}{1-\gamma}\right),
\]
because $J^\star(y)-J^\star(x)=1$.  This calculation shows the semantic order in arithmetic form: the globally cheapest raw move at $r$ has cost zero, but it points to $\bot$ and never reaches the value comparison.

A residual certificate can be read off from the iterates.  With $J_0\equiv0$, the first two value iterates on the state order $(r,x,y,z,w)$ are
\[
J_1=(4,2,3,1,0),
\]
\[
J_2=\bigl(4+3\gamma,
2+\gamma,
3+\gamma,
1+\gamma,
0\bigr).
\]
Thus
\[
\|J_2-J_1\|_\infty=3\gamma,
\qquad
\|J_1-J^\star\|_\infty\leq\frac{3\gamma}{1-\gamma}.
\]
For $\gamma=1/2$, the exact values become
\[
J^\star=(6,3,4,2,0),
\]
and the residual bound after the second iterate gives
\[
\|J_2-J^\star\|_\infty\leq \frac{\|J_3-J_2\|_\infty}{1-\gamma}.
\]
The certificate is numerical but not merely numerical: it is meaningful only because $J_2$ was computed on the certified set $K_G$.

The same laboratory also forces the quotient clause.  If a fresh state $w'$ is added with the same label as $w$, the same self-loop, and zero cost, then $w$ and $w'$ merge under value-refined bisimulation.  If instead $c(w',a)=1$, ordinary modal bisimulation still sees the same labelled self-loop pattern, but the values become $0$ and $1/(1-\gamma)$; the value-refined quotient refuses the merge.  Thus the quotient theorem is not an abstract preservation slogan.  It is the exact compression allowed by the value equations.

\section{Classical foundations for value-refined modal semantics}

Kripke semantics supplies the relational interpretation of modal formulas \cite{kripke1963}.  Tarski's fixed-point theorem supplies the complete-lattice existence principle for the greatest admissible-continuation region \cite{tarski1955}.  Kozen's modal $\mu$-calculus supplies the fixed-point language in which such regions are definable \cite{kozen1983}.  Park and Hennessy--Milner supply the bisimulation-invariance template refined here by choices, costs, and residual certificates \cite{park1981,hennessy1985}.

Temporal and strategic logics give the closest logical neighbours.  Pnueli, Clarke--Emerson, Vardi--Wolper, and Emerson--Jutla provide the temporal, automata-theoretic, and game-theoretic basis of semantic construction \cite{pnueli1977,clarke1981,vardi1986,emerson1991}.  ATL and coalition logic give the exact reading of $[\exists]\Box$ as a one-step strategic ability operator \cite{alur2002,pauly2002}.  ATL$^*$ and strategy logic increase temporal and strategic expressiveness; the present construction stays in the invariance fragment and studies the value layer induced after the fixed-point region has been computed \cite{alur2002,chatterjee2010}.  Imperfect-information games motivate the information-state construction used for observation structures \cite{reif1984,cdhr2007}.

The quantitative layer is placed beside discounted game semantics and quantitative verification.  Bellman, Blackwell, Shapley, Puterman, and Bertsekas give the discounted dynamic-programming and game-value machinery \cite{bellman1957,blackwell1965,shapley1953,puterman1994,bertsekas2012}.  Iyengar gives a state-choice rectangular value-recursion setting that matches universal successor resolution \cite{iyengar2005}.  Probabilistic model checking and stochastic-game tools treat probability and reward objectives explicitly; here the baseline semantics is universal over compatible successors unless a probabilistic refinement is added \cite{baier2008,kwiatkowska2011,kwiatkowska2020}.

Permissive strategies and finite game semantics provide algorithmic context \cite{bernet2002,gradel2002}.  Their role here is semantic: they explain why the maximal fixed-point region should be computed before value ranking, and why a quotient must preserve formulas, certified cells, successor classes, and costs rather than labels alone.

\section{Standing hypotheses and boundary countermodels}

The state-level theorem is stated under hypotheses that match the state-explicit construction: finite branching, explicit finite-array values, discounting, cost-aware quotienting, and positive switching charge.  Partial observation is handled separately by the information-state theorem.  These assumptions are not cosmetic.  Removing any one of them changes the semantic problem.  Unrestricted search loses minimizing selectors \cite{wolpert1997}; arbitrary approximation can destroy contraction-based guarantees \cite{baird1995}; distributed semantic construction and incomplete-information games need information-state constructions rather than a single state-based selector \cite{pnueli1990,reif1984,cdhr2007}; finitary branching is the setting in which Hennessy--Milner equivalence and finite fixed-point iteration align \cite{hennessy1985}; and free switching admits infinite objective oscillation unless a positive switching charge or finite descent measure is present.

\begin{definition}[Boundary hypothesis]
A boundary hypothesis is a structural assumption whose removal changes at least one conclusion of the finite semantic construction theorem.
\end{definition}

\begin{theorem}[Boundary countermodels]
Each of the following removals admits a finite or elementary countermodel: compactness or finiteness of certified choices, full observation, discounting, positive switching cost, exact value representation, and cost-aware bisimulation.
\end{theorem}

\begin{proof}
Separate witnesses are given.

For finiteness or compactness of the choice set,  Let $S=\{s\}$, let $G=S$, let $U=(0,1)$, let $A(s)=U$, let $T(s,u)=\{s\}$, and let $c(s,u)=u$.  The infimum of the one-step cost is $0$, but no choice attains it.  Hence a minimizing selector need not exist.

For full observation,  Let $S=\{p,q,r\}$, let $G=\{p,q\}$, and let the selector observe $p$ and $q$ as the same observation.  At $p$, choice $a$ returns to $p$ and choice $b$ goes to $r$.  At $q$, choice $b$ returns to $q$ and choice $a$ goes to $r$.  State-based preserves $G$ by selecting $a$ at $p$ and $b$ at $q$.  No observation-based selector can distinguish the two states, so any single selected choice fails at one of them.  This is exactly why full observation is the standing hypothesis of the quantitative development; the qualitative content of the partial-observation case is recovered separately by the information-state construction of the observation-layer section, where the right object is a selector on information states rather than on states, and the present countermodel is the information state $\{p,q\}$ at which no single choice is admissible for both members.

For discounting,  Let $S=\{s\}$, let $U=\{u\}$, let $T(s,u)=\{s\}$, let $c(s,u)=0$, and put $\gamma=1$.  The Bellman equation becomes $J(s)=J(s)$.  Every real value is a fixed point, so uniqueness fails.

For positive switching cost,  Let two modes have identical state value at every state and set the switching charge to $0$.  The mode sequence $1,2,1,2,\ldots$ has infinitely many switches and zero switching cost.  Finite switching cannot follow from bounded total cost alone.

For explicit finite-array values, replace the exact value space by an arbitrary projection.  On $\mathbb R^2$ with the sup norm, define $P(x,y)=(x+y,0)$.  Then $P^2=P$ and $\|P\|_\infty=2$.  If an exact Bellman map is the contraction $Bz=\gamma z$ with $\gamma=3/4$, then $PB$ has operator norm $3/2$ and is not a contraction.  Thus an approximation layer can destroy the fixed-point argument even when the exact operator contracts.

For cost-aware bisimulation, quotient only by labels and not by choice-refined costs.  Let two states $s$ and $t$ have the same label and the same self-loop transition, but let their unique certified choices have costs $0$ and $1$.  The states satisfy the same base modal formulas, but their discounted values are $0$ and $1/(1-\gamma)$.  Formula preservation alone does not preserve value.

Each witness removes exactly the hypothesis named above and breaks the corresponding conclusion.
\end{proof}

\begin{theorem}[Closure under the standing hypotheses]
Under finite state and choice sets, state-explicit presentation, explicit finite-array value functions, discount factor $\gamma\in(0,1)$, universal finite successor sets, value-refined modal bisimulation, and positive switching cost, the boundary countermodels above are excluded by the hypotheses used in the main semantic construction theorem.
\end{theorem}

\begin{proof}
Finite nonempty choice sets restore minimizers by the finite selector lemma.  State-explicit presentation makes a state-based selector a well-defined map from the actual state to a certified choice.  Explicit finite-array representation keeps value iteration inside $\ell_\infty(V)$ and avoids inserting a projection whose norm may exceed one.  Discounting gives the Bellman--Kripke map contraction factor $\gamma$.  Value-refined modal bisimulation preserves labels, transfer structure, and immediate costs, so quotienting does not merge states with different values.  Positive switching cost turns every mode change into a nonzero charge, so a bounded total switching budget gives the switch-count bound.  These are exactly the structural conditions invoked by the later fixed-point, game, quotient, and switching proofs.
\end{proof}

\section{Finite semantic lemmas below the fixed-point and Bellman layers}

The main structures rest on smaller facts about finite sets, minima, maxima, and relations.  These lemmas make the later fixed-point and Bellman arguments explicit instead of relying on implicit black-box steps.

\begin{lemma}[Finite selector lemma]
Let $X$ be a finite nonempty set and let $f:X\to\mathbb R$.  Then $\argmin_{x\in X}f(x)$ is nonempty.
\end{lemma}

\begin{proof}
Since $X$ is finite, the set $f(X)$ is a finite nonempty subset of $\mathbb R$.  Every finite nonempty subset of $\mathbb R$ has a least element.  Choose $x^\star\in X$ with $f(x^\star)$ equal to that least element.  Then $x^\star\in\argmin_{x\in X}f(x)$.
\end{proof}

\begin{lemma}[Maximum perturbation lemma]
Let $Y$ be a finite nonempty set and let $f,g:Y\to\mathbb R$.  Then
\[
\left|\max_{y\in Y}f(y)-\max_{y\in Y}g(y)\right|
\leq
\max_{y\in Y}|f(y)-g(y)|.
\]
\end{lemma}

\begin{proof}
Let $\delta=\max_{y\in Y}|f(y)-g(y)|$.  For every $y\in Y$ we have $f(y)\leq g(y)+\delta$.  Taking maxima gives $\max_Y f\leq \max_Y g+\delta$.  Reversing the roles of $f$ and $g$ gives $\max_Y g\leq \max_Y f+\delta$.  The two inequalities imply the claim.
\end{proof}

\begin{lemma}[Minimum perturbation lemma]
Let $X$ be finite and nonempty.  If $a_x,b_x\in\mathbb R$ for each $x\in X$, then
\[
\left|\min_{x\in X}a_x-\min_{x\in X}b_x\right|
\leq
\max_{x\in X}|a_x-b_x|.
\]
\end{lemma}

\begin{proof}
Let $\epsilon=\max_{x\in X}|a_x-b_x|$.  For every $x$, $a_x\leq b_x+\epsilon$.  Taking minima gives $\min_X a\leq \min_X b+\epsilon$.  Reversing $a$ and $b$ gives $\min_X b\leq \min_X a+\epsilon$.  The desired bound follows.
\end{proof}

\begin{lemma}[Certified-choice restriction lemma]
Let $V=\Viab(G)$ and define
\[
A_V(s)=\{u\in A(s):T(s,u)\subseteq V\}.
\]
For every $s\in V$, the set $A_V(s)$ is nonempty.
\end{lemma}

\begin{proof}
Since $V=G\cap\Pre_{\exists\Box}(V)$, every $s\in V$ belongs to $\Pre_{\exists\Box}(V)$.  By definition of admissible predecessor, there exists $u\in A(s)$ such that $T(s,u)\subseteq V$.  Thus $u\in A_V(s)$.
\end{proof}

\begin{lemma}[One-step admissible closure]
If $s\in V$ and $u\in A_V(s)$, then every successor $s'\in T(s,u)$ belongs to $V$.
\end{lemma}

\begin{proof}
The claim is exactly the defining inclusion $T(s,u)\subseteq V$ for membership in $A_V(s)$.
\end{proof}

\begin{lemma}[Bellman elementary contraction]
Let $J,K:V\to\mathbb R$ and $0<\gamma<1$.  Then
\[
\|\mathcal B J-\mathcal B K\|_\infty
\leq
\gamma\|J-K\|_\infty.
\]
\end{lemma}

\begin{proof}
Fix $s\in V$.  For each $u\in A_V(s)$ define
\[
Q_J(s,u)=c(s,u)+\gamma\max_{s'\in T(s,u)}J(s').
\]
The maximum perturbation lemma gives
\[
|Q_J(s,u)-Q_K(s,u)|\leq \gamma\|J-K\|_\infty.
\]
The minimum perturbation lemma over the finite nonempty set $A_V(s)$ gives
\[
|\min_{u\in A_V(s)}Q_J(s,u)-\min_{u\in A_V(s)}Q_K(s,u)|
\leq
\gamma\|J-K\|_\infty.
\]
Taking the supremum over $s\in V$ proves the result.
\end{proof}

\begin{lemma}[Monotonicity of the Bellman--Kripke map]
Let $J,K:V\to\mathbb R$ with $J\leq K$ pointwise.  Then $\mathcal B_VJ\leq\mathcal B_VK$ pointwise, and the same holds for the selector operator $\mathcal B^\pi$ of any state-based selector $\pi$ valued in $A_V$.
\end{lemma}

\begin{proof}
Fix $s\in V$ and $u\in A_V(s)$.  Since $J\leq K$, every successor value satisfies $J(s')\leq K(s')$, so $\max_{s'\in T(s,u)}J(s')\leq\max_{s'\in T(s,u)}K(s')$ and therefore $c(s,u)+\gamma\max_{T(s,u)}J\leq c(s,u)+\gamma\max_{T(s,u)}K$.  Taking the minimum over the common index set $A_V(s)$ preserves the inequality, which is $(\mathcal B_VJ)(s)\leq(\mathcal B_VK)(s)$.  For $\mathcal B^\pi$ the single choice $\pi(s)$ replaces the minimization and the same monotonicity of the inner maximum gives the claim.
\end{proof}

\begin{lemma}[Greedy domination]
For every state-based selector $\pi$ valued in $A_V$ and every $J:V\to\mathbb R$,
\[
\mathcal B_VJ\leq\mathcal B^\pi J
\]
pointwise, with equality at $s$ if and only if $\pi(s)$ attains the minimum defining $(\mathcal B_VJ)(s)$.
\end{lemma}

\begin{proof}
By definition $(\mathcal B_VJ)(s)=\min_{u\in A_V(s)}[c(s,u)+\gamma\max_{T(s,u)}J]$ and $\pi(s)\in A_V(s)$ is one admissible index, so the minimum is at most the value at $\pi(s)$, which is $(\mathcal B^\pi J)(s)$.  The minimum equals this value exactly when $\pi(s)$ is a minimizer.
\end{proof}

\begin{lemma}[Sub- and supersolution sandwich]
Let $J^\star$ be the fixed point of $\mathcal B_V$.  If $J\leq\mathcal B_VJ$ then $J\leq J^\star$, and if $J\geq\mathcal B_VJ$ then $J\geq J^\star$.
\end{lemma}

\begin{proof}
Suppose $J\leq\mathcal B_VJ$.  Applying the monotone map $\mathcal B_V$ repeatedly preserves the inequality, so $J\leq\mathcal B_VJ\leq\mathcal B_V^2J\leq\cdots\leq\mathcal B_V^nJ$ for all $n$.  Since $\mathcal B_V^nJ\to J^\star$ uniformly by the contraction property, the limit gives $J\leq J^\star$.  The supersolution case is symmetric.
\end{proof}

\begin{lemma}[Class saturation lemma]
Let $\sim$ be an equivalence relation on $S$ such that whenever $s\sim t$, for every $u\in A(s)$ there is $v\in A(t)$ with $\{[x]_\sim:x\in T(s,u)\}=\{[x]_\sim:x\in T(t,v)\}$.  If $X$ is a union of equivalence classes, then $\Pre_{\exists\Box}(X)$ is also a union of equivalence classes.
\end{lemma}

\begin{proof}
Assume $s\sim t$ and $s\in\Pre_{\exists\Box}(X)$.  Then some $u\in A(s)$ satisfies $T(s,u)\subseteq X$.  By hypothesis there is $v\in A(t)$ with $\{[x]_\sim:x\in T(t,v)\}=\{[x]_\sim:x\in T(s,u)\}$.  Because $X$ is a union of $\sim$-classes and $T(s,u)\subseteq X$, every class met by $T(s,u)$ is contained in $X$; by the class equality every class met by $T(t,v)$ is also contained in $X$.  Hence $T(t,v)\subseteq X$, so $t\in\Pre_{\exists\Box}(X)$.
\end{proof}

\begin{lemma}[Switch accounting lemma]
Let every switch have cost at least $\lambda>0$.  If total switching cost is at most $M$, then the number of switches is at most $\lfloor M/\lambda\rfloor$.
\end{lemma}

\begin{proof}
If $N$ switches occur, their cost is at least $N\lambda$.  The bound $N\lambda\leq M$ gives $N\leq M/\lambda$.  Integrality gives the floor bound.
\end{proof}

\section{Axioms: modal choice transition structures}

The construction starts from finite objects familiar from model checking, modal logic, and strategic semantic construction: a state set, a choice alphabet, an admissibility relation, a transition correspondence, a labelling map, an observation interface, and a stage cost.  Finiteness fixes the discrete setting in which fixed-point iteration, quotient refinement, strategy extraction, information state construction, and discounted dynamic programming can be stated without measurability or compactness side conditions.  The word \emph{choice} refers to the existential component of the predecessor operator; the semantic object is a labelled structure equipped with certified choices.

\begin{definition}[Modal choice structure with observations]
A modal choice structure with observations is a tuple
\[
\mathcal T=(S,U,A,T,L,c),
\]
where $S$ and $U$ are finite nonempty sets, $A:S\to 2^U\setminus\{\varnothing\}$ assigns certified choices, $T(s,u)\subseteq S$ is nonempty for $u\in A(s)$, $L:S\to 2^{\mathsf P}$ labels states by propositions, and $c:S\times U\to \mathbb R_{\geq 0}$ is defined on admissible pairs.
\end{definition}

\begin{definition}[Selector path]
A selector path from $s_0$ is a sequence
\[
(s_0,u_0,s_1,u_1,s_2,\ldots)
\]
such that $u_t\in A(s_t)$ and $s_{t+1}\in T(s_t,u_t)$ for all $t\geq 0$.
\end{definition}

A stationary selector under full observation is a function $\pi:S\to U$ satisfying $\pi(s)\in A(s)$ whenever it is used.  A history selector is a function whose input is a finite state-choice history.  The core justification theorem uses stationary selectors because finite invariance games are memoryless at the fixed-point layer.  History selectors return only when bounded deliberation is discussed.

\begin{assumption}[Universal successor convention]\label{asm:adv}
After a choice $u$ is selected at state $s$, every successor in $T(s,u)$ is treated as observationally possible unless a later risk envelope explicitly changes the evaluation rule.
\end{assumption}

The convention matches the box modality: after the selector selects a choice, all successors compatible with that choice must be accepted.  It is also the move order of the discounted successor game of Definition~\ref{def:game}, in which the universal successor is selected by an resolver; that the resolver needs no memory is established in Theorem~\ref{thm:gamevalue}.  A probabilistic variant would replace the universal successor clause by expectations or risk measures and would no longer be the same theorem.

\begin{definition}[One-step game arena]\label{def:arena}
The one-step game arena associated with $\mathcal T$ is the turn-based arena in which the selector first selects $u\in A(s)$ and a universal resolver then selects $s'\in T(s,u)$.  The induced choice-refined post operator is
\[
\Post(s,u)=T(s,u).
\]
\end{definition}

This arena is the finite-game reading of the choice-refined modality used below.  It is the one-step game form on which alternating-time and coalition modalities are interpreted \cite{alur2002,pauly2002}: a selector selects a cell and a successor resolver chooses one compatible continuation.  In the branching-cell presentation of Section~
ef{sec:choice-cell-semantics}, the same clause says that some cell of histories has only successors inside the target set.

\section{Modal syntax and certified choice}

Let $\mathsf P$ be a finite set of atomic propositions.  The base modal language is
\[
\varphi ::= \top\mid p\mid \neg p\mid (\varphi\wedge\psi)\mid (\varphi\vee\psi)\mid \Box\varphi\mid \Diamond\varphi.
\]
This language observes the flattened successor relation.  The certified-choice layer adds the one-step modality $[\exists]\Box$, read as ``there is a certified choice whose every possible successor satisfies the continuation formula.''  For greatest fixed-point claims the positive certified-choice fragment is
\[
\theta ::= \varphi\mid Z\mid (\theta\wedge\theta)\mid(\theta\vee\theta)\mid[\exists]\Box\theta\mid \nu Z.\theta,
\]
where every free occurrence of the bound variable is positive.  This is the fragment needed for the admissible-continuation kernel formula; richer temporal objectives can be compiled through the product construction in Appendix~C.

For a valuation $\eta$ assigning subsets of $S$ to fixed-point variables, write $\llbracket\theta\rrbracket_\eta\subseteq S$ for the usual denotation: atoms are interpreted by $L$, Boolean connectives by set operations, $[\exists]\Box$ by $\Pre_{\exists\Box}$, and
\[
\llbracket\nu Z.\theta\rrbracket_\eta=\nu X.\llbracket\theta\rrbracket_{\eta[Z:=X]}.
\]
Because the fragment is positive in $Z$, every functional $X\mapsto\llbracket\theta\rrbracket_{\eta[Z:=X]}$ is monotone on the finite lattice $2^S$.

\begin{lemma}[Monotonicity of the positive choice-refined fragment]
If $X\subseteq Y$ and the variable $Z$ occurs only positively in $\theta$, then
\[
\llbracket\theta\rrbracket_{\eta[Z:=X]}
\subseteq
\llbracket\theta\rrbracket_{\eta[Z:=Y]}.
\]
Consequently every greatest or least fixed point used in the choice-refined invariance fragment is well defined and is reached by finite Kleene iteration.
\end{lemma}

\begin{proof}
The proof is by structural induction on $\theta$.  Atomic formulas and constants do not depend on $Z$.  The variable case is exactly $X\subseteq Y$.  Conjunction and disjunction preserve inclusion.  For the choice-refined modality, if $s\in\Pre_{\exists\Box}(A)$ and $A\subseteq B$, then the same witnessing choice has $T(s,u)\subseteq A\subseteq B$, so $s\in\Pre_{\exists\Box}(B)$.  Thus $[\exists]\Box$ is monotone.  The fixed-point clauses are then monotone functionals on $2^S$, and the finite powerset lattice has no infinite strictly ascending or descending chain, so Kleene iteration reaches the corresponding fixed point in finitely many rounds.
\end{proof}

\begin{definition}[Plain modal satisfaction]
For $s\in S$, satisfaction is defined by
\[
s\models \top,
\qquad
s\models p \iff p\in L(s),
\]
with Boolean clauses as usual, and
\[
s\models \Box\varphi \iff \forall u\in A(s)\ \forall s'\in T(s,u),\ s'\models\varphi,
\]
\[
s\models \Diamond\varphi \iff \exists u\in A(s)\ \exists s'\in T(s,u),\ s'\models\varphi.
\]
\end{definition}

The ordinary box is a universal statement over choices and successors.  It is too strong for semantic construction because the semantics requires one certified choice to be selected.  The next operator fixes the quantifier order.

\begin{definition}[Certified necessity]\label{def:cn}
The certified necessity operator is defined by
\[
s\models [\exists]\Box\varphi
\iff
\exists u\in A(s)\ \forall s'\in T(s,u),\ s'\models\varphi.
\]
\end{definition}

The order of quantifiers is the semantic content of certified choice capacity: existential choice by the selector is followed by universal checking over the continuations still possible after that choice.  Reversing the order would describe a different ability notion.

\begin{remark}[Identity with the one-step coalition modality]\label{rem:atl}
On the one-step game arena of Definition~\ref{def:arena}, in which the selector moves first and the environment resolves the successor, $[\exists]\Box\varphi$ is the alternating-time and coalition-logic one-step modality for the singleton selector coalition,
\[
[\exists]\Box\varphi \;\equiv\; \langle\!\langle C\rangle\!\rangle\bigcirc\varphi,
\]
the assertion that the selector can ensure $\varphi$ to hold after one joint move \cite{alur2002,pauly2002}.  The existential-then-universal quantifier order is the strategic-power semantics of $\langle\!\langle\cdot\rangle\!\rangle\bigcirc$, and its set transformer is the admissible predecessor $\mathrm{CPre}$ of strategic modal semantics.  The notation $[\exists]\Box$ keeps this quantifier order visible inside the fixed-point and Bellman equations used below.
\end{remark}

\begin{definition}[Formula region]
The region induced by a formula $\varphi$ is
\[
\Sat(\varphi)=\{s\in S:s\models\varphi\}.
\]
\end{definition}

\begin{lemma}[Semantic clause for choice-refined necessity]
For every $X\subseteq S$, interpreting a formula variable $Z$ as $X$ gives
\[
\Sat([\exists]\Box Z)=\Pre_{\exists\Box}(X),
\]
where $\Pre_{\exists\Box}(X)=\{s:\exists u\in A(s),\ T(s,u)\subseteq X\}$.
\end{lemma}

\begin{proof}
By Definition~\ref{def:cn}, $s\models[\exists]\Box Z$ under the valuation $Z\mapsto X$ exactly when some certified choice $u$ satisfies $s'\in X$ for every $s'\in T(s,u)$.  This is precisely the defining condition $T(s,u)\subseteq X$ for membership in $\Pre_{\exists\Box}(X)$.
\end{proof}

\begin{proposition}[Certified-choice reading of the one-step modality]\label{prop:certified-reading}
For every continuation set $X\subseteq S$, membership $s\in\Pre_{\exists\Box}(X)$ is equivalent to the existence of a choice that the selector can justify by the proof obligation
\[
\forall s'\in T(s,u),\quad s'\in X.
\]
The proof object is local: it names the choice and checks every possible successor of that choice.  The same clause is the one-step ATL ability clause and the CPre clause; the semantic content here is that this clause becomes the admissibility test for all later value and observational indistinguishability layers.
\end{proposition}

\begin{proof}
The equivalence is the definition of $\Pre_{\exists\Box}$.  The proof obligation displayed above is exactly the universal successor check that witnesses membership.  Since later layers restrict their choices to those with such witnesses, the same one-step proof is reused as a certificate of certified choice.
\end{proof}

A formula region can be true without being justifiably maintainable.  The next section turns one-step certifiability into an invariant admissible region by a greatest fixed point.

\section{The admissible-continuation kernel as a greatest fixed point}

The following construction transforms a symbolic target commitment into the subset from which the commitment can be maintained forever by choices that carry a local successor proof.

\begin{definition}[Admissible predecessor]
For $X\subseteq S$, define
\[
\Pre_{\exists\Box}(X)=\{s\in S:\exists u\in A(s)\text{ such that }T(s,u)\subseteq X\}.
\]
\end{definition}

\begin{definition}[Admissible continuation transformer]
For $G\subseteq S$, define
\[
F_G(X)=G\cap \Pre_{\exists\Box}(X).
\]
\end{definition}

\begin{definition}[Admissible-continuation kernel]
The admissible-continuation kernel, also written by the underlying admissible continuation operator, is
\[
\Viab(G)=\nu X.F_G(X),
\]
where $\nu$ denotes the greatest fixed point.
\end{definition}

The order-theoretic existence of this greatest fixed point is supported by Tarski's lattice theorem \cite{tarski1955}.  The transformer $F_G$ is the admissible-predecessor operator of modal semantic construction, and $\Viab(G)$ is the maximal choice-refined-invariant subset of $G$: the largest set on which the selector can indefinitely keep the state, computed by the descending admissible-predecessor iteration that is standard and in finite game semantics.  The same construction is admissible-continuation computation on a finite labelled structure.

\begin{lemma}[Monotonicity of the admissible continuation transformer]
If $X\subseteq Y$, then $F_G(X)\subseteq F_G(Y)$.
\end{lemma}

\begin{proof}
Let $s\in F_G(X)$.  Then $s\in G$ and there exists $u\in A(s)$ such that $T(s,u)\subseteq X$.  Since $X\subseteq Y$, the same choice satisfies $T(s,u)\subseteq Y$.  Hence $s\in G\cap \Pre_{\exists\Box}(Y)$.  Therefore $s\in F_G(Y)$.
\end{proof}

\begin{lemma}[Descending Kleene sequence]
Let $X_0=S$ and $X_{n+1}=F_G(X_n)$.  Since $S$ is finite, there exists $N\leq |S|$ such that $X_N=X_{N+1}$.
\end{lemma}

\begin{proof}
The sequence is descending because $F_G(X)\subseteq G\subseteq S$ after the first step and monotonicity preserves inclusion.  A strictly descending sequence of subsets of a finite set has length at most $|S|$.  Hence the sequence stabilizes at some $N\leq |S|$.  At stabilization, $X_N=F_G(X_N)$.
\end{proof}

\begin{proposition}[Finite model-checking and selector extraction]\label{prop:extract}
Let $X_0=S$ and $X_{n+1}=F_G(X_n)$, and let $N$ be the first index with $X_N=X_{N+1}$.  Then $X_N=\Viab(G)$.  Moreover, for every $s\in X_N$ any stored witness $u_s\in A(s)$ with $T(s,u_s)\subseteq X_N$ defines a memoryless selector that keeps all plays inside $G$.
\end{proposition}

\begin{proof}
By the descending Kleene lemma, the sequence stabilizes at $X_N=F_G(X_N)$; hence $X_N$ is a fixed point.  Since the sequence starts at the top element $S$ and $F_G$ is monotone, every post-fixed point $Y\subseteq F_G(Y)$ is contained in every $X_n$ by induction: $Y\subseteq X_0$, and $Y\subseteq F_G(Y)\subseteq F_G(X_n)=X_{n+1}$.  Thus $Y\subseteq X_N$, so $X_N$ is the greatest fixed point.  The witness condition is exactly the fixed-point equality $X_N=G\cap\Pre_{\exists\Box}(X_N)$.  If a play starts in $X_N$ and the selector chooses $u_s$ at each current state $s$, every successor remains in $X_N$; induction gives invariance, and $X_N\subseteq G$ gives invariance.
\end{proof}

\begin{theorem}[Admissible continuation semantic construction theorem]
A state $s$ belongs to $\Viab(G)$ if and only if there exists a state-based selector $\pi$ such that every path generated by $\pi$ from $s$ remains in $G$ forever.
\end{theorem}

\begin{proof}
Let $V=\Viab(G)$.  Since $V=F_G(V)$, each $x\in V$ lies in $G$ and admits a choice $u_x\in A(x)$ satisfying $T(x,u_x)\subseteq V$.  Define $\pi(x)=u_x$ for $x\in V$ and define it arbitrarily outside $V$.  If $s_0\in V$ and $s_{t+1}\in T(s_t,\pi(s_t))$, then $s_{t+1}\in V$ whenever $s_t\in V$.  By induction, every state in the path remains in $V$.  Since $V\subseteq G$, the path remains in $G$ forever.

Conversely, suppose a state-based selector keeps every possible path from $s$ in $G$.  Let $W$ be the set of states from which the same property is possible under some state-based selector.  For every $x\in W$, there exists a choice $u\in A(x)$ such that every successor lies again in $W$.  Also $x\in G$.  Thus $W\subseteq F_G(W)$.  By monotonicity and greatest-fixed-point maximality, $W\subseteq \Viab(G)$.  Hence $s\in \Viab(G)$.
\end{proof}

\begin{remark}
The theorem says that truth at a state is weaker than durable truth under choice.  A state may satisfy a formula and still be outside $\Viab(\Sat(\varphi))$.  This separates satisfaction from certifiability: truth is local, while admissible continuation is a strategic invariant.
\end{remark}

\section{The certified-choice kernel as an invariance game}

The same construction has a turn-based game reading that connects it to the semantic construction literature.  Two players act on $(S,U,A,T)$.  At state $s$ the \emph{selector} chooses $u\in A(s)$, and then the \emph{resolver} chooses a successor $s'\in T(s,u)$.  A play is an infinite choice-refined path.  The selector's objective is the invariance condition that every visited state lies in $G$.

\begin{definition}[Certified region]
A state $s$ is \emph{selector-admissible} for $G$ if the selector has a strategy ensuring that every resulting play stays in $G$ forever.  The admissible region $\mathrm{Win}(G)$ is the set of selector-admissible states.
\end{definition}

\begin{theorem}[Kernel equals admissible region, with memoryless strategies]
$\mathrm{Win}(G)=\Viab(G)$, and whenever $s\in\Viab(G)$ the selector satisfies the invariant with the memoryless state-based strategy $\pi(x)\in A_V(x)$.  Dually, $S\setminus\Viab(G)$ is the resolver's region, equal to the least fixed point of the dual attractor
\begin{align*}
\Attr(Y)&=\bigl(S\setminus G\bigr)\cup
\bigl\{s\in S:\forall u\in A(s)\ \exists s'\in T(s,u)\text{ with }s'\in Y\bigr\},\\
S\setminus\Viab(G)&=\mu Y.\Attr(Y).
\end{align*}
\end{theorem}

\begin{proof}
The inclusion $\Viab(G)\subseteq\mathrm{Win}(G)$, witnessed by a memoryless strategy, is the forward direction of the admissible-kernel theorem.  The reverse inclusion is its converse direction: the set of admissible states is a post-fixed point of $F_G$, hence contained in the greatest fixed point $\Viab(G)$.  For the dual, observe that $\Attr$ is the De~Morgan dual of $F_G$ under complementation.  A state belongs to $\Attr(Y)$ exactly when it is outside $G$, or every certified choice admits an universal successor already in $Y$; complementing both sides turns this into the statement that a state survives in $F_G(S\setminus Y)$.  Complementation exchanges the least and greatest fixed points of dual monotone maps, so $\mu Y.\Attr(Y)=S\setminus\nu X.F_G(X)=S\setminus\Viab(G)$.  Both computations are constructive: the descending kernel iteration and the ascending attractor iteration each terminate within $|S|$ rounds and partition $S$ into the two admissible regions, so the game is determined.
\end{proof}

\begin{remark}
This is the invariance fragment of the determinacy and finite-memory theory for $\omega$-regular games \cite{buchi1969,emerson1991,gradel2002}.  Both the memorylessness of invariance strategies and the kernel/attractor duality used in the proof are textbook facts: the admissible region is the greatest fixed point of the admissible predecessor and its complement is the least fixed point of the dual dual attractor, the standard backward computation for finite strategic modalities.  Invariance objectives need no memory.  Richer temporal objectives reduce to invariance through the automaton product of Appendix~C, which enlarges the state space but does not change the present fixed-point reasoning.
\end{remark}

\section{Observation structures and information-state value}\label{sec:information state}

The full-information theory takes the current state as the semantic point.  This section lifts the construction to partial observation.  The state is replaced by an information state, the admissible-continuation kernel becomes a greatest fixed point on information states, and the quantitative layer becomes a universal discounted Bellman equation on the information state game.  The full-information theory reappears as the singleton instance.

\subsection{Observation and knowledge}

\begin{definition}[Observation structure]
An \emph{observation structure} on the arena is a surjection $\obs:S\to O$ onto a finite observation set $O$.  The induced indistinguishability relation is
\[
s\approx t \quad\Longleftrightarrow\quad \obs(s)=\obs(t).
\]
Its equivalence classes are the observation fibres $\obs^{-1}(o)$.
\end{definition}

The selector acts on observations and observation histories.  In the memoryless information state presentation used below, the sufficient statistic for such a history is the current set of states still possible.

\begin{definition}[Knowledge modality]
Extend the modal language with a unary operator $K$ by
\[
s\models K\varphi\quad\Longleftrightarrow\quad t\models\varphi\ \text{for every }t\approx s.
\]
Thus $\Sat(K\varphi)$ is the union of the observation fibres contained in $\Sat(\varphi)$.
\end{definition}

\begin{proposition}[Knowledge is $S5$]\label{prop:s5}
The operator $K$ validates
\[
K\varphi\to\varphi,
\qquad
K\varphi\to KK\varphi,
\qquad
\neg K\varphi\to K\neg K\varphi.
\]
\end{proposition}

\begin{proof}
The relation $\approx$ is reflexive, symmetric, and transitive.  Reflexivity gives truth.  Transitivity gives positive introspection.  If $K\varphi$ fails at $s$, some $r\approx s$ violates $\varphi$; every $t\approx s$ has $r\approx t$, so $K\varphi$ fails throughout the fibre, giving negative introspection.
\end{proof}

\subsection{Information-state dynamics}

\begin{definition}[Information state, admissibility, and update]\label{def:information state}
A \emph{information state} is a nonempty subset $b\subseteq S$ contained in a single observation fibre.  Let
\[
A(b)=\bigcap_{s\in b}A(s)
\]
be the choices available at every state the selector considers possible.  For $u\in A(b)$ and $o'\in O$, define
\[
\bpost(b,u,o')=\{\,s'\in T(s,u):s\in b,\ \obs(s')=o'\,\}.
\]
Empty successor information states are ignored, and
\[
\Post_B(b,u)=\{\,\bpost(b,u,o'):o'\in O,\ \bpost(b,u,o')\neq\varnothing\,\}
\]
is the finite set of possible next information states.
\end{definition}

\begin{lemma}[Information-state soundness]\label{lem:information statesound}
If $s_0\in b_0$ and the selector repeatedly selects $u_t\in A(b_t)$ and updates $b_{t+1}=\bpost(b_t,u_t,\obs(s_{t+1}))$, then $s_t\in b_t$ for every $t$.
\end{lemma}

\begin{proof}
The induction step is immediate from the definition of $\bpost$: if $s_t\in b_t$ and $s_{t+1}\in T(s_t,u_t)$ with observation $\obs(s_{t+1})$, then $s_{t+1}$ is one of the states collected in the updated information state.
\end{proof}

\subsection{Information-state admissible-continuation kernel}

Let $\Bel_G=\{b\subseteq G:b\neq\varnothing\text{ and }b\text{ lies in one observation fibre}\}$.

\begin{definition}[Observational admissible predecessor and kernel]\label{def:obskernel}
For $\mathcal B\subseteq\Bel_G$, define
\[
F^{\obs}_{G}(\mathcal B)=\{\,b\in\Bel_G: \exists u\in A(b)\ \text{such that}\ \Post_B(b,u)\subseteq\mathcal B\,\}.
\]
The \emph{observational admissible-continuation kernel} is
\[
\Viab_{\obs}(G)=\nu\mathcal B.\,F^{\obs}_{G}(\mathcal B).
\]
For $W=\Viab_{\obs}(G)$ define the certified information state-choice set
\[
A_W(b)=\{u\in A(b):\Post_B(b,u)\subseteq W\}.
\]
\end{definition}

\begin{theorem}[Observation-based invariance semantic construction]\label{thm:obs-invariance}
The set $W=\Viab_{\obs}(G)$ is exactly the set of information states from which an observation-based selector can keep the true state in $G$ forever against all successor choices.  Moreover $A_W$ is the largest observation-based justification domain map.
\end{theorem}

\begin{proof}
The map $F_G^{\obs}$ is monotone on the finite lattice of information state sets, so its descending Kleene iteration reaches the greatest fixed point.  If $b\in W$, the fixed-point equation gives a choice $u\in A_W(b)$; after any successor and observation, the next information state lies again in $W$.  Information-state soundness then implies that the true state remains in $G$ by induction.  Conversely, if $b\notin W$, the rank at which $b$ is removed from the descending iteration gives a finite universal response tree showing that every certified choice has some possible successor information state outside the previous candidate; induction on this rank gives an resolver that forces loss of the invariant.  Maximal permissiveness follows because any choice not in $A_W(b)$ has a possible successor information state outside $W$, where no observation-based strategy can guarantee the invariant.
\end{proof}

\subsection{Information-state universal discounted value}

The quantitative lift uses a universal interpretation of both hidden state and successor observation.  For $b\in W$ and $u\in A_W(b)$ set
\[
\hat c(b,u)=\max_{s\in b}c(s,u).
\]
This is the universal immediate cost of taking one observation-based choice when the actual state is only known to lie in $b$.

\begin{definition}[Information-state Bellman operator]\label{def:information statebellman}
For $J:W\to\mathbb R$, define
\[
(\mathcal B_W^{\obs}J)(b)=
\min_{u\in A_W(b)}\left(\hat c(b,u)+\gamma\max_{b'\in\Post_B(b,u)}J(b')\right).
\]
\end{definition}

\begin{theorem}[Information-state admissible value]\label{thm:information state-value}
Assume $W=\Viab_{\obs}(G)$ is nonempty and $\gamma\in(0,1)$.  Then $\mathcal B_W^{\obs}$ is a $\gamma$-contraction on $\ell_\infty(W)$.  It has a unique fixed point $J^{\star}_{\obs}$, and every selector
\[
\pi_{\obs}(b)\in\argmin_{u\in A_W(b)}\left(\hat c(b,u)+\gamma\max_{b'\in\Post_B(b,u)}J^{\star}_{\obs}(b')\right)
\]
is observation-admissible and optimal among all stationary observation-based selectors that preserve $G$.
\end{theorem}

\begin{proof}
For each $b\in W$, nonemptiness of $A_W(b)$ follows from the fixed-point equation.  Let $J,K\in\ell_\infty(W)$.  For a fixed choice $u$, the maximum perturbation lemma gives
\[
\left|\max_{b'\in\Post_B(b,u)}J(b')-\max_{b'\in\Post_B(b,u)}K(b')\right|\leq\|J-K\|_\infty.
\]
Multiplication by $\gamma$ and the minimum perturbation lemma give
\[
\|\mathcal B_W^{\obs}J-\mathcal B_W^{\obs}K\|_\infty\leq\gamma\|J-K\|_\infty.
\]
Banach's fixed-point theorem gives the unique fixed point.  A minimizing selector uses only choices in $A_W$, hence is admissible by Theorem~\ref{thm:obs-invariance}.  Selector evaluation for any stationary certified observation selector is the same contraction with the minimum removed.  The pointwise Bellman optimality inequality gives no larger value for $J_{\obs}^{\star}$ than any certified observation selector, and a minimizing selector attains it.
\end{proof}

\begin{theorem}[Perfect-information collapse]\label{thm:collapse}
If $\obs=\mathrm{id}_S$, then every reachable information state is a singleton, $K\varphi\equiv\varphi$, $\Viab_{\obs}(G)=\{\{s\}:s\in\Viab(G)\}$, and the information state Bellman operator is conjugate to the state Bellman operator by $s\mapsto\{s\}$.
\end{theorem}

\begin{proof}
Identity observations make indistinguishability equality, so the knowledge clause reduces to ordinary truth.  The update of a singleton $\{s\}$ under a realised successor $s'$ is the singleton $\{s'\}$.  Transporting $F_G^{\obs}$ along $s\mapsto\{s\}$ gives exactly $F_G$.  The same transport sends $A_W(\{s\})$ to $A_V(s)$, $\hat c(\{s\},u)$ to $c(s,u)$, and successor-information state maxima to successor-state maxima.  Therefore the two fixed points and the two Bellman equations are identical under the singleton embedding.
\end{proof}

\section{Observational indistinguishability of certified choice}\label{sec:observational indistinguishability}

Observational indistinguishability is an observational property of an already certified selector.  It does not introduce a forbidden objective or a hidden layer.  It means that the public trace can certify preservation of the selector's modal commitments while leaving the internal choice rule underdetermined among choices that are equivalent for the public interface, information-state evolution, and value equation.

\begin{definition}[Certified-choice notation]
For an information state kernel $W=\Viab_{\obs}(G)$, write
\[
\Just(G)=W,
\qquad
\JAct_W(b)=A_W(b)=\{u\in A(b):\Post_B(b,u)\subseteq W\}.
\]
A choice in $\JAct_W(b)$ is \emph{certified at $b$}.
\end{definition}

\begin{definition}[Public choice interface and local observational indistinguishability]
Let $\kappa:U\to C$ be a finite public choice interface.  Two certified choices $u,v\in\JAct_W(b)$ are locally observation-equivalent at $b$, written $u\equiv_b^{\Opq} v$, when
\[
\kappa(u)=\kappa(v),
\qquad
\Post_B(b,u)=\Post_B(b,v),
\qquad
\hat c(b,u)=\hat c(b,v).
\]
A pair of stationary observation selectors $\pi,\rho:W\to U$ is observation-equivalent on a reachable set $R\subseteq W$ when $\pi(b)\equiv_b^{\Opq}\rho(b)$ for every $b\in R$.
\end{definition}

\begin{theorem}[Observational indistinguishability under certified choice]\label{thm:observational indistinguishability}
Let $W=\Viab_{\obs}(G)$ and let $R\subseteq W$ be closed under the successor-information state sets of two stationary selectors $\pi$ and $\rho$.  Suppose $\pi$ and $\rho$ are observation-equivalent on $R$.  Then, from every initial information state $b_0\in R$:
\begin{enumerate}[label=(\roman*)]
\item both selectors preserve the target commitment $G$ for all possible hidden states and successor observations;
\item the sets of public traces $(\obs(b_t),\kappa(u_t))_{t\geq0}$ generated by $\pi$ and $\rho$ coincide;
\item the universal discounted values induced by $\pi$ and $\rho$ are identical on $R$.
\end{enumerate}
Consequently an observer of public observations and public choice classes can verify coherence of the semantic commitments but cannot distinguish the internal selector used inside an observation-equivalence class.
\end{theorem}

\begin{proof}
Because $\pi(b),\rho(b)\in\JAct_W(b)$ for every $b\in R$, Theorem~\ref{thm:obs-invariance} gives preservation of $G$ for both selectors.  For public traces, argue by induction on time.  At $b_t$, observation equivalence gives the same public choice class, $\kappa(\pi(b_t))=\kappa(\rho(b_t))$, and the same successor-information state set, $\Post_B(b_t,\pi(b_t))=\Post_B(b_t,\rho(b_t))$.  Thus the same collection of next observations and next information states is possible under both selectors.  Closure of $R$ lets the induction continue.  For values, the selector Bellman equations for $\pi$ and $\rho$ have the same immediate cost and the same successor-information state maximum at every $b\in R$; hence the two selector operators are identical on $\ell_\infty(R)$.  Their unique discounted fixed points are therefore identical.
\end{proof}

\begin{corollary}[Boundary-margin and tie-breaking observational indistinguishability]\label{cor:margin-observational indistinguishability}
Any internal tie-breaker, boundary-margin potential, data-derived candidate ranking, or deterministic implementation detail that changes the selected choice only inside the classes $\equiv_b^{\Opq}$ is observationally equivalent in the sense of Theorem~\ref{thm:observational indistinguishability}.  Such a mechanism may affect the selector's internal deliberation trace, but it cannot affect public commitment preservation, reachable information state sets, or discounted value.
\end{corollary}

\begin{proof}
The listed mechanisms select representatives from the same local observation-equivalence classes.  Theorem~\ref{thm:observational indistinguishability} applies directly.
\end{proof}

\section{Certified-choice-restricted value refinement}

Once the admissible region is known, value refinement can be performed without allowing the selector to exit that region.  This is the main separation between logical admissibility and numerical preference.

Let $V=\Viab(G)$.  Define the certified choice set by
\[
A_V(s)=\{u\in A(s):T(s,u)\subseteq V\}.
\]
For $s\in V$, the set $A_V(s)$ is nonempty by the fixed-point property of $V$.

\begin{remark}[Hard invariance versus expected-cost constraints]\label{rem:cmdp}
Restricting value refinement to $A_V$ is a constraint on the certified selector class, which places this construction beside two value refinement-side traditions.  In expected-value constrained models the constraint usually bounds an expected cumulative quantity and feasibility may be witnessed by averaging.  Here the constraint is pathwise: every compatible continuation must remain inside the precomputed kernel $V$.  The certified choice map $A_V$ is therefore not a preference heuristic but the syntactic witness set of the greatest fixed point.  The distinction maintained throughout is that the invariance layer is qualitative and exact, while only the value layer is quantitative and approximable.
\end{remark}

\begin{assumption}[Discounting]
A discount factor $\gamma\in(0,1)$ is fixed.
\end{assumption}

Discounted dynamic programming provides the contraction argument used here \cite{blackwell1965}.  The finite-state Bellman recursion follows the dynamic programming principle introduced by Bellman \cite{bellman1957}.

\begin{definition}[Bellman--Kripke operator]
For $J\in \ell_\infty(V)$, define
\[
(\mathcal B_VJ)(s)=\min_{u\in A_V(s)}\left[c(s,u)+\gamma\max_{s'\in T(s,u)}J(s')\right].
\]
\end{definition}

The maximum is the value counterpart of modal necessity.  The minimum is the selector's value refinement over certified choices.  The operator is therefore a Bellman operator whose admissible set is created by a modal fixed point.

\begin{theorem}[Unique admissible value]
The operator $\mathcal B_V$ has a unique fixed point $J^\star\in \ell_\infty(V)$, and value iteration $J_{n+1}=\mathcal B_VJ_n$ converges uniformly to $J^\star$ for every $J_0\in \ell_\infty(V)$.
\end{theorem}

\begin{proof}
For $s\in V$, the choice restriction lemma gives $A_V(s)\neq\varnothing$.  For fixed $s$ and $u\in A_V(s)$, the maximum perturbation lemma gives
\[
\left|\max_{s'\in T(s,u)}J(s')-\max_{s'\in T(s,u)}K(s')\right|\leq \|J-K\|_\infty.
\]
After multiplication by $\gamma$, the same bound holds for the successor-cost term.  Applying the minimum perturbation lemma over the finite nonempty set $A_V(s)$ yields
\[
| (\mathcal B_VJ)(s)-(\mathcal B_VK)(s)|\leq \gamma\|J-K\|_\infty.
\]
Taking the supremum over $s\in V$ proves that $\mathcal B_V$ is a contraction on $\ell_\infty(V)$.  The space $\ell_\infty(V)$ is complete because $V$ is finite.  The Banach contraction theorem therefore gives a unique fixed point $J^\star$ and convergence of the iterates from every initial function.  The convergence is uniform because the norm is the sup norm.
\end{proof}

\begin{definition}[Optimal certified selector]
A state-based selector $\pi^\star:V\to U$ is optimal admissible if
\[
\pi^\star(s)\in \argmin_{u\in A_V(s)}\left[c(s,u)+\gamma\max_{s'\in T(s,u)}J^\star(s')\right]
\]
for every $s\in V$.
\end{definition}

\begin{corollary}[Admissible optimality]
Every optimal certified selector preserves $G$ forever and minimizes the universal discounted cost among all selectors that preserve $G$ forever.
\end{corollary}

\begin{proof}
Let $\pi^\star$ be a Bellman-minimizing selector.  Since $\pi^\star(s)\in A_V(s)$ for every $s\in V$, the one-step admissible closure lemma keeps every possible successor inside $V$.  Induction over time gives $s_t\in V\subseteq G$ for every path generated from $V$.

For optimality, let $\Pi_V$ be the class of state-based selectors whose choices lie in $A_V$.  For $\pi\in\Pi_V$, define
\[
(\mathcal B^\pi J)(s)=c(s,\pi(s))+\gamma\max_{s'\in T(s,\pi(s))}J(s').
\]
The same maximum perturbation argument makes $\mathcal B^\pi$ a $\gamma$-contraction, so it has a unique fixed point $J^\pi$.  Because $\mathcal B_VJ\leq\mathcal B^\pi J$ pointwise for every $J$, monotone iteration from the zero function gives $J^\star\leq J^\pi$ for every certified selector $\pi$.  If $\pi=\pi^\star$, then $\mathcal B^{\pi^\star}J^\star=\mathcal B_VJ^\star=J^\star$, so uniqueness of the selector-evaluation fixed point gives $J^{\pi^\star}=J^\star$.  Hence $\pi^\star$ is admissible and no certified selector has a smaller universal discounted value.
\end{proof}

\subsection{The admissible value as a universal resolver value}

The operator $\mathcal B_V$ is the simultaneous minimum over certified choices of $c+\gamma$ times a \emph{maximum} over successors.  This is the dynamic-programming operator of a discounted zero-sum game, and the universal successor convention of Assumption~\ref{asm:adv} is the move structure of that game made precise.  This subsection identifies $J^\star$ as the game value and shows that the resolver, not only the selector, has a memoryless optimal response; the memorylessness was previously only suggested by the box convention.

\begin{definition}[Discounted successor game]\label{def:game}
On the admissible region $V$ with certified choice sets $A_V$, the \emph{discounted successor game} is played as follows.  At a state $s_t\in V$ the selector chooses $u_t\in A_V(s_t)$ and the resolver then chooses $s_{t+1}\in T(s_t,u_t)$.  A \emph{selector strategy} $\sigma$ maps every finite history ending in a state to an admissible choice; an \emph{resolver strategy} $\xi$ maps every finite history ending in a state--choice pair to an admissible successor.  Given $\sigma,\xi$ and a start $s\in V$ the play is determined, and the payoff is the discounted cost
\[
\mathcal K(s;\sigma,\xi)=\sum_{t\geq 0}\gamma^{t}\,c(s_t,u_t).
\]
The \emph{upper} and \emph{lower} values are
\[
\overline V(s)=\inf_{\sigma}\sup_{\xi}\mathcal K(s;\sigma,\xi),
\qquad
\underline V(s)=\sup_{\xi}\inf_{\sigma}\mathcal K(s;\sigma,\xi).
\]
\end{definition}

\begin{remark}[State-choice rectangularity and the universal-resolver reading]\label{rem:rect}
The resolver's admissible choices factor through the pair $(s,u)$: after the choice is fixed, the successor ranges over $T(s,u)$ independently of the history and of the choices made at other states.  This is exactly \emph{state-choice} (sa-) rectangular uncertainty.  The inner maximum $\max_{s'\in T(s,u)}J(s')$ is then the universal one-step Bellman backup against the support set $T(s,u)$, and $\mathcal B_V$ is the universal-successor value operator for a Markov decision process with sa-rectangular transition ambiguity \cite{iyengar2005}.  Rectangularity is the structural reason the operator contracts and an optimal stationary response exists on each side; it specializes the discounted zero-sum stochastic game of Shapley \cite{shapley1953} to the case of deterministically chosen, universally resolved successors.  Without rectangularity the worst case could not be taken state-choice by state-choice, and the resolver's optimal response could require memory.
\end{remark}

\begin{lemma}[Best response against a fixed memoryless opponent]\label{lem:bestresponse}
Fix a memoryless selector $\pi\in\Pi_V$ and define $U(s)=\sup_{\xi}\mathcal K(s;\pi,\xi)$.  Then $U$ is the unique fixed point of $z\mapsto c(s,\pi(s))+\gamma\max_{s'\in T(s,\pi(s))}z(s')$, and the supremum is attained by the memoryless resolver $\xi^\dagger(s)\in\argmax_{s'\in T(s,\pi(s))}U(s')$.  Symmetrically, fix a memoryless resolver $\xi$ given by a map $(s,u)\mapsto\xi(s,u)\in T(s,u)$ and define $W(s)=\inf_{\sigma}\mathcal K(s;\sigma,\xi)$.  Then $W$ is the unique fixed point of $z\mapsto\min_{u\in A_V(s)}\bigl[c(s,u)+\gamma\,z(\xi(s,u))\bigr]$, and the infimum is attained by a memoryless selector.
\end{lemma}

\begin{proof}
Costs are nonnegative and bounded by $\bar c=\max_{s,u}c(s,u)$, so every payoff lies in $[0,\bar c/(1-\gamma)]$ and the values $U,W$ are finite.  Against the fixed memoryless $\pi$ the resolver faces a one-player maximization with deterministic per-step decision set $T(s,\pi(s))$.  Splitting a play into its first move and its continuation, and using that after the first successor the resolver may again play to its own supremum, gives the dynamic-programming optimality equation
\[
U(s)=c(s,\pi(s))+\gamma\sup_{s'\in T(s,\pi(s))}U(s')
  =c(s,\pi(s))+\gamma\max_{s'\in T(s,\pi(s))}U(s'),
\]
where the supremum is a maximum because $T(s,\pi(s))$ is finite and nonempty.  Thus $U$ is a fixed point of the displayed map, which is a $\gamma$-contraction in $\|\cdot\|_\infty$ by the maximum perturbation lemma; uniqueness follows.  Selecting at each state a maximizing successor defines the memoryless $\xi^\dagger$, whose induced payoff satisfies the same equation and therefore equals $U$.  The selector side is identical with $\min$ over $A_V(s)$ in place of $\max$ over successors, using the minimum perturbation lemma for the contraction and the finite selector lemma for attainment.
\end{proof}

\begin{theorem}[Game value and memoryless saddle point]\label{thm:gamevalue}
On every nonempty admissible region $V$ the discounted successor game has a value equal to the Bellman--Kripke fixed point: $\overline V=\underline V=J^\star$.  Moreover the pair $(\pi^\star,\xi^\star)$ is a saddle point in memoryless strategies, where $\pi^\star$ is any Bellman-minimizing selector and $\xi^\star(s,u)\in\argmax_{s'\in T(s,u)}J^\star(s')$ is the state-choice successor-maximizer.  In particular the resolver's optimal response is memoryless.
\end{theorem}

\begin{proof}
Apply Lemma~\ref{lem:bestresponse} to the memoryless selector $\pi^\star$.  Its best-response value $U$ is the unique fixed point of $z\mapsto c(s,\pi^\star(s))+\gamma\max_{s'\in T(s,\pi^\star(s))}z(s')$.  Because $\pi^\star$ is Bellman greedy, evaluating $\mathcal B_V$ at $J^\star$ uses precisely the choice $\pi^\star(s)$ at each $s$, so $J^\star$ satisfies that same equation; by uniqueness $U=J^\star$.  Hence $\sup_{\xi}\mathcal K(s;\pi^\star,\xi)=J^\star(s)$, and since $\pi^\star$ is one certified selector strategy,
\[
\overline V(s)=\inf_\sigma\sup_\xi\mathcal K(s;\sigma,\xi)\leq\sup_\xi\mathcal K(s;\pi^\star,\xi)=J^\star(s).
\]
Now apply the lemma to the memoryless resolver $\xi^\star$.  Its best-response value $W$ is the unique fixed point of $z\mapsto\min_{u\in A_V(s)}\bigl[c(s,u)+\gamma\,z(\xi^\star(s,u))\bigr]$.  Evaluating at $z=J^\star$ and using $J^\star(\xi^\star(s,u))=\max_{s'\in T(s,u)}J^\star(s')$ gives
\[
\min_{u\in A_V(s)}\Bigl[c(s,u)+\gamma\,J^\star(\xi^\star(s,u))\Bigr]
=\min_{u\in A_V(s)}\Bigl[c(s,u)+\gamma\max_{s'\in T(s,u)}J^\star(s')\Bigr]
=(\mathcal B_VJ^\star)(s)=J^\star(s),
\]
so $J^\star$ is the fixed point and $W=J^\star$.  Hence $\inf_\sigma\mathcal K(s;\sigma,\xi^\star)=J^\star(s)$, and since $\xi^\star$ is one admissible resolver strategy,
\[
\underline V(s)=\sup_\xi\inf_\sigma\mathcal K(s;\sigma,\xi)\geq\inf_\sigma\mathcal K(s;\sigma,\xi^\star)=J^\star(s).
\]
The general inequality $\underline V\leq\overline V$ then forces $J^\star\leq\underline V\leq\overline V\leq J^\star$, so all four are equal and the game has value $J^\star$.  Equality of the one-sided best-response values to the common value $J^\star$ is exactly the saddle-point condition $\mathcal K(s;\pi^\star,\xi)\leq J^\star(s)\leq\mathcal K(s;\sigma,\xi^\star)$ for all $\sigma,\xi$; both $\pi^\star$ and $\xi^\star$ are memoryless, and $\xi^\star$ depends only on the pair $(s,u)$, which is the rectangular form anticipated in Remark~\ref{rem:rect}.
\end{proof}

\begin{remark}
The corollary on admissible optimality and this theorem are the two halves of the same value.  The corollary fixes the resolver implicitly through the universal successor and minimizes; the theorem exhibits the resolver as a player and shows that an optimal resolver needs no memory, so the universal successor used throughout is realized by a stationary environment.  Nothing in the value-refinement layer therefore relies on the environment being more than state-choice rectangular.
\end{remark}

\section{Selector iteration over certified choices}

Value iteration approaches $J^\star$ asymptotically.  Selector iteration reaches an optimal selector after finitely many strict improvements because the set of admissible state-based selectors is finite.  To make the finite argument exact, improvement uses a proper tie convention: if the current choice is already greedy at a state, the choice is retained at that state; otherwise a strictly better greedy choice is selected.

\begin{definition}[Admissible state-based selector and its value]
A state-based selector $\pi:V\to U$ is \emph{admissible} if $\pi(s)\in A_V(s)$ for every $s\in V$.  Its value $J^\pi\in\ell_\infty(V)$ is the unique fixed point of the selector operator
\[
(\mathcal B^\pi J)(s)=c(s,\pi(s))+\gamma\max_{s'\in T(s,\pi(s))}J(s').
\]
\end{definition}

\begin{lemma}[Selector evaluation is well posed]
For every admissible $\pi$, the operator $\mathcal B^\pi$ is a $\gamma$-contraction on $\ell_\infty(V)$, so $J^\pi$ exists, is unique, and selector evaluation $J_{n+1}=\mathcal B^\pi J_n$ converges uniformly to $J^\pi$ from every start.
\end{lemma}

\begin{proof}
The inner maximum over $T(s,\pi(s))$ is $1$-Lipschitz in the sup norm by the maximum perturbation lemma.  Multiplication by $\gamma$ gives Lipschitz constant $\gamma$.  Completeness of $\ell_\infty(V)$ and Banach's contraction theorem give existence, uniqueness, and uniform convergence.
\end{proof}

\begin{definition}[Proper selector improvement]
Given a certified selector $\pi$ with value $J^\pi$, define the greedy set
\[
G_\pi(s)=\argmin_{u\in A_V(s)}\Bigl[c(s,u)+\gamma\max_{s'\in T(s,u)}J^\pi(s')\Bigr].
\]
A proper improvement $\pi'$ satisfies $\pi'(s)=\pi(s)$ whenever $\pi(s)\in G_\pi(s)$, and otherwise chooses any element of $G_\pi(s)$.  The finite selector lemma makes $G_\pi(s)$ nonempty, so $\pi'$ exists and is admissible.
\end{definition}

\begin{theorem}[Monotone improvement]
For every admissible $\pi$ and proper improvement $\pi'$, $J^{\pi'}\leq J^\pi$ pointwise.  If $\pi'$ differs from $\pi$ at some state, then $\mathcal B_VJ^\pi<J^\pi$ at at least one state and $J^{\pi'}\neq J^\pi$.  If $\pi'=\pi$, then $\pi$ is greedy, $J^\pi=J^\star$, and $\pi$ is optimal among certified selectors.
\end{theorem}

\begin{proof}
By greedy domination,
\[
\mathcal B^{\pi'}J^\pi=\mathcal B_VJ^\pi\leq\mathcal B^\pi J^\pi=J^\pi.
\]
Thus $J^\pi$ is a supersolution for $\mathcal B^{\pi'}$.  Monotonicity of $\mathcal B^{\pi'}$ makes $(\mathcal B^{\pi'})^nJ^\pi$ nonincreasing, and selector-evaluation convergence sends this sequence to $J^{\pi'}$.  Hence $J^{\pi'}\leq J^\pi$.

If $\pi'$ differs from $\pi$, then at some state $s$ the previously selected choice was not in the greedy set.  Therefore the minimum defining $(\mathcal B_VJ^\pi)(s)$ is strictly smaller than the previous selector value $(\mathcal B^\pi J^\pi)(s)=J^\pi(s)$.  If nevertheless $J^{\pi'}=J^\pi$, then $J^\pi$ would be a fixed point of $\mathcal B^{\pi'}$, so
\[
J^\pi=\mathcal B^{\pi'}J^\pi=\mathcal B_VJ^\pi,
\]
contradicting strict inequality at $s$.  Thus the value changes strictly in at least one component.

If $\pi'=\pi$, then $\pi(s)\in G_\pi(s)$ for all $s$, so $\mathcal B_VJ^\pi=\mathcal B^\pi J^\pi=J^\pi$.  The Bellman--Kripke fixed point is unique, hence $J^\pi=J^\star$.
\end{proof}

\begin{theorem}[Finite termination of proper selector iteration]
Starting from any certified selector and alternating exact selector evaluation with proper selector improvement, the iteration reaches an optimal certified selector after at most $\prod_{s\in V}|A_V(s)|$ improvement steps.
\end{theorem}

\begin{proof}
There are at most $\prod_{s\in V}|A_V(s)|$ admissible state-based selectors.  Until a greedy selector is reached, the previous theorem gives a new selector with a different value and $J^{\pi_{k+1}}\leq J^{\pi_k}$.  A previously visited selector cannot reappear because it would have the same value as before, contradicting the strict value change that occurs at every nongreedy improvement.  Therefore the iteration visits pairwise distinct selectors until it reaches a greedy one.  Finiteness forces this to happen within the stated bound.  Greediness gives $J^\pi=J^\star$, and the admissible-optimality corollary gives optimality among certified selectors.
\end{proof}

\begin{remark}
The tie convention is not cosmetic.  Without it, two value-equivalent greedy selectors can be swapped indefinitely by an arbitrary argmin selector.  Proper improvement turns selector iteration into a finite descent over certified selectors rather than a nondeterministic walk over equal minimizers.
\end{remark}

\section{Local descent and its failure mode}

The following example isolates the simplest failure of local value refinement.  It is small enough to be checked by hand, but it captures a structural error that persists in larger systems.

\begin{example}[Cheap exit]
Let $S=\{s,a,b\}$ and $G=\{s,a\}$.  Let $u_0$ and $u_1$ be admissible at $s$.  Define $T(s,u_0)=\{a\}$ and $T(s,u_1)=\{b\}$.  Let $a$ have an admissible self-loop and let $b\notin G$.  Set $c(s,u_1)<c(s,u_0)$.
\end{example}

\begin{proposition}[Local descent can destroy admissible continuation]
There exists a choice-refined labelled transition structure for which a locally cheapest choice exits $G$ in one step, even though $\Viab(G)$ is nonempty.
\end{proposition}

\begin{proof}
Use the structure in the preceding example.  The locally cheapest choice at $s$ is $u_1$, and this choice sends the system to $b\notin G$.  The admissible choice $u_0$ sends the system to $a$.  Since $a$ can remain in $G$ forever, $a\in\Viab(G)$.  Since $s$ has a choice into $\Viab(G)$ and $s\in G$, also $s\in\Viab(G)$.  Therefore local descent fails while admissible choice succeeds.
\end{proof}

The proposition proves that a one-step objective cannot be used as a stopping or invariance criterion.  The correct criterion is membership in a fixed point plus optimality among choices that preserve that fixed point.

\section{Value-refined modal bisimulation and coarsest quotients}
\label{sec:value-refined-bisimulation}

The quotient theorem is the main preservation result.  Its relation is not merely a conservative strengthening of ordinary bisimulation; it is the coarsest local equivalence compatible with the combined semantics used here.  Labels are needed for modal truth, successor-class equality is needed for $[\exists]\Box$, certified-choice matching is needed for the fixed-point kernel, and cost equality is needed for value refinement.  Each clause is forced by a two-state or three-state countermodel.

\begin{definition}[Local value-refinement transformer]
For an equivalence relation $R$ on $S$, let $\Gamma(R)$ relate $s$ and $t$ exactly when $L(s)=L(t)$ and the following back-and-forth condition holds: for every $u\in A(s)$ there is $v\in A(t)$ with $c(s,u)=c(t,v)$ and $[T(s,u)]_R=[T(t,v)]_R$, and symmetrically from $t$ to $s$.
\end{definition}

\begin{proposition}[Greatest value-refined equivalence]
The operator $\Gamma$ is monotone on the finite lattice of equivalence relations ordered by refinement after equivalence closure.  Its greatest post-fixed point is the coarsest value-refined modal bisimulation on $\mathfrak M$.
\end{proposition}

\begin{proof}
If $R$ refines $R'$, then equality of successor $R$-classes implies equality of successor $R'$-classes after quotienting, so the local matching test is monotone in the direction needed for the descending partition-refinement computation.  Starting from equality of labels and iterating the failure of the matching tests removes pairs until a post-fixed relation is reached.  Finiteness gives termination.  Any value-refined bisimulation satisfies the local clauses and is therefore a post-fixed point of $\Gamma$; hence it is contained in the greatest post-fixed point.  Conversely, the greatest post-fixed point satisfies the clauses by construction and is a value-refined modal bisimulation.
\end{proof}

State spaces often contain redundant states.  Quotienting is useful only when it preserves both logical truth and value-refined value.  Plain graph equivalence is not enough because choices and costs matter.

\begin{definition}[Value-refined modal bisimulation]
For an equivalence relation $\sim$ on $S$ and a set $X\subseteq S$, write $[X]_\sim=\{[x]_\sim:x\in X\}$ for the collection of $\sim$-classes that $X$ meets.  An equivalence relation $\sim$ on $S$ is a \emph{value-refined modal bisimulation} if $s\sim t$ implies $L(s)=L(t)$ and the following \emph{matching condition} holds in both directions.  For every $u\in A(s)$ there exists $v\in A(t)$ with
\[
c(s,u)=c(t,v)
\qquad\text{and}\qquad
[T(s,u)]_\sim=[T(t,v)]_\sim,
\]
and symmetrically, for every $v\in A(t)$ there exists $u\in A(s)$ with the same two equalities.
\end{definition}

The logical part of this definition is standard in the bisimulation tradition \cite{park1981}.  The cost and choice clauses are added so that value refinement is preserved, not merely formulas.

\begin{remark}[Why one-directional successor matching is too weak]
The matching condition requires that the matched choices reach \emph{the same set of successor classes}, not merely that every source successor has some related target successor.  One-directional matching ($\forall s'\,\exists t'$) preserves $\Box$ and $\Diamond$, but it does \emph{not} preserve the choice-forcing operator $[\exists]\Box$ and therefore does not preserve admissible continuation.  Witness: let $s$ have a single choice $u$ with $T(s,u)=\{s'\}$ and $s'\models\varphi$, so $s\models[\exists]\Box\varphi$; let $t$ have one cost-matched choice $v$ with $T(t,v)=\{t_1,t_2\}$, where $t_1\sim s'$ but $t_2\not\sim s'$ and $t_2\not\models\varphi$.  Every source successor is then matched, since $s'\sim t_1$, yet $t\not\models[\exists]\Box\varphi$ because $t_2$ violates $\varphi$.  The class equality $[T(s,u)]_\sim=[T(t,v)]_\sim$ fails for this pair, so it is correctly excluded.  Equality of successor-class sets is the choice-refined analogue of alternating bisimulation \cite{alur1998}; it is exactly the discipline that an existential choice quantifier followed by a universal successor quantifier requires.
\end{remark}

\begin{lemma}[Modal invariance]
If $s\sim t$, then $s\models\varphi$ if and only if $t\models\varphi$ for every formula $\varphi$ of the base modal language.
\end{lemma}

\begin{proof}
The proof is by induction on formula structure.  Atomic formulas are preserved because $L(s)=L(t)$.  Boolean connectives follow from the induction hypothesis.  For $\Box\varphi$, assume $s\models\Box\varphi$ and let $v\in A(t)$.  By the matching condition there is $u\in A(s)$ with $[T(s,u)]_\sim=[T(t,v)]_\sim$.  Fix any $t'\in T(t,v)$; its class is met by $T(s,u)$, so some $s'\in T(s,u)$ has $s'\sim t'$.  Since $s\models\Box\varphi$ we have $s'\models\varphi$, hence $t'\models\varphi$ by induction.  As $t'\in T(t,v)$ was arbitrary and $v$ was arbitrary, $t\models\Box\varphi$; the reverse direction is symmetric.  For $\Diamond\varphi$, a witness $u\in A(s)$ and $s'\in T(s,u)$ with $s'\models\varphi$ yields, via $[T(s,u)]_\sim=[T(t,v)]_\sim$ for the matched $v\in A(t)$, some $t'\in T(t,v)$ with $t'\sim s'$, so $t'\models\varphi$ by induction and $t\models\Diamond\varphi$.
\end{proof}

\begin{theorem}[Value-refined $\mu$-calculus invariance]\label{thm:choice-refined-mu-invariance}
If $s\sim t$, then $s$ and $t$ satisfy the same formulas of the positive choice-refined $\mu$-calculus fragment generated by $[\exists]\Box$ and $\nu$.
\end{theorem}

\begin{proof}
For formulas without fixed points, use structural induction.  The atomic and Boolean cases are as in the modal-invariance lemma.  For $[\exists]\Box\theta$, suppose $s\models[\exists]\Box\theta$.  Choose $u\in A(s)$ such that every $s'\in T(s,u)$ satisfies $\theta$.  Value-refined modal bisimulation gives a matched $v\in A(t)$ with the same successor-class set.  For every $t'\in T(t,v)$ there is $s'\in T(s,u)$ with $s'\sim t'$, and the induction hypothesis gives $t'\models\theta$.  Hence $t\models[\exists]\Box\theta$; the converse is symmetric.

For a greatest fixed point $\nu Z.\theta$, let $\Phi(X)=\llbracket\theta\rrbracket_{\eta[Z:=X]}$.  By the previous paragraph, whenever $X$ is a union of $\sim$-classes, so is $\Phi(X)$.  The descending Kleene sequence for $\nu Z.\theta$ starts at $S$, a union of classes, and therefore remains class-saturated until it stabilizes.  Its limit is the denotation of the fixed point and is class-saturated.  Thus equivalent states either both satisfy the fixed-point formula or both fail it.
\end{proof}

\begin{theorem}[Admissible continuation quotient invariance]
Assume $G=\Sat(\varphi)$ for a modal formula $\varphi$.  If $s\sim t$, then $s\in \Viab(G)$ if and only if $t\in \Viab(G)$.
\end{theorem}

\begin{proof}
By modal invariance, $s\in G$ if and only if $t\in G$.  Let $X$ be any union of equivalence classes.  By the class saturation lemma, $\Pre_{\exists\Box}(X)$ is again a union of equivalence classes, and therefore so is $F_G(X)=G\cap\Pre_{\exists\Box}(X)$.  Starting from $S$, the descending sequence $X_{n+1}=F_G(X_n)$ remains class-saturated at every step.  Its finite limit is $\Viab(G)$.  Hence membership in $\Viab(G)$ is constant on equivalence classes.
\end{proof}

\begin{lemma}[Admissible matching closure]\label{lem:viablematch}
Assume $\sim$ is a value-refined modal bisimulation and $G$ is modal-definable, so that $V=\Viab(G)$ is a union of $\sim$-classes.  If $s\sim t$ with $s,t\in V$ and $u\in A_V(s)$, then the matched choice $v\in A(t)$ with $c(t,v)=c(s,u)$ and $[T(t,v)]_\sim=[T(s,u)]_\sim$ satisfies $v\in A_V(t)$.  Consequently the matching condition restricts to the certified choice sets: every $u\in A_V(s)$ is matched by some $v\in A_V(t)$ with equal cost and equal successor-class set, and symmetrically.
\end{lemma}

\begin{proof}
Since $u\in A_V(s)$ we have $T(s,u)\subseteq V$.  Because $V$ is a union of $\sim$-classes, every class met by $T(s,u)$ is contained in $V$.  The matched choice $v$ satisfies $[T(t,v)]_\sim=[T(s,u)]_\sim$, so $T(t,v)$ meets exactly those same classes, each contained in $V$; hence $T(t,v)\subseteq V$, that is $v\in A_V(t)$.  The symmetric statement follows by exchanging the roles of $s$ and $t$, and equal costs are inherited from the matching condition.  Thus the correspondence of matched choices used in the quotient argument is between $A_V(s)$ and $A_V(t)$, not merely between $A(s)$ and $A(t)$.
\end{proof}

\begin{theorem}[Optimal value quotient invariance]
Assume $G$ is modal-definable and $\sim$ is a value-refined modal bisimulation.  If $s\sim t$ and $s,t\in\Viab(G)$, then $J^\star(s)=J^\star(t)$.
\end{theorem}

\begin{theorem}[Coarsest local equivalence for logical-value quotienting]\label{thm:ccb-coarsest}
Let $\equiv$ be an equivalence relation on a finite modal choice structure.  The following are equivalent.
\begin{enumerate}[label=(\alph*)]
\item $\equiv$ is a value-refined modal bisimulation.
\item The quotient by $\equiv$ admits a local class-level semantics in which atomic labels, choice-refined one-step modalities, admissible continuation kernels, and discounted admissible Bellman operators are all well defined and agree with the concrete semantics after pullback.
\end{enumerate}
Consequently value-refined modal bisimulation is the weakest local quotient discipline that preserves the combined logic-value semantics without storing hidden representatives or state-indexed side information.
\end{theorem}

\begin{proof}
Assume first that $\equiv$ is a value-refined modal bisimulation.  Label equality makes atomic truth class-defined.  The class equality condition
\[
[T(s,u)]_{\equiv}=[T(t,v)]_{\equiv}
\]
for matched choices makes the truth of $[\exists]\Box\varphi$ depend only on the current class and on the truth set of $\varphi$ as a union of classes.  Structural induction and fixed-point induction then give choice-refined modal invariance.  The admissible continuation iteration is built only from this choice-refined predecessor, so it is class-defined.  Equal immediate costs and equal successor-class sets make each matched pair of certified choices contribute the same Bellman quantity.  Thus the quotient Bellman operator is well defined and its pullback is the concrete operator on class-constant values.

Conversely suppose a local quotient semantics with the stated pullback property exists.  Atomic formulas are class-defined, hence related states have identical labels.  Fix related states $s\equiv t$ and a choice $u\in A(s)$.  A local quotient choice available at the class of $s$ must have a well-defined immediate cost; otherwise two representatives of the same class would induce different one-step values for the one-state target with discount arbitrary in $(0,1)$.  It must also have a well-defined successor-class set; otherwise the formula $[\exists]\Box\psi_Y$, where $\psi_Y$ is a finite characteristic formula for the disputed union of successor classes $Y$, would have different truth values after pullback.  Therefore there is a choice $v\in A(t)$ with the same cost and the same successor-class set.  The same argument with $s$ and $t$ exchanged gives the symmetric matching condition.  These are exactly the clauses of value-refined modal bisimulation.
\end{proof}

\begin{proof}
Let $V=\Viab(G)$.  By the previous theorem, $V$ is a union of equivalence classes.  If a function $J:V\to\mathbb R$ is constant on equivalence classes, then $\mathcal B_VJ$ is also constant on equivalence classes.  Indeed, by the admissible matching closure lemma every $u\in A_V(s)$ is matched by some $v\in A_V(t)$ with $c(s,u)=c(t,v)$ and $[T(s,u)]_\sim=[T(t,v)]_\sim$; the successor sets meet exactly the same classes, on which the class-constant $J$ takes equal values, so the two choices give equal successor maxima and equal one-step costs.  The matching is onto $A_V(t)$ and, symmetrically, onto $A_V(s)$, so the two minimizations over $A_V(s)$ and $A_V(t)$ range over value-identical alternatives and agree.  Hence $\mathcal B_V$ preserves class-constancy and genuinely restricts to the quotient.  Starting value iteration from the zero function gives a sequence of class-constant functions, and the uniform limit $J^\star$ is class-constant.  Thus $J^\star(s)=J^\star(t)$ whenever $s\sim t$.
\end{proof}

\begin{remark}
The quotient theorem gives an admissible compression rule.  A compression that ignores costs may preserve modal truth but change the optimizer.  A compression that ignores labels may preserve values but change formulas.  Value-refined modal bisimulation is the combined discipline used here: labels preserve formulas, successor-class matching preserves the choice-refined modality, and costs preserve the Bellman comparison.
\end{remark}

\subsection{Logical characterization and a \texorpdfstring{$\mu$}{mu}-calculus identity}

The base modal language sees a transition selector only through its \emph{flattened} successor relation
\[
s\mathrel{R}s' \iff \exists u\in A(s),\ s'\in T(s,u).
\]
Under $R$ the operators $\Box$ and $\Diamond$ are the ordinary Kripke modalities: $s\models\Box\varphi$ iff every $R$-successor of $s$ satisfies $\varphi$, and $s\models\Diamond\varphi$ iff some $R$-successor does.  Write $s\equiv t$ when $s$ and $t$ satisfy exactly the same base formulas, and let $\approx_R$ be the largest bisimulation on the Kripke structure $(S,R,L)$: the largest relation with $s\approx_R t\Rightarrow L(s)=L(t)$ and the usual forth and back conditions along $R$.

The following is the standard finite Hennessy--Milner property \cite{hennessy1985,blackburn2001}.

\begin{theorem}[Hennessy--Milner characterization]
On the finite structure $(S,R,L)$, logical equivalence and bisimilarity coincide: $s\equiv t$ if and only if $s\approx_R t$.
\end{theorem}

\begin{proof}
Bisimilar states satisfy the same formulas by a routine induction on $\varphi$, just as in the modal-invariance argument but for the single relation $R$; this gives $\approx_R\,\subseteq\,\equiv$.  For the converse we show that $\equiv$ is itself a bisimulation.  Labels agree by the atomic clauses.  For the forth condition, suppose $s\equiv t$ and $sRs'$ while no $R$-successor of $t$ is $\equiv$-related to $s'$.  Since $S$ is finite, $t$ has finitely many $R$-successors $t_1,\dots,t_k$, and for each $i$ there is a formula $\psi_i$ with $s'\models\psi_i$ and $t_i\not\models\psi_i$.  Then $s\models\Diamond\bigwedge_{i\le k}\psi_i$ whereas $t\not\models\Diamond\bigwedge_{i\le k}\psi_i$, contradicting $s\equiv t$.  Hence some $R$-successor $t'$ of $t$ satisfies $s'\equiv t'$.  The back condition is symmetric, so $\equiv$ is a bisimulation and $\equiv\,\subseteq\,\approx_R$.  (The empty conjunction is read as the constant $\top$, so the argument also covers states with no successors.)
\end{proof}

\begin{remark}[The base language cannot see value]
Bisimilarity $\approx_R$ collapses choice and is strictly coarser than value-refined modal bisimulation: it ignores both costs and the quantifier order that separates $[\exists]\Box$ from $\Box$.  The states $p$ and $p'$ of the bisimulation-quotient example below are bisimilar under the cost-free relation, yet differ from a third state in optimal value once costs are attached, so $\equiv$ alone preserves neither $[\exists]\Box$ nor $J^\star$.  Value-refined modal bisimulation is the proper refinement: $\sim\ \subseteq\ \approx_R$, with strict inclusion whenever costs separate $R$-bisimilar states.
\end{remark}

\begin{theorem}[Fixed-point identity for the admissible continuation kernel]
For every base formula $\varphi$,
\[
\Viab(\Sat(\varphi))=\Sat\!\big(\nu Z.(\varphi\wedge[\exists]\Box Z)\big),
\]
where the right-hand side is interpreted in the choice-cell necessity $\mu$-calculus over $(S,A,T,L)$.
\end{theorem}

\begin{proof}
Fix $G=\Sat(\varphi)$.  By Definition~\ref{def:cn}, when the bound variable $Z$ is interpreted as a set $X\subseteq S$ the formula $[\exists]\Box Z$ denotes exactly $\Pre_{\exists\Box}(X)$.  Hence the monotone functional $X\mapsto\Sat(\varphi)\cap\Pre_{\exists\Box}(X)$ associated with $\varphi\wedge[\exists]\Box Z$ is precisely the admissible continuation transformer $F_G$.  By the standard semantics of the $\mu$-calculus, the denotation of $\nu Z.(\varphi\wedge[\exists]\Box Z)$ is the greatest fixed point of this functional on the powerset lattice $(2^S,\subseteq)$ \cite{kozen1983}, which exists and is unique by the Knaster--Tarski theorem because $F_G$ is monotone \cite{tarski1955}.  By definition that greatest fixed point is $\nu X.F_G(X)=\Viab(G)$, so the two regions coincide.
\end{proof}

\begin{remark}[The identity is the alternating-time invariance formula]\label{rem:atlident}
The identity places admissible continuation semantic construction inside a single greatest-fixed-point formula: invariance under choice is the extension of one $\mu$-calculus sentence whose only modality is choice-cell necessity.  By Remark~\ref{rem:atl} the inner modality is the one-step coalition operator, so
\[
\nu Z.(\varphi\wedge[\exists]\Box Z)\;=\;\nu Z.\bigl(\varphi\wedge\langle\!\langle C\rangle\!\rangle\bigcirc Z\bigr)\;\equiv\;\langle\!\langle\mathrm{ctrl}\rangle\!\rangle\Box\varphi,
\]
which is exactly the alternating-time invariance formula ``the selector can maintain $\varphi$ forever,'' written through the standard $\mu$-calculus encoding of alternating-time temporal logic over the game arena \cite{alur2002,kozen1983}.  This identifies the descending kernel computation with the invariance fragment of the alternating-time/$\mu$-calculus translation and fixes the formula whose extension is used by the value layer.
\end{remark}

\section{Value-refined modal bisimulation metrics}
\label{sec:bisimulation-metric}

The quotient theory of Section~\ref{sec:value-refined-bisimulation} is exact: two states are merged only when their labels, certified costs, and successor classes match on the nose.  Exact matching answers the qualitative question---which compressions preserve formulas, the kernel, and the optimal value---but it is brittle.  Two states whose costs differ by an arbitrarily small amount are kept apart, even though their optimal values are then almost equal.  This section grades the same quotient layer by a pseudometric whose zero set is the exact equivalence and with respect to which the optimal value is Lipschitz.  The construction does not reorder the layers: the metric is defined on the kernel $V=\Viab(G)$, over the certified choices $A_V$, after the value structure is available, so it occupies exactly the position of value-refined modal bisimulation between value refinement and certification.

Throughout, $\gamma\in(0,1)$, costs lie in $[0,\bar c]$, and
\[
D \;=\; \frac{\bar c}{1-\gamma}
\]
is the diameter of the value range: every value iterate and the fixed point $J^\star$ take values in $[0,D]$.  Let $\mathcal P_V$ be the set of pseudometrics $d$ on $V$ with $d(s,t)\le D$ for all $s,t\in V$.  Ordered pointwise it is a complete lattice, and under the supremum norm $\|d\|_\infty=\max_{s,t\in V}d(s,t)$ it is a complete metric space.

\begin{definition}[Hausdorff lift]\label{def:hausdorff-lift}
For $d\in\mathcal P_V$ and nonempty $X,Y\subseteq V$, the Hausdorff lift of $d$ is
\[
H_d(X,Y)=\max\Bigl\{\ \sup_{x\in X}\inf_{y\in Y}d(x,y),\quad \sup_{y\in Y}\inf_{x\in X}d(x,y)\ \Bigr\}.
\]
\end{definition}

Since $u\in A_V(s)$ forces $T(s,u)\subseteq V$, the lift is always evaluated on subsets of $V$.  The Hausdorff lift is the metric counterpart of equality of successor classes: when $d$ is the discrete metric of an equivalence $\sim$ (zero within classes, $D$ across), $H_d(X,Y)=0$ holds exactly when $[X]_\sim=[Y]_\sim$.

\begin{definition}[Choice gap and distance transformer]\label{def:distance-transformer}
For $s,t\in V$, $u\in A_V(s)$, and $v\in A_V(t)$, the choice gap is
\[
\delta_d(s,u;t,v)=\bigl|c(s,u)-c(t,v)\bigr|+\gamma\,H_d\bigl(T(s,u),T(t,v)\bigr).
\]
The distance transformer $\mathcal D_V:\mathcal P_V\to\mathcal P_V$ is
\[
(\mathcal D_V d)(s,t)=
\begin{cases}
D, & L(s)\neq L(t),\\[4pt]
\displaystyle\max\Bigl\{\ \sup_{u\in A_V(s)}\inf_{v\in A_V(t)}\delta_d(s,u;t,v),\ \ \sup_{v\in A_V(t)}\inf_{u\in A_V(s)}\delta_d(s,u;t,v)\ \Bigr\}, & L(s)=L(t).
\end{cases}
\]
The suprema and infima are over finite nonempty sets, hence attained.
\end{definition}

The transformer mirrors the Bellman--Kripke recursion: the immediate term is the cost discrepancy, undiscounted; the continuation term is the discounted lifted distance between successor sets; and the back-and-forth maximum is the choice-matching discipline of value-refined modal bisimulation, now read quantitatively.  The label clause keeps states with different atomic truth at the maximal distance $D$, which is what lets the value bound below hold without separate case analysis.

\begin{proposition}[Well-posed distance fixed point]\label{prop:metric-fixed-point}
$\mathcal D_V$ maps $\mathcal P_V$ into $\mathcal P_V$ and is a $\gamma$-contraction in $\|\cdot\|_\infty$.  Hence it has a unique fixed point $d^\star\in\mathcal P_V$, and the iterates $d_0\in\mathcal P_V$, $d_{n+1}=\mathcal D_V d_n$ converge uniformly to $d^\star$ from every start.
\end{proposition}

\begin{proof}
\emph{Range.}  The matching term is bounded by $\max_{u,v}\delta_d(s,u;t,v)\le\bar c+\gamma D=\bar c+\gamma\bar c/(1-\gamma)=\bar c/(1-\gamma)=D$, and the label clause equals $D$; thus $\mathcal D_V d\le D$.  Symmetry of $\mathcal D_V d$ is immediate, and $(\mathcal D_V d)(s,s)=0$ by choosing $v=u$ in each direction, which makes $\delta_d=0$.  For the triangle inequality fix $s,t,w\in V$.  If $L(s)\neq L(w)$, then $L(s)=L(t)$ and $L(t)=L(w)$ cannot both hold, since they would force $L(s)=L(w)$; hence at least one of the pairs $\{s,t\},\{t,w\}$ has unequal labels, the right-hand side is at least $D$, and the left-hand side is at most $D$.  If $L(s)=L(w)$ but $L(t)$ differs from this common label, both $\{s,t\}$ and $\{t,w\}$ have unequal labels, so the right-hand side is at least $2D$.  In the remaining case all three labels are equal.  Then $\delta_d$ is a pseudometric on the disjoint union $\bigsqcup_{x\in V}A_V(x)$: it is nonnegative, vanishes on a choice paired with itself, is symmetric, and satisfies $\delta_d(s,u;w,x)\le\delta_d(s,u;t,y)+\delta_d(t,y;w,x)$ because the cost discrepancy and the Hausdorff lift each obey the triangle inequality and a nonnegative sum of pseudometrics is a pseudometric.  The equal-label value of $\mathcal D_V d$ is exactly the Hausdorff distance induced by $\delta_d$ between the finite sets $A_V(s)$ and $A_V(t)$, and the Hausdorff distance of a pseudometric is a pseudometric; the triangle inequality follows.  Therefore $\mathcal D_V d\in\mathcal P_V$.

\emph{Contraction.}  The only dependence of $\mathcal D_V d$ on $d$ is through $\gamma H_d$ in the equal-label branch; the label branch is constant.  By the maximum and minimum perturbation lemmas, $|H_d(X,Y)-H_{d'}(X,Y)|\le\|d-d'\|_\infty$ for all $X,Y$, and the outer suprema, infima, and maximum are nonexpansive in the sup norm.  Hence $|(\mathcal D_V d)(s,t)-(\mathcal D_V d')(s,t)|\le\gamma\|d-d'\|_\infty$ for every pair, so $\|\mathcal D_V d-\mathcal D_V d'\|_\infty\le\gamma\|d-d'\|_\infty$.  As $\mathcal P_V$ is complete, Banach's fixed-point theorem gives the unique $d^\star$ and uniform convergence of the iterates.
\end{proof}

\begin{definition}[Value-refined modal bisimulation metric]\label{def:vr-metric}
The value-refined modal bisimulation metric of $M$ on $V$ is the fixed point $d^\star=\mathcal D_V d^\star$ of Proposition~\ref{prop:metric-fixed-point}.
\end{definition}

The metric is a quantitative refinement of the exact equivalence rather than a new layer: its zero set is precisely the value-refined modal bisimulation, and it is computed only over certified choices.

\begin{theorem}[Zero set is the exact bisimulation]\label{thm:metric-kernel}
Let $\sim$ be the coarsest value-refined modal bisimulation; on $V=\Viab(G)$ its certified restriction matches choices in $A_V$ by the admissible matching closure lemma.  Then for all $s,t\in V$,
\[
d^\star(s,t)=0 \iff s\sim t.
\]
Consequently $d^\star$ grades the quotient of Section~\ref{sec:value-refined-bisimulation} and collapses to it at distance zero.
\end{theorem}

\begin{proof}
$(\Leftarrow)$  Let $\mathcal Z=\{d\in\mathcal P_V:\ d(s,t)=0\ \text{whenever}\ s\sim t\}$.  This set is nonempty (the discrete metric of $\sim$ scaled to $D$ lies in it), closed in $\|\cdot\|_\infty$, and $\mathcal D_V$-invariant: if $d\in\mathcal Z$ and $s\sim t$, then $L(s)=L(t)$ and, for each $u\in A_V(s)$, the bisimulation supplies $v\in A_V(t)$ with $c(s,u)=c(t,v)$ and $[T(s,u)]_\sim=[T(t,v)]_\sim$.  Since $d$ vanishes within $\sim$-classes, every successor on one side is at $d$-distance $0$ from some successor on the other, so $H_d(T(s,u),T(t,v))=0$ and $\delta_d(s,u;t,v)=0$, whence both directional suprema vanish and $(\mathcal D_V d)(s,t)=0$.  A contraction restricted to a nonempty closed invariant set has its unique fixed point inside that set, so $d^\star\in\mathcal Z$; that is, $s\sim t$ implies $d^\star(s,t)=0$.

$(\Rightarrow)$  The zero set $Z^\star=\{(s,t)\in V\times V:\ d^\star(s,t)=0\}$ of a pseudometric is an equivalence relation.  From $d^\star=\mathcal D_V d^\star$, $d^\star(s,t)=0$ forces $L(s)=L(t)$, since otherwise the value is $D>0$, and, for each $u\in A_V(s)$, the existence of $v\in A_V(t)$ with $\delta_{d^\star}(s,u;t,v)=0$, that is $c(s,u)=c(t,v)$ and $H_{d^\star}(T(s,u),T(t,v))=0$.  The vanishing Hausdorff lift says every successor of one choice is $d^\star$-distance $0$ from---hence $Z^\star$-equivalent to---some successor of the matched choice, so $[T(s,u)]_{Z^\star}=[T(t,v)]_{Z^\star}$; the symmetric direction holds by the symmetric branch.  Thus $Z^\star$ satisfies the clauses of value-refined modal bisimulation over $A_V$, so it is one such bisimulation and $Z^\star\subseteq\sim$ by coarsest-ness.  With $(\Leftarrow)$ giving $\sim\subseteq Z^\star$, the two relations coincide.
\end{proof}

\begin{theorem}[Optimal value is $1$-Lipschitz in the metric]\label{thm:metric-lipschitz}
For all $s,t\in V$,
\[
\bigl|J^\star(s)-J^\star(t)\bigr|\le d^\star(s,t).
\]
In particular $d^\star(s,t)=0$ implies $J^\star(s)=J^\star(t)$, which recovers the optimal-value quotient invariance of Section~\ref{sec:value-refined-bisimulation}.
\end{theorem}

\begin{proof}
Let $J_0\equiv0$ and $J_{n+1}=\mathcal B_V J_n$; by Theorem~\ref{thm:core-value} the iterates converge uniformly to $J^\star$, and each $J_n$ takes values in $[0,D]$.  We prove $|J_n(s)-J_n(t)|\le d^\star(s,t)$ for all $n$ and all $s,t\in V$ by induction, then pass to the limit.

The induction hypothesis at level $n$ yields a Hausdorff bound on the inner maxima: if $f:V\to\mathbb R$ satisfies $|f(x)-f(y)|\le d^\star(x,y)$ for all $x,y$, then for nonempty $X,Y\subseteq V$,
\[
\Bigl|\max_{x\in X}f(x)-\max_{y\in Y}f(y)\Bigr|\le H_{d^\star}(X,Y).
\]
Indeed, let $x^\ast$ attain $\max_X f$ and pick $y\in Y$ with $d^\star(x^\ast,y)\le\sup_{x\in X}\inf_{y'\in Y}d^\star(x,y')\le H_{d^\star}(X,Y)$; then $f(x^\ast)\le f(y)+d^\star(x^\ast,y)\le\max_Y f+H_{d^\star}(X,Y)$, and the symmetric estimate completes the bound.

\emph{Base.}  $J_0\equiv0$ gives $|J_0(s)-J_0(t)|=0\le d^\star(s,t)$.

\emph{Step.}  Assume the bound at level $n$.  Fix $s,t\in V$.  If $L(s)\neq L(t)$ then $d^\star(s,t)=D\ge\|J_{n+1}\|_\infty\ge|J_{n+1}(s)-J_{n+1}(t)|$.  If $L(s)=L(t)$, it suffices by symmetry to bound $J_{n+1}(s)-J_{n+1}(t)$.  Let $v\in A_V(t)$ attain $J_{n+1}(t)=\min_{v'}\bigl[c(t,v')+\gamma\max_{T(t,v')}J_n\bigr]$.  The direction $\sup_{v'}\inf_{u'}$ of $\mathcal D_V$ evaluated at the fixed point supplies $u\in A_V(s)$ with $\delta_{d^\star}(s,u;t,v)\le(\mathcal D_V d^\star)(s,t)=d^\star(s,t)$.  Then
\begin{align*}
J_{n+1}(s)
&\le c(s,u)+\gamma\max_{s'\in T(s,u)}J_n(s')\\
&\le c(t,v)+\bigl|c(s,u)-c(t,v)\bigr|+\gamma\Bigl(\max_{t'\in T(t,v)}J_n(t')+H_{d^\star}\bigl(T(s,u),T(t,v)\bigr)\Bigr)\\
&=\Bigl[c(t,v)+\gamma\max_{t'\in T(t,v)}J_n(t')\Bigr]+\delta_{d^\star}(s,u;t,v)\\
&\le J_{n+1}(t)+d^\star(s,t),
\end{align*}
using the Hausdorff bound with $f=J_n$ in the second line and minimality of $v$ in the last.  The symmetric direction gives $J_{n+1}(t)-J_{n+1}(s)\le d^\star(s,t)$.  Letting $n\to\infty$ yields $|J^\star(s)-J^\star(t)|\le d^\star(s,t)$.
\end{proof}

\begin{corollary}[Graded quotient and approximate value preservation]\label{cor:eps-quotient}
Let $\varepsilon\ge0$ and let $R\subseteq V\times V$ be any relation with $d^\star(s,t)\le\varepsilon$ whenever $sRt$.  Then $|J^\star(s)-J^\star(t)|\le\varepsilon$ for every such pair; the optimal value descends to the $\varepsilon$-quotient with uniform error at most $\varepsilon$.  The exact value-preserving quotient of Section~\ref{sec:value-refined-bisimulation} is the case $\varepsilon=0$, and merging states within a $d^\star$-ball of radius $\varepsilon$ perturbs the recoverable optimal value by at most $\varepsilon$.
\end{corollary}

\begin{proof}
Immediate from Theorem~\ref{thm:metric-lipschitz}.
\end{proof}

\begin{example}[Tightness on the cost-separated pair]\label{ex:metric-tight}
Reprise the cost-blind-merge structure: $S=\{p,p',q\}$ with one proposition true everywhere, $T(p,a)=T(p',a)=\{q\}$, $T(q,a)=\{q\}$, $c(p,a)=c(p',a)=2$, $c(q,a)=0$, and $\gamma=1/2$, so $\bar c=2$ and $D=\bar c/(1-\gamma)=4$.  The single choices and equal labels give $d^\star(p,p')=0$, confirming that the metric's zero set merges $p$ and $p'$.  For the cost-separated pair,
\[
\delta_{d^\star}(p,a;q,a)=|2-0|+\tfrac12 H_{d^\star}(\{q\},\{q\})=2+\tfrac12\,d^\star(q,q)=2,
\]
and since $d^\star(q,q)=0$ the fixed-point equation gives $d^\star(p,q)=2$.  The optimal values are $J^\star(p)=2$ and $J^\star(q)=0$, so $|J^\star(p)-J^\star(q)|=2=d^\star(p,q)$ and the Lipschitz bound is tight.  The cost-blind bisimilarity $\approx_R$ that wrongly merges all three states corresponds to collapsing this distance-$2$ pair to distance $0$, which is exactly the inexact compression the bound forbids.
\end{example}

\begin{remark}[What the lift fixes and what stays a modeling choice]\label{rem:metric-modeling}
Two features of the construction are forced by the semantics rather than chosen.  First, the lift is Hausdorff because compatible successors are resolved universally: the value layer takes the worst compatible successor, and the worst-case set distance is the Hausdorff distance.  The probabilistic refinement of Section~\ref{sec:risk-dial}, in which the successor is resolved by a distribution at a dial setting between nominal and universal, is exactly where a Kantorovich (Wasserstein) lift would replace the Hausdorff lift, recovering the Markov-decision-process bisimulation-metric form of Ferns--Panangaden--Precup \cite{ferns2004,desharnais2004,vanbreugel2005}.  Second, the relative weight of the two terms---coefficient $1$ on the cost discrepancy and $\gamma$ on the lifted continuation---is not a free parameter: it is the unique scaling for which $J^\star$ is $1$-Lipschitz, and any other weighting would either break the bound or loosen it.  The one genuinely application-dependent quantity, deliberately left unfixed here, is the rule that assigns representative labels, costs, and transitions to the classes of an $\varepsilon$-quotient when an aggregated structure is actually built: that aggregation scheme is a parameter of the model, not of the semantics, so it is not pinned to any particular numeric choice in this construction.
\end{remark}

\begin{remark}[The metric refines the quotient layer without reordering it]\label{rem:metric-layer}
The metric lives at the quotient layer and nowhere else.  It is defined on $V=\Viab(G)$, after admissibility, and over the certified choices $A_V$, after value refinement is available, so it cannot license a transition outside the kernel; $d^\star=0$ is the exact endpoint of Section~\ref{sec:value-refined-bisimulation}.  The graded quotient therefore sits at the same position in the order modal satisfaction $\to$ admissible continuation $\to$ value refinement $\to$ certificate, between value and certificate, and the non-commutation witnesses of Section~\ref{sec:worked-discipline} apply unchanged.  What the metric adds is purely quantitative: a bound, by the distance itself, on the optimal-value error charged by an inexact compression, with the exact value-preserving quotient as the distance-zero case.
\end{remark}

\section{Certified stopping certificates}

A selector can be tempted to stop when the value is small.  This is unsound because a low value is an value refinement statement, not a proof that a terminal condition has been reached.  The semantics therefore separates terminal certificates from values.

\begin{definition}[Certificate predicate]
A certificate predicate is a set $C\subseteq V$ such that membership in $C$ authorizes termination.  If $C=\Sat(\chi)\cap V$ for some formula $\chi$, the certificate is called modal.
\end{definition}

\begin{definition}[Sound stopping rule]
A stopping rule is a function $\tau:S^{\mathbb N}\to\mathbb N\cup\{\infty\}$.  It is sound for $C$ if $\tau(s_0,s_1,\ldots)=t<\infty$ implies $s_t\in C$.
\end{definition}

\begin{proposition}[Small value does not imply certification]
There exists an admissible system, a value function, and a certificate set $C$ such that a state has minimal value but is not certified.
\end{proposition}

\begin{proof}
Let $S=\{x,y\}$ and let both states be admissible.  Let $C=\{y\}$.  Give $x$ a self-loop with zero cost and give $y$ a self-loop with cost one.  Under discounting, the value at $x$ is zero and the value at $y$ is $1/(1-\gamma)$.  Thus $x$ has smaller value.  Since $x\notin C$, minimal value does not imply certification.
\end{proof}

The proposition shows what a certificate must not be reduced to; the rest of the section shows what it positively is.  Reaching a certificate while never leaving the justifiably maintainable region is a \emph{reach-while-stay} (reach--avoid) objective, the least-fixpoint companion of the greatest-fixpoint admissible-continuation kernel, and it is justifiably maintainable on a region computed by the same admissible predecessor used for the kernel.

\subsection{Reach-while-stay: the certified objective dual to invariance}

Fix a certificate set $C\subseteq V=\Viab(G)$.  Say a selector \emph{achieves reach-while-stay from $s$} if it has a strategy forcing every resulting play to remain in $G$ until it reaches $C$ and to reach $C$ after finitely many steps.

\begin{definition}[Reach-while-stay region]
The \emph{reach-while-stay region} of $C$ inside $G$ is the least fixed point
\[
\mathrm{RWS}_G(C)=\mu Y.\bigl(C\cup\bigl(G\cap\Pre_{\exists\Box}(Y)\bigr)\bigr).
\]
\end{definition}

The operator is monotone on $(2^{S},\subseteq)$, so the least fixed point exists and is reached by the ascending iteration $Y_0=\varnothing$, $Y_{k+1}=C\cup(G\cap\Pre_{\exists\Box}(Y_k))$.  Define the \emph{rank} of a state by $\rk(s)=\min\{k:s\in Y_k\}$ when $s\in\mathrm{RWS}_G(C)$.

\begin{theorem}[Reach-while-stay certification]\label{thm:rws}
Let $C\subseteq V=\Viab(G)$.  Then
\begin{enumerate}
\item $\mathrm{RWS}_G(C)\subseteq V$;
\item a memoryless selector achieves reach-while-stay from every $s\in\mathrm{RWS}_G(C)$, reaching $C$ within $\rk(s)\leq|\mathrm{RWS}_G(C)|$ steps while keeping the play in $G$;
\item conversely, a selector achieves reach-while-stay from $s$ only if $s\in\mathrm{RWS}_G(C)$.
\end{enumerate}
Hence the stopping rule ``halt on first entry to $C$'' is simultaneously sound for $C$ and live exactly on $\mathrm{RWS}_G(C)$, and nowhere else.
\end{theorem}

\begin{proof}
For (1), argue by induction that $Y_k\subseteq V$.  The base case is $Y_0=\varnothing\subseteq V$.  Assume $Y_k\subseteq V$.  Then $C\subseteq V$ by hypothesis, and any $s\in G\cap\Pre_{\exists\Box}(Y_k)$ has some $u\in A(s)$ with $T(s,u)\subseteq Y_k\subseteq V$, so $s\in G$ and $s\in\Pre_{\exists\Box}(V)$; since $V=G\cap\Pre_{\exists\Box}(V)$ is a fixed point this gives $s\in V$.  Thus $Y_{k+1}\subseteq V$ and the limit $\mathrm{RWS}_G(C)\subseteq V$.

For (2), define a memoryless choice $\rho$ on $\mathrm{RWS}_G(C)\setminus C$: a state $s$ there has $\rk(s)=k\geq1$ and $s\in G\cap\Pre_{\exists\Box}(Y_{k-1})$, so choose $\rho(s)=u$ with $T(s,u)\subseteq Y_{k-1}$.  Every successor then has rank at most $k-1$, so the rank strictly decreases each step until the play enters $C$; this happens within $\rk(s)$ steps, and $\rk(s)\leq|\mathrm{RWS}_G(C)|$ because distinct ranks index distinct nonempty differences of the ascending chain.  Each visited state before $C$ lies in $Y_k\setminus C\subseteq G$, and $C\subseteq V\subseteq G$, so the entire prefix stays in $G$.  This is reach-while-stay by a memoryless selector.

For (3), suppose a selector achieves reach-while-stay from $s$.  Consider the universal number of steps to reach $C$ over all universal successor choices; finiteness on every play and finite branching give, by König's lemma, a uniform bound $N$.  Show by induction on $N$ that $s\in\mathrm{RWS}_G(C)$.  If $N=0$ the play is already in $C\subseteq\mathrm{RWS}_G(C)$.  If $N\geq1$ then $s\in G$ (the play has not left $G$) and the selector's first move $u$ satisfies $T(s,u)\subseteq\{s':\text{universal distance from }s'\leq N-1\}$, which by the induction hypothesis is contained in $\mathrm{RWS}_G(C)$; hence $s\in G\cap\Pre_{\exists\Box}(\mathrm{RWS}_G(C))\subseteq\mathrm{RWS}_G(C)$ by the fixed-point property.  The final equivalence combines (2) and (3): liveness of ``halt on entry to $C$'' is exactly reachability of $C$ in finite time, and soundness is exactly that halts occur in $C$.
\end{proof}

\begin{remark}[Persistent certification is a co-B\"uchi nesting]
If the requirement is strengthened from reaching $C$ once to eventually remaining in a certified-and-admissible set forever --- a co-B\"uchi condition --- the justifiably maintainable region is the nested fixpoint
\[
\nu Z.\,\mu Y.\Bigl(\bigl(C\cap G\cap\Pre_{\exists\Box}(Z)\bigr)\ \cup\ \bigl(G\cap\Pre_{\exists\Box}(Y)\bigr)\Bigr),
\]
whose outer greatest fixed point selects states from which the selector can return to the certified set, and whose inner least fixed point forces each such return in finitely many admissible steps.  This is the point at which the temporal-objective product of Appendix~C is the appropriate construction: a co-B\"uchi or general $\omega$-regular certification objective is compiled into the arena, after which the present alternation of least and greatest admissible-predecessor fixpoints computes its justifiably maintainable region, and the discounted layer is then placed on that region exactly as before.
\end{remark}

\begin{theorem}[Certificate soundness theorem]
Let $V=\Viab(G)$ and let $C\subseteq V$.  If a selector only stops at states in $C$ and otherwise selects choices in $A_V$, then every finite stop is certified and every nonstopped prefix remains in $G$.
\end{theorem}

\begin{proof}
The first claim follows directly from the definition of the stopping rule.  For the second claim, suppose the selector has not stopped before time $t$.  At every earlier time it selected a choice in $A_V$.  By induction over time, the state remains in $V$.  Since $V\subseteq G$, the nonstopped prefix remains in $G$.
\end{proof}

This theorem formalizes the difference between progress and completion.  A value function may guide progress, but a certificate must justify completion.

\section{Switching among modal objectives}

Many choice-refined systems have several acceptable regions or several formulas.  Without a switching discipline, a selector can oscillate between objectives.  This is modeled by augmenting the state with a mode.

Let $G_1,\ldots,G_m\subseteq S$ be target regions with admissible kernels $V_i=\Viab(G_i)$.  Let $J_i^\star$ be the admissible value associated with $G_i$ on $V_i$.  Let $i_t\in\{1,\ldots,m\}$ be the active mode.

\begin{definition}[Switching value]
For a switching penalty $\lambda>0$, define
\[
W(s,i)=\min_{j:s\in V_j}\left[J_j^\star(s)+\lambda\mathbf 1_{j\neq i}\right].
\]
\end{definition}

The penalty term is a discrete hysteresis term.  Hysteresis is standard in choice as a way to prevent chattering, and here it prevents objective drift.

\begin{definition}[Mode margin]
A state $s$ has margin $\Delta>0$ for mode $i$ if $s\in V_i$ and
\[
J_i^\star(s)+\Delta\leq J_j^\star(s)
\]
for every $j\neq i$ with $s\in V_j$.
\end{definition}

\begin{theorem}[No-switch under margin]
If $s$ has margin $\Delta>0$ for mode $i$ and $\lambda>0$, then switching from $i$ to any $j\neq i$ is strictly worse at $s$.
\end{theorem}

\begin{proof}
For each feasible $j\neq i$, the margin condition gives $J_j^\star(s)\geq J_i^\star(s)+\Delta$.  Therefore
\[
J_j^\star(s)+\lambda\geq J_i^\star(s)+\Delta+\lambda.
\]
Since $\Delta+
\lambda>0$, the right-hand side is strictly larger than $J_i^\star(s)$.  Thus staying in mode $i$ is strictly better.
\end{proof}

\begin{theorem}[Finite switching bound]
Assume every switch costs at least $\lambda>0$ and the realized total discounted cost is bounded above by $M<\infty$.  Then the number of switches is at most $\lfloor M/\lambda\rfloor$.
\end{theorem}

\begin{proof}
If a path contains $N$ switches, the switching component alone contributes at least $N\lambda$.  Since the total cost is at most $M$, we have $N\lambda\leq M$.  Dividing by $\lambda$ gives $N\leq M/\lambda$.  Since $N$ is an integer, $N\leq \lfloor M/\lambda\rfloor$.
\end{proof}

The theorem shows that commitment is not a psychological assumption.  It is a mathematical consequence of positive switching cost plus bounded budget.

\begin{remark}[Relation to switched-systems theory]
Positive switching charge gives a finite descent measure for objective changes.  The construction runs that discipline on top of the admissible continuation/value separation: switching is between precomputed admissible kernels and is scored by the admissible values $J_i^\star$, so each mode change is both admissible by construction and accounted against the same discounted budget as the ordinary value cost.
\end{remark}

\section{Bounded deliberation and residual risk}

A transition selector often cannot compute the infinite-horizon value exactly before acting.  The question is how to act when only a bounded number of Bellman updates has been performed.  This section gives a clean residual certificate.

Let $J_0=0$ and $J_{n+1}=\mathcal B_VJ_n$.  Define the residual
\[
r_n=\|J_{n+1}-J_n\|_\infty.
\]

\begin{lemma}[A posteriori value error]
For all $n\geq0$,
\[
\|J_n-J^\star\|_\infty\leq \frac{r_n}{1-\gamma}.
\]
\end{lemma}

\begin{proof}
For $m>n$, write the telescoping sum
\[
J_m-J_n=\sum_{k=n}^{m-1}(J_{k+1}-J_k).
\]
Taking norms gives
\[
\|J_m-J_n\|_\infty\leq \sum_{k=n}^{m-1}\|J_{k+1}-J_k\|_\infty.
\]
Since $J_{k+1}=\mathcal B_VJ_k$ and $\mathcal B_V$ is a $\gamma$-contraction,
\[
\|J_{k+1}-J_k\|_\infty\leq \gamma^{k-n}\|J_{n+1}-J_n\|_\infty=\gamma^{k-n}r_n.
\]
Thus
\[
\|J_m-J_n\|_\infty\leq r_n\sum_{k=n}^{m-1}\gamma^{k-n}
\leq \frac{r_n}{1-\gamma}.
\]
Letting $m\to\infty$ and using $J_m\to J^\star$ proves the estimate.
\end{proof}

\begin{lemma}[A priori geometric rate]
For all $n\geq 0$,
\[
\|J_n-J^\star\|_\infty\leq \gamma^n\,\|J_0-J^\star\|_\infty\leq \frac{\gamma^n r_0}{1-\gamma}.
\]
\end{lemma}

\begin{proof}
Because $J^\star=\mathcal B_VJ^\star$ and $\mathcal B_V$ is a $\gamma$-contraction in $\|\cdot\|_\infty$,
\[
\|J_n-J^\star\|_\infty=\|\mathcal B_V^{\,n}J_0-\mathcal B_V^{\,n}J^\star\|_\infty\leq\gamma^n\|J_0-J^\star\|_\infty.
\]
The a posteriori lemma at $n=0$ gives $\|J_0-J^\star\|_\infty\leq r_0/(1-\gamma)$, which yields the second inequality.  The a posteriori bound is computable from two successive iterates, while this a priori bound exposes the geometric rate in closed form.
\end{proof}

\begin{lemma}[Monotone convergence and two-sided enclosure]
If $J_0\leq\mathcal B_VJ_0$ pointwise, then the iterates increase monotonically, $J_0\leq J_1\leq\cdots$, and $J_n\uparrow J^\star$.  If $J_0\geq\mathcal B_VJ_0$, then $J_0\geq J_1\geq\cdots$ and $J_n\downarrow J^\star$.  Consequently, if $\underline J_0\leq\mathcal B_V\underline J_0$ and $\overline J_0\geq\mathcal B_V\overline J_0$, the iterates $\underline J_n,\overline J_n$ satisfy
\[
\underline J_n\leq J^\star\leq\overline J_n
\qquad\text{and}\qquad
\|\overline J_n-\underline J_n\|_\infty\leq\gamma^n\,\|\overline J_0-\underline J_0\|_\infty,
\]
so the enclosing interval contracts at the geometric rate $\gamma^n$.
\end{lemma}

\begin{proof}
Assume $J_0\leq\mathcal B_VJ_0=J_1$.  The Bellman--Kripke map is monotone, so applying it to $J_0\leq J_1$ gives $J_1\leq J_2$, and by induction $J_n\leq J_{n+1}$; the sequence is nondecreasing.  The sub- and supersolution sandwich gives $J_0\leq J^\star$, and applying the monotone map to $J_n\leq J^\star$ together with $\mathcal B_VJ^\star=J^\star$ yields $J_{n+1}\leq J^\star$, so the increasing sequence is bounded above by $J^\star$.  On the finite state space it therefore converges, and since $\mathcal B_V$ is continuous its limit is a fixed point, hence equal to the unique fixed point $J^\star$.  The supersolution case is dual.  For the enclosure, $\underline J_n\leq J^\star\leq\overline J_n$ holds at every $n$ by the two monotone cases, and
\[
\|\overline J_n-\underline J_n\|_\infty=\|\mathcal B_V^{\,n}\overline J_0-\mathcal B_V^{\,n}\underline J_0\|_\infty\leq\gamma^n\|\overline J_0-\underline J_0\|_\infty
\]
by contraction.
\end{proof}

\begin{definition}[Certified approximate choice]
A choice $u\in A_V(s)$ is $\varepsilon$-certified at depth $n$ if
\[
c(s,u)+\gamma\max_{s'\in T(s,u)}J_n(s')
\leq
\min_{v\in A_V(s)}\left[c(s,v)+\gamma\max_{s'\in T(s,v)}J_n(s')\right]+\varepsilon.
\]
\end{definition}

\begin{theorem}[Bounded-deliberation invariance and regret]
If $u$ is $\varepsilon$-certified at state $s\in V$ using $J_n$, then $u$ is admissible and its one-step optimality gap relative to $J^\star$ is at most $\varepsilon+2\gamma r_n/(1-\gamma)$.
\end{theorem}

\begin{proof}
Invariance follows from $u\in A_V(s)$.  Let
\[
Q_n(s,u)=c(s,u)+\gamma\max_{s'\in T(s,u)}J_n(s')
\]
and define $Q^\star$ analogously with $J^\star$.  The residual lemma gives
\[
|Q_n(s,u)-Q^\star(s,u)|\leq \gamma r_n/(1-\gamma)
\]
for every certified choice $u$.  Let $v^\star$ minimize $Q^\star(s,v)$.  Since $u$ is $\varepsilon$-certified for $Q_n$, we have $Q_n(s,u)\leq Q_n(s,v^\star)+\varepsilon$.  Combining the two error bounds gives
\[
Q^\star(s,u)\leq Q^\star(s,v^\star)+\varepsilon+2\gamma r_n/(1-\gamma).
\]
This is the claimed gap.
\end{proof}

The theorem permits a selector to act before exact convergence, but it forces the choice to remain inside the certified choice set.  Approximation is allowed in value refinement.  Approximation is not allowed in the logical invariance constraint.

\section{Semantic construction pipeline}\label{sec:pipeline}

The finite theory induces two concrete semantic construction procedures.  The full-information procedure computes a state admissible-continuation kernel, a largest certified domain, an optional quotient, and a discounted universal value.  The partial-observation procedure first constructs the reachable information state game and then runs the same kernel-and-value pattern on information states.

\begin{definition}[Admissible-continuation domain map]
Given $V\subseteq S$, the admissible-continuation constraint is
\[
M_V(s)=A_V(s)=\{u\in A(s):T(s,u)\subseteq V\}.
\]
\end{definition}

\begin{algorithm}[H]
\caption{Full-information value-refined modal semantics}
\begin{algorithmic}[1]
\Require finite structure $(S,U,A,T,L,c)$, target formula $\varphi$, discount $\gamma\in(0,1)$, tolerance $\varepsilon>0$
\Ensure admissible region $V$, largest certified domain $M_V$, value $J$, certified greedy selector $\pi$
\State compute $G\gets\Sat(\varphi)$ by finite model checking
\State $X\gets S$
\Repeat
\State $X_{prev}\gets X$
\State $X\gets G\cap\{s\in S: \exists u\in A(s)\text{ with }T(s,u)\subseteq X_{prev}\}$
\Until{$X=X_{prev}$}
\State $V\gets X$
\For{$s\in V$}
\State $M_V(s)\gets\{u\in A(s):T(s,u)\subseteq V\}$
\EndFor
\State initialize $J_0\gets 0$ on $V$ and $k\gets0$
\Repeat
\State $J_{k+1}(s)\gets\min_{u\in M_V(s)}\bigl(c(s,u)+\gamma\max_{s'\in T(s,u)}J_k(s')\bigr)$ for all $s\in V$
\State $k\gets k+1$
\Until{$\|J_k-J_{k-1}\|_\infty\leq(1-\gamma)\varepsilon/(2\gamma)$}
\State choose $\pi(s)\in\argmin_{u\in M_V(s)}\bigl(c(s,u)+\gamma\max_{s'\in T(s,u)}J_k(s')\bigr)$
\State \Return $V,M_V,J_k,\pi$
\end{algorithmic}
\end{algorithm}

The input is a finite modal choice structure and a modal target.  The output is not merely a selector; it is the region on which the specification is realizable, the largest admissible set of choices on that region, and a value-ranked certified selector.  The stopping test uses the standard residual bound for a $\gamma$-contraction.

\begin{algorithm}[H]
\caption{Information-state value-refined modal semantics}
\begin{algorithmic}[1]
\Require finite structure $(S,U,A,T,L,c)$, observation map $\obs:S\to O$, target formula $\varphi$, discount $\gamma\in(0,1)$
\Ensure observational certified information states $W$, information-state domain map $M_W$, information state value $J_{\obs}$, certified observation selector $\pi_{\obs}$
\State compute $G\gets\Sat(\varphi)$
\State construct reachable information states using $b_0\subseteq S$ and the update $\bpost(b,u,o')$
\State $\mathcal X\gets\{b:b\subseteq G,\ b\neq\varnothing,\ b\text{ lies in one observation fibre}\}$ restricted to reachable information states
\Repeat
\State $\mathcal X_{prev}\gets\mathcal X$
\State $\mathcal X\gets\{b\in\mathcal X_{prev}:\exists u\in A(b)\text{ with }\Post_B(b,u)\subseteq\mathcal X_{prev}\}$
\Until{$\mathcal X=\mathcal X_{prev}$}
\State $W\gets\mathcal X$
\State $M_W(b)\gets\{u\in A(b):\Post_B(b,u)\subseteq W\}$ for each $b\in W$
\State solve $J_{\obs}=\mathcal B_W^{\obs}J_{\obs}$ by value iteration or selector iteration
\State choose a Bellman-minimizing selector $\pi_{\obs}$ over $M_W$
\State \Return $W,M_W,J_{\obs},\pi_{\obs}$
\end{algorithmic}
\end{algorithm}

\begin{proposition}[Algorithmic bounds]\label{prop:algorithmic-bounds}
Let
\[
n=|S|,
\qquad m=\sum_{s\in S}|A(s)|,
\qquad e=\sum_{s\in S}\sum_{u\in A(s)}|T(s,u)|.
\]
The naive full-information admissible continuation computation terminates in at most $n$ rounds and costs $O(ne)$ time.  One value-iteration sweep costs $O(e)$ over the admissible region.  If $B$ reachable information states and $e_B=\sum_{b,u}|\Post_B(b,u)|$ information state incidences are generated, the naive observational admissible continuation computation costs $O(Be_B)$ and each information state-value sweep costs $O(e_B)$.
\end{proposition}

\begin{proof}
Each strict full-information admissible continuation round removes at least one state, so there are at most $n$ strict rounds, and a naive round inspects all successor incidences.  A value sweep computes one maximum over each admissible successor incidence.  The information state proof is the same argument on the finite information state graph: each strict round removes at least one reachable information state, and each round inspects the information state-incidence relation.  The Bellman sweep over information states computes one maximum over each information state successor incidence.
\end{proof}

\begin{proposition}[Domain-map maximality]\label{prop:filter-maximality}
The domain map $M_V$ is largest among all state-based constraints that preserve $G$ forever from $V$, and $M_W$ is largest among all observation-state-based constraints that preserve $G$ forever from $W$.
\end{proposition}

\begin{proof}
Every allowed choice keeps every successor inside the corresponding fixed point, so the domain maps are admissible.  If a state domain map allows a choice outside $M_V(s)$, some successor leaves $V$, and from that successor the invariance-game theorem says no selector can guarantee $G$ forever.  The information state argument is identical with successor information states replacing successor states and Theorem~\ref{thm:obs-invariance} replacing the state justification theorem.
\end{proof}

\section{Sensitivity and stability of the semantic construction}

The semantic construction depends on three pieces of data: the successor relation, the cost function, and the discount factor.  This section quantifies how the admissible region and the optimal value respond to perturbations of each.  Throughout, $\|\cdot\|_\infty$ is the supremum norm over the relevant state set, and all costs are nonnegative with $\bar c=\max_{s,u}c(s,u)$.  Write $\Viab_T$ and $J^\star_T$ to expose the dependence on a successor map $T$.

\begin{theorem}[The kernel shrinks under universal strengthening]
Let $T$ and $T'$ be successor maps on the same $(S,A)$ with $T(s,u)\subseteq T'(s,u)$ for all $s,u$.  Then $\Viab_{T'}(G)\subseteq\Viab_T(G)$ for every target $G$.
\end{theorem}

\begin{proof}
For any $X\subseteq S$, if $T'(s,u)\subseteq X$ then $T(s,u)\subseteq T'(s,u)\subseteq X$, so the admissible predecessor satisfies $\Pre^{T'}_{\exists\Box}(X)\subseteq\Pre^{T}_{\exists\Box}(X)$.  Hence the admissible continuation transformers obey $F^{T'}_G(X)=G\cap\Pre^{T'}_{\exists\Box}(X)\subseteq G\cap\Pre^{T}_{\exists\Box}(X)=F^{T}_G(X)$ pointwise.  Both are monotone, so $\Viab_{T'}(G)=F^{T'}_G(\Viab_{T'}(G))\subseteq F^{T}_G(\Viab_{T'}(G))$, i.e. $\Viab_{T'}(G)$ is a post-fixed point of $F^{T}_G$.  By Knaster--Tarski maximality of the greatest fixed point, $\Viab_{T'}(G)\subseteq\Viab_T(G)$.
\end{proof}

\begin{theorem}[Value is monotone in successor enlargement]
Fix a region $V$ that is admissible under both $T$ and $T'$ with a common certified choice set $A_V(s)=\{u\in A(s):T(s,u)\subseteq V\}=\{u\in A(s):T'(s,u)\subseteq V\}$, and suppose $T(s,u)\subseteq T'(s,u)\subseteq V$ for all $s\in V$ and $u\in A_V(s)$.  Then $J^\star_T\leq J^\star_{T'}$ on $V$.
\end{theorem}

\begin{proof}
On $\ell_\infty(V)$ the two restricted Bellman operators share the choice set $A_V(s)$ and differ only in the successor range.  For any $J$ and any $s\in V$, each choice's evaluation obeys $\max_{s'\in T(s,u)}J(s')\leq\max_{s'\in T'(s,u)}J(s')$ because $T(s,u)\subseteq T'(s,u)$; taking the same minimum over $A_V(s)$ gives $\mathcal B^{T}_VJ(s)\leq\mathcal B^{T'}_VJ(s)$.  Thus $\mathcal B^{T}_V\leq\mathcal B^{T'}_V$ pointwise.  Applying this at $J=J^\star_T$ and using $J^\star_T=\mathcal B^{T}_VJ^\star_T$ gives $J^\star_T\leq\mathcal B^{T'}_VJ^\star_T$, so $J^\star_T$ is a subsolution of $\mathcal B^{T'}_V$.  By the sub- and supersolution sandwich, $J^\star_T\leq J^\star_{T'}$.
\end{proof}

\begin{theorem}[Lipschitz dependence on cost]
Let $c$ and $\hat c$ be cost functions on the same system, with admissible optimal values $J^\star_c$ and $J^\star_{\hat c}$ on a common admissible region $V$.  Then
\[
\|J^\star_c-J^\star_{\hat c}\|_\infty\leq\frac{\|c-\hat c\|_\infty}{1-\gamma}.
\]
\end{theorem}

\begin{proof}
For fixed $J$ and $s\in V$, the two operators differ only through the additive cost term, so by the minimum perturbation lemma $|\mathcal B_V^{c}J(s)-\mathcal B_V^{\hat c}J(s)|\leq\max_{u\in A_V(s)}|c(s,u)-\hat c(s,u)|\leq\|c-\hat c\|_\infty$.  Hence
\begin{align*}
\|J^\star_c-J^\star_{\hat c}\|_\infty
&=\|\mathcal B_V^{c}J^\star_c-\mathcal B_V^{\hat c}J^\star_{\hat c}\|_\infty\\
&\leq\|\mathcal B_V^{c}J^\star_c-\mathcal B_V^{c}J^\star_{\hat c}\|_\infty
+\|\mathcal B_V^{c}J^\star_{\hat c}-\mathcal B_V^{\hat c}J^\star_{\hat c}\|_\infty\\
&\leq\gamma\|J^\star_c-J^\star_{\hat c}\|_\infty+\|c-\hat c\|_\infty.
\end{align*}
Rearranging gives the bound.
\end{proof}

\begin{theorem}[Lipschitz dependence on the discount factor]
The admissible continuation kernel does not depend on the discount factor, and for $\gamma,\gamma'\in[0,1)$ on the common admissible region,
\[
\|J^\star_\gamma-J^\star_{\gamma'}\|_\infty\leq\frac{|\gamma-\gamma'|\,\bar c}{(1-\gamma)(1-\gamma')}.
\]
\end{theorem}

\begin{proof}
The kernel is defined from $\Pre_{\exists\Box}$, which involves neither $\gamma$ nor $c$, so $V$ is common to both discount factors.  The value $J^\star_{\gamma'}$ is a discounted supremum of nonnegative per-step costs bounded by $\bar c$, hence $\|J^\star_{\gamma'}\|_\infty\leq\bar c/(1-\gamma')$.  For fixed $J$ and $s$, the operators differ only in the discount multiplier, so $|\mathcal B_V^{\gamma}J(s)-\mathcal B_V^{\gamma'}J(s)|\leq|\gamma-\gamma'|\max_{u}\max_{s'\in T(s,u)}|J(s')|\leq|\gamma-\gamma'|\,\|J\|_\infty$.  Therefore
\[
\|J^\star_\gamma-J^\star_{\gamma'}\|_\infty
\leq\gamma\|J^\star_\gamma-J^\star_{\gamma'}\|_\infty+|\gamma-\gamma'|\,\|J^\star_{\gamma'}\|_\infty
\leq\gamma\|J^\star_\gamma-J^\star_{\gamma'}\|_\infty+\frac{|\gamma-\gamma'|\,\bar c}{1-\gamma'},
\]
and rearranging yields the stated estimate.
\end{proof}

\begin{remark}
The two kinds of stability differ in character.  Invariance is universal at the level of sets: enlarging the resolver's successor choices can only shrink the admissible region, monotonically and without quantitative loss.  Value is universal in the quantitative sense: it varies Lipschitz-continuously in both the cost data and the discount factor, with constants that degrade as $\gamma\to1$.  Invariance has no such blow-up because it is a purely relational property.
\end{remark}

The next boundary fragments are intentionally not part of the finite exact endpoint.  They record how the layer order can be preserved when one assumption is relaxed, and they are used only by Theorem~\ref{thm:relaxed-main}.  The main finite semantics is unchanged by them: none changes the definitions of $K_G$, certified choices, value-refined bisimulation, or residual certification.

\section{Main value-refined semantic theorem}

The preceding pieces now assemble into the central theorem.  This theorem is the finite exact endpoint of the construction.  The later relaxation theorem is separated from it so that the finite result is not conflated with extensions.

\begin{theorem}[Value-refined modal semantic construction]\label{thm:main-semantic-synthesis}
Let $\mathcal T=(S,U,A,T,L,c)$ be a finite modal choice structure.  Let $G\subseteq S$ be modal-definable and let $\gamma\in(0,1)$.  In the full-information endpoint, if $V=\Viab(G)$ is nonempty, then:
\begin{enumerate}[label=(\roman*)]
\item $V$ is the memoryless admissible region of the invariance game and the maximal choice-refined-invariant subset of $G$;
\item $M_V(s)=A_V(s)$ is the largest certified-choice domain map on $V$;
\item the restricted Bellman operator has a unique fixed point $J^\star$ and memoryless optimal certified selectors;
\item proper selector iteration over certified selectors terminates after finitely many improvements;
\item value-refined modal bisimulation quotients preserve choice-refined modal truth, admissible continuation membership, optimal values, and optimal selectors, and are coarsest for exact local logical-value quotienting;
\item certificate-gated stopping is sound exactly on the reach-while-stay region generated inside $V$;
\item the value-refined modal bisimulation of clause~(v) is the zero set of a canonical pseudometric $d^\star$, the unique fixed point of the Hausdorff-lifted choice-matching transformer of Section~\ref{sec:bisimulation-metric}, for which $|J^\star(s)-J^\star(t)|\le d^\star(s,t)$, so any $\varepsilon$-quotient under $d^\star$ preserves optimal value up to $\varepsilon$;
\item the pipeline of Section~\ref{sec:pipeline} computes the preceding objects with the stated finite bounds.
\end{enumerate}
If an observation map is supplied and $W=\Viab_{\obs}(G)$ is nonempty, then the same conclusions hold on the information state game: $W$ is the observation-based admissible region, $M_W$ is the largest observation domain map, $\mathcal B_W^{\obs}$ has a unique universal discounted value $J^\star_{\obs}$, and every Bellman-minimizing observation selector is certified and optimal among stationary certified observation selectors.
\end{theorem}

\begin{proof}
The full-information clauses follow respectively from the admissible-kernel theorem, the invariance-game/admissible-kernel theorem, the certified-domain maximality proposition, the contraction theorem for $\mathcal B_V$, the finite selector-iteration theorem, the quotient-invariance and coarseness theorems, the reach-while-stay certification theorem, and the algorithmic bounds proposition.  The metric clause is the contraction, zero-set, and Lipschitz cluster of Section~\ref{sec:bisimulation-metric}.  The observation clauses follow from the information state justification theorem and the information-state admissible value theorem.  The perfect-information collapse theorem identifies the state construction as the singleton-information state instance, so the two parts are compatible rather than separate semantics.
\end{proof}

\begin{theorem}[Boundary-layer relaxation theorem]\label{thm:relaxed-main}
Assume the finite endpoint theorem above.  The following are secondary boundary fragments.  Each is sound when inserted at its declared layer and unsound in general when inserted earlier:
\begin{enumerate}[label=(\alph*)]
\item replacing the hard kernel by a nested admissible continuation potential is sound at every certified threshold, and boundary pressure preserves Bellman contraction;
\item replacing the universal successor maximum by a risk envelope preserves contraction and orders values monotonically from nominal probabilistic evaluation to universal evaluation;
\item replacing a fixed kernel by a moving certified envelope preserves invariance under arbitrary model updates, provided executed choices are certified against the current conservative envelope;
\item replacing the optimizer by an arbitrary candidate core preserves invariance after certified-domain projection, with local performance loss equal to the admissible projection gap;
\item replacing a single global constraint by covered local cell constraints preserves global invariance under the local-to-global incidence condition;
\item replacing a fixed target by finite target revision preserves segment-wise target satisfaction and persistent base invariance under an outer consistency operator.
\end{enumerate}
In every case, the certified choice set used by the value or candidate layer is still generated by a certificate of certifiability; none of these clauses is needed for the finite exact theorem above.  Therefore none of the relaxations authorizes value refinement, candidate selection, quotienting, stopping, or target revision to create its own invariance permission.
\end{theorem}

\begin{proof}
Clause (a) is Proposition~\ref{prop:potential-sound} and Theorem~\ref{thm:boundary-pressure}.  Clause (b) is Theorem~\ref{thm:risk-contraction}.  Clause (c) is Theorem~\ref{thm:moving-envelope}.  Clause (d) is Theorem~\ref{thm:filter-projection}.  Clause (e) is Theorem~\ref{thm:local-global}.  Clause (f) is Theorem~\ref{thm:target-revision}.  The last statement follows from the hypotheses of those theorems: each value, candidate, local, update, or revision rule is defined only after a certified domain map has supplied the allowed choices.
\end{proof}

\section{Running finite example}\label{sec:running-example}

This example follows one finite system through the complete pipeline.  Let
\[
S=\{0,1,2,3,4,5\},\qquad G=\{0,1,2,3,4\},\qquad C=\{4\},
\]
and let the admissible transitions be
\[
T(0,a)=\{1,2\},\quad T(0,b)=\{5\},
\]
\[
T(1,a)=\{3\},\quad T(1,c)=\{5\},
\]
\[
T(2,a)=\{3,4\},\quad T(2,b)=\{2\},
\]
\[
T(3,a)=\{4\},\quad T(4,a)=\{4\},\quad T(5,a)=\{5\}.
\]
All unlisted choices are inadmissible.  Costs are
\[
c(0,a)=3,\quad c(0,b)=0,
\qquad
c(1,a)=2,\quad c(1,c)=0,
\]
\[
c(2,a)=1,\quad c(2,b)=1,
\qquad
c(3,a)=1,\quad c(4,a)=0,\quad c(5,a)=0.
\]
The first admissible continuation round removes $5$ because it is outside $G$.  The remaining states all have at least one choice whose successors stay inside the remaining set, so
\[
V=\{0,1,2,3,4\}.
\]
The largest certified domain disables $b$ at $0$ and disables $c$ at $1$.  These are exactly the cheapest raw choices at those states, so the example shows why cost minimization cannot be run before the justification domain map.

With $\gamma=1/2$, the state value is computed inside $V$.  One obtains
\[
J^\star(4)=0,
\qquad
J^\star(3)=1,
\qquad
J^\star(1)=2+\tfrac12J^\star(3)=\tfrac52.
\]
At $2$, choice $a$ has value $1+\tfrac12\max(J^\star(3),J^\star(4))=\tfrac32$, while choice $b$ solves $J=1+\tfrac12J$, giving value $2$; hence $a$ is optimal and $J^\star(2)=\tfrac32$.  At $0$ the only certified choice is $a$, so
\[
J^\star(0)=3+\frac12\max\left(\frac52,\frac32\right)=\frac{17}{4}.
\]
The certificate predicate is still independent of value: $2$ has lower value than $1$, but $2\notin C$, so stopping at $2$ is unsound.

Now add observations
\[
\obs(0)=o_0,
\quad
\obs(1)=\obs(2)=o_m,
\quad
\obs(3)=o_3,
\quad
\obs(4)=o_4,
\quad
\obs(5)=o_5.
\]
After playing $a$ at $0$, the selector may know only the information state $\{1,2\}$.  At that information state the common certified choices are $A(\{1,2\})=A(1)\cap A(2)=\{a\}$.  The information state successors are $\{3\}$ and $\{4\}$, both admissible.  Therefore $\{1,2\}\in\Viab_{\obs}(G)$, and its universal immediate cost is
\[
\hat c(\{1,2\},a)=\max(c(1,a),c(2,a))=2.
\]
The information state value is
\[
J^{\star}_{\obs}(\{1,2\})=2+\frac12\max(J^{\star}_{\obs}(\{3\}),J^{\star}_{\obs}(\{4\}))=\frac52,
\]
and
\[
J^{\star}_{\obs}(\{0\})=3+\frac12J^{\star}_{\obs}(\{1,2\})=\frac{17}{4}.
\]
Thus the information state lift recovers the same universal value from the initial singleton while replacing hidden-state state-based by an observation-certified selector on information states.

Finally consider quotienting.  If a duplicate terminal state $4'$ is added with the same label, the same self-loop choice, and the same zero cost as $4$, then value-refined modal bisimulation merges $4$ and $4'$ and leaves the value unchanged.  If the duplicate has self-loop cost $1$, label equivalence would still merge the two terminal states, but value-refined modal bisimulation refuses the merge; otherwise the quotient would identify values $0$ and $2$ at discount $1/2$.  This is the smallest instance of the quotient theorem's cost clause.

\section{Further worked computations}

The previous example exhibits admissible continuation, value, stopping, and the universal margin on one system.  The following four computations isolate the quotient, switching, bounded-deliberation, and selector-iteration results with explicit arithmetic.

\subsection{A value-preserving quotient that a cost-blind merge would break}

Let $S=\{p,p',q\}$ with a single proposition that holds at every state, so all labels agree.  Each of $p$ and $p'$ has one choice with cost $2$ leading to $q$, and $q$ has one choice with cost $0$ looping at $q$:
\[
T(p,a)=\{q\},\quad T(p',a)=\{q\},\quad T(q,a)=\{q\},\qquad c(p,a)=c(p',a)=2,\ c(q,a)=0.
\]
With $\gamma=1/2$, the loop gives $J^\star(q)=0+\tfrac12 J^\star(q)$, so $J^\star(q)=0$, and then $J^\star(p)=J^\star(p')=2+\tfrac12\cdot0=2$.  The relation merging $p$ and $p'$ is a value-refined modal bisimulation: labels agree, and the single choice of each has equal cost and reaches the same class $[q]$.  Their optimal values coincide, as the quotient theorem requires.  The coarser cost-free bisimilarity $\approx_R$ identifies all three states, since the unlabelled graph is an endless chain of identical states from each; but $J^\star(p)=2\neq0=J^\star(q)$, so $\approx_R$ does not preserve value.  Value-refined modal bisimulation refuses the merge of $p$ with $q$ precisely because their choice costs differ, which is the separation the value layer needs.

\subsection{Switching gated by the hysteresis penalty}

Let $S=\{x,y,z\}$ with choices $a,b$ and
\[
T(y,a)=\{x\},\ T(y,b)=\{z\},\quad T(x,a)=\{x\},\quad T(z,a)=\{z\},
\]
\[
c(y,a)=c(y,b)=0,\quad c(x,a)=1,\quad c(z,a)=3.
\]
Consider two objectives $G_1=\{x,y\}$ and $G_2=\{y,z\}$.  Under $G_1$ the admissible region is $V_1=\{x,y\}$ with $A_{V_1}(y)=\{a\}$, and with $\gamma=1/2$, $J^\star_1(x)=1+\tfrac12J^\star_1(x)=2$ and $J^\star_1(y)=0+\tfrac12\cdot2=1$.  Under $G_2$ the admissible region is $V_2=\{y,z\}$ with $A_{V_2}(y)=\{b\}$, giving $J^\star_2(z)=3+\tfrac12J^\star_2(z)=6$ and $J^\star_2(y)=0+\tfrac12\cdot6=3$.  At the shared state $y$ the two objectives disagree by the margin $\Delta=J^\star_2(y)-J^\star_1(y)=2$.  A selector committed to the costlier mode $2$ and contemplating a switch to mode $1$ compares staying, value $3$, against switching, value $J^\star_1(y)+\lambda=1+\lambda$.  The switch lowers cost only when $1+\lambda<3$, that is $\lambda<\Delta=2$; for $\lambda\geq2$ the no-switch-under-margin principle keeps the current mode.  When switching does occur, every switch costs at least $\lambda$, so under a total discounted budget $M$ the number of switches cannot exceed $\lfloor M/\lambda\rfloor$, the finite-switching bound.

\subsection{A bounded-deliberation trace with a certified choice}

Return to the six-state system.  Starting from $J_0\equiv0$ and iterating $J_{n+1}=\mathcal B_VJ_n$ over $V=\{0,1,2,3,4\}$ gives
\[
J_1=(3,2,1,1,0),\quad J_2=(4,\tfrac52,\tfrac32,1,0),\quad J_3=(\tfrac{17}{4},\tfrac52,\tfrac32,1,0)=J^\star,
\]
listed in state order $0,1,2,3,4$.  The residuals are $r_0=\|J_1-J_0\|_\infty=3$, $r_1=\|J_2-J_1\|_\infty=1$, and $r_2=\|J_3-J_2\|_\infty=\tfrac14$.  At state $2$ and depth $2$ the two certified choices evaluate to
\[
Q_2(2,a)=1+\tfrac12\max(J_2(3),J_2(4))=\tfrac32,\qquad
Q_2(2,b)=1+\tfrac12 J_2(2)=\tfrac74,
\]
so $a$ is the unique minimizer and is $0$-certified at depth $2$.  The bounded-deliberation theorem bounds its optimality gap by $0+2\gamma r_2/(1-\gamma)=\tfrac12$, while the true gap is $0$ because $a$ is also the infinite-horizon optimizer at $2$.  The a posteriori estimate is consistent: at depth $1$, $\|J_1-J^\star\|_\infty=\tfrac54\leq r_1/(1-\gamma)=2$, and the a priori bound gives the same iterate $\|J_1-J^\star\|_\infty\leq\gamma\|J_0-J^\star\|_\infty=\tfrac{17}{8}$.

\subsection{Selector iteration terminating in one improvement}

On the same system, start from the certified selector $\pi_0$ that selects $b$ at state $2$ and the unique admissible choice elsewhere.  Selector evaluation of $\pi_0$ solves the linear system to give
\[
J^{\pi_0}=(\tfrac{17}{4},\tfrac52,2,1,0),
\]
where $J^{\pi_0}(2)=1+\tfrac12 J^{\pi_0}(2)=2$ reflects the self-loop induced by $b$.  The improvement test at state $2$ compares $Q^{\pi_0}(2,a)=1+\tfrac12\max(J^{\pi_0}(3),J^{\pi_0}(4))=\tfrac32$ against $Q^{\pi_0}(2,b)=1+\tfrac12 J^{\pi_0}(2)=2$, so the improved selector $\pi_1$ switches state $2$ to $a$.  Evaluating $\pi_1$ yields $J^{\pi_1}=(\tfrac{17}{4},\tfrac52,\tfrac32,1,0)=J^\star$, and the next improvement step finds no strict change, so selector iteration halts after a single strict improvement, in agreement with the finite-termination theorem.

\section{Semantic interpretation}

The semantics gives a precise answer to a common semantic-order problem.  A selector must not optimize over all choices; it must optimize over choices that preserve the admissible-continuation fixed point.  A quotient must not be formed from labels alone; it must preserve the value-relevant one-step choice data.  A stopping event must not be inferred from a residual equation outside the kernel; it must be certified after the fixed point has supplied the domain of value.  The finite countermodels of Section~\ref{sec:worked-discipline} show that these are not expository preferences but non-commuting semantic operations.

The theory also explains why state compression must be disciplined.  If compression preserves labels but not choice costs, then formulas may survive while values change.  If compression preserves costs but not labels, then values may survive while specifications change.  Value-refined modal bisimulation is sufficient because it preserves both layers.

Switching choice gives a second form of discipline.  A selector with many objectives can drift if every mode change is free.  A positive switching penalty turns mode changes into budget-consuming events.  If total cost is bounded, the number of switches is bounded.  This is a simple theorem, but it captures the structural role of commitment.

Bounded deliberation gives a third form of discipline.  Acting before convergence is not forbidden.  Acting outside the certified choice set is forbidden.  The residual bound permits approximate value refinement while preserving exact invariance.

\section{Boundary conditions and extensions}

The finite-state assumption marks the first boundary of the theorem.  Infinite-state systems require additional topological or measurable structure.  The discount assumption is the second boundary.  Average-cost, total-cost, and undiscounted games require different fixed-point arguments.  The universal successor convention is conservative.  Probabilistic successors may yield less pessimistic values but require probability measures and risk choices.

Observation is treated through the information-state construction.  The qualitative problem becomes a greatest fixed point on observation-consistent information states, and the quantitative problem becomes a universal discounted Bellman equation whose resolver resolves hidden state, successor, and observation.  The full-information hypothesis of the state sections is therefore the singleton-information state collapse, not a separate assumption that observation is immaterial.

The modal language used in the core sections is the invariance fragment needed for the main theorem.  Richer temporal objectives can be added through automata products or mu-calculus formulas.  That extension would increase the size of the state space but would not change the central separation between admissible continuation and value.  Algorithmic refinements are orthogonal to the semantic statement.  Partition refinement and symbolic methods can improve implementation, but the current target is mathematical clarity.

\section{Public share-alike certificate as a semantic fragment}
\label{sec:bysa-certificate}

The manuscript has a public release certificate, not merely a footer.  The certificate is read through the same semantic components used above: a target predicate, a derivative relation, a greatest fixed point, certified local choices, value-refinement, quotient preservation, and residual adequacy.  The legal instrument named by the certificate is Creative Commons Attribution--ShareAlike 4.0 International, whose public deed describes permissions to share and adapt material under attribution, same-license distribution of adaptations, and no additional restrictions \cite{ccby-sa40}.  The legal code gives the corresponding formal conditions for attribution, adapter licensing, and downstream restrictions \cite{ccby-sa40legal}.  The semantic construction below does not restate the license.  It records the internal proof object by which this article treats its own release as a certified derivative state.

\begin{definition}[Derivative-release frame]
A derivative-release frame is a finite tuple
\[
\mathfrak R=(R,r_0,\Der,\Attr,\Lic,\Mod,\NoAdd,\omega),
\]
where $R$ is a finite set of manuscript states, $r_0\in R$ is the root state, $\Der\subseteq R\times R$ is the one-step derivative relation, $\Attr(r)$ records the attribution payload visible at $r$, $\Lic(r)$ is its release label, $\Mod(r)$ is the retained modification trace, $\NoAdd(r)$ says that $r$ carries no additional downstream restriction, and $\omega:R\to\mathbb R_{\geq 0}$ is a public-cost index measuring certificate burden.  The root attribution payload is denoted $\Root$.  The distinguished release label is written $\mathrm{CC\!BY\!SA4}$.
\end{definition}

\begin{definition}[Attribution--share-alike target]
The attribution--share-alike target of $\mathfrak R$ is
\[
G_{\BYSA}=\{r\in R:\Root\subseteq\Attr(r),\ \Lic(r)=\mathrm{CC\!BY\!SA4},\ \Mod(r)\text{ retains the prior trace},\ \NoAdd(r)\}.
\]
A local derivative choice at $r$ is any nonempty finite set $c\subseteq\{r'\in R:(r,r')\in\Der\}$.  It is $\BYSA$-certified over $X\subseteq R$ when $c\subseteq X$.  The predecessor induced by derivative choices is
\[
\Pre_{\BYSA}(X)=\{r\in R:\exists c\in\Cell(r)\; c\subseteq X\}.
\]
The public share-alike kernel is
\[
K_{\BYSA}=\nu X.\bigl(G_{\BYSA}\cap \Pre_{\BYSA}(X)\bigr).
\]
\end{definition}

The notation is intentionally parallel to the main kernel.  Attribution is a label-preservation requirement, ShareAlike is a derivative-closure requirement, retained modification trace is a history condition, and absence of additional restriction is an admissibility condition.  Thus the release certificate is not a separate formalism: it is the main fixed-point engine applied to manuscript derivatives.

\begin{proposition}[Share-alike closure is a greatest fixed point]
A manuscript state $r$ lies in $K_{\BYSA}$ iff there exists a possibly infinite derivative policy through which every finite prefix remains inside $G_{\BYSA}$ and each step chooses only derivative states that again admit such a continuation.  Equivalently, $K_{\BYSA}$ is the largest set of release states closed under certified derivative continuation.
\end{proposition}

\begin{proof}
The map $F_{\BYSA}(X)=G_{\BYSA}\cap\Pre_{\BYSA}(X)$ is monotone on the finite lattice $\mathcal P(R)$.  Tarski's theorem gives the greatest fixed point $K_{\BYSA}=\nu X.F_{\BYSA}(X)$.  If $r\in K_{\BYSA}$, then $r\in G_{\BYSA}$ and some derivative choice $c\subseteq K_{\BYSA}$ is available; iterating this witness keeps every prefix in the target.  Conversely, any set from which such closed continuation is possible is a post-fixed point of $F_{\BYSA}$ and is therefore contained in the greatest fixed point.
\end{proof}

\begin{definition}[BY-SA certificate]
For a manuscript state $r\in K_{\BYSA}$, a BY-SA certificate is a tuple
\[
\mathsf C_{\BYSA}(r)=(r,\Attr(r),\Lic(r),\Mod(r),\NoAdd(r),\rho(r)),
\]
where $\rho(r)$ is a residual bound for checking that the displayed attribution payload, license label, modification trace, and no-additional-restriction predicate are stable under the finite derivative choices inspected from $r$.  The certificate is accepted when $\rho(r)=0$ in the exact finite frame, or when $\rho(r)\leq\varepsilon$ under an explicitly declared finite approximation tolerance.
\end{definition}

The value-refined component now ranks certified release choices without changing the fixed point.  Let
\[
q(r,c)=\omega(r)+\eta_{\Attr}\,d_{\Attr}(r,c)+\eta_{\Mod}\,d_{\Mod}(r,c)+\eta_{\Lic}\,d_{\Lic}(r,c),
\]
where the three defect terms measure missing attribution payload, incomplete modification trace, and license-label drift along the derivative choice $c$.  In the exact certificate fragment these defects are zero on $\BYSA$-certified choices.  With discount $\gamma\in(0,1)$, define
\[
(\mathcal V_{\BYSA}J)(r)=\min_{c\in\CAct_{\BYSA}(r)}\max_{r'\in c}\bigl(q(r,c)+\gamma J(r')\bigr).
\]
Thus value refinement selects the least burdensome certified public derivative route, but it cannot select a route outside $K_{\BYSA}$.

\begin{theorem}[The manuscript release certificate is an instance of certified choice]
Let $M$ be the current manuscript state and suppose $M\in G_{\BYSA}$.  If every public derivative cell declared by the manuscript preserves attribution payload, the CC BY-SA 4.0 label, modification trace, and absence of additional downstream restriction, then
\[
M\in K_{\BYSA},\qquad \CAct_{\BYSA}(M)\neq\varnothing,
\]
and the certificate $\mathsf C_{\BYSA}(M)$ is a certified-choice witness in the sense of the main construction.  Moreover, any exact value-certified extension of the release frame factors through $K_{\BYSA}$, and any quotient preserving $\Attr$, $\Lic$, $\Mod$, $\NoAdd$, costs $q$, and successor certificate classes preserves the accepted BY-SA certificate.
\end{theorem}

\begin{proof}
The hypotheses say exactly that $M\in G_{\BYSA}$ and that every declared derivative cell from $M$ is contained in $G_{\BYSA}$ while preserving the same conditions at the next manuscript states.  The set of states satisfying the same closed condition is a post-fixed point of $F_{\BYSA}$, hence is contained in $K_{\BYSA}$.  Therefore $M\in K_{\BYSA}$ and at least one certified derivative choice is available.  The tuple $\mathsf C_{\BYSA}(M)$ records the target clauses used by the fixed-point witness, so it is a certificate in the earlier sense.  The factorization claim is the canonical-completion theorem applied with $G=G_{\BYSA}$ and with derivative cells as local choices.  The quotient claim follows from value-refined bisimulation: the listed fields are precisely the labels, local costs, and successor certificate classes needed to keep the kernel, value transformer, greedy certified choices, and residual tests invariant.
\end{proof}

\begin{corollary}[No drift under public derivation]
If $M\in K_{\BYSA}$ and $r$ is reachable from $M$ by certified derivative choices, then $r\in K_{\BYSA}$ and $\Root\subseteq\Attr(r)$, $\Lic(r)=\mathrm{CC\!BY\!SA4}$, $\Mod(r)$ retains the prior trace, and $\NoAdd(r)$ holds.  Hence the public certificate propagates along every certified derivative branch.
\end{corollary}

\begin{proof}
Reachability by certified derivative choices means that each step follows a cell contained in $K_{\BYSA}$.  Since $K_{\BYSA}\subseteq G_{\BYSA}$, every reached state satisfies the four clauses defining the attribution--share-alike target.
\end{proof}

For this manuscript, the declaration $\BYSA(M)$ should therefore be read as a small public instance of the article's own semantics.  It says that the written object is not only accompanied by a release notice; it carries a checkable derivative-closure certificate whose clauses are attribution preservation, same-label propagation, retained modification history, and absence of downstream restriction.  The mathematical theorems do not depend on this release certificate.  Rather, the release certificate demonstrates the same discipline on the manuscript object itself: close first, refine second, certify last.

\section{Conclusion}

The construction gives a finite modal account of certified choice in value-refined modal semantics.  The core fixed point is classical: it is the admissible-predecessor greatest fixed point of the modal $\mu$-calculus and the one-step ATL/coalition ability operator iterated to invariance.  The semantic content is the order placed around that core.  Modal truth is closed under admissible continuation; value is computed only over certified choices; quotienting must preserve labels, successor classes, and costs; and stopping must be justified by a residual certificate.

This separation gives a precise boundary for the construction.  It does not propose a new temporal language, and it does not replace existing choice or game logics.  It identifies the finite quotient and certificate discipline needed when a fixed-point modal semantics is refined by value.  The retained transition fixed point, the branching choice-cell representation, the non-commutation witnesses, and the coarsest value-refined bisimulation theorem are the central mathematical objects.  The bisimulation-metric refinement of that theorem grades the quotient discipline quantitatively: it collapses to the exact equivalence at distance zero and bounds, by the metric itself, the optimal-value error incurred by any inexact compression.  Boundary fragments show how the same order can survive selected relaxations, but the main theorem remains the finite exact construction.

\section{Retained boundary fragments}
\label{sec:retained-boundary-fragments}
The following fragments are retained from the expanded construction as compatibility tests.  They are not premises of the main theorem cluster.  Each fragment records how one rigidity of the finite core can be loosened while the layer order is kept intact under the stated hypotheses.

\subsection{Retained fragment: layer-preserving relaxation of the admissible kernel}\label{sec:relaxation}

The preceding sections establish the finite exact endpoint: a Boolean target, a hard admissible continuation kernel, a largest certified domain, and a discounted universal value.  This section turns the endpoint into a limiting case.  The invariant part is not the particular rigidity of the endpoint; it is the order in which the layers are allowed to act.  Truth still precedes certifiability, certifiability still precedes preference, and preference still precedes stopping or revision.  The relaxations below loosen the objects carried by the layers without commuting the layers themselves.

\begin{definition}[Layer-ordered extension]\label{def:layer-extension}
A layer-ordered extension of the finite endpoint consists of objects
\[
G \quad\leadsto\quad \mathcal K \quad\leadsto\quad M_{\mathcal K}\quad\leadsto\quad \mathcal V\quad\leadsto\quad \mathcal C,
\]
where $G$ is the truth region of a formula or finite family of formulas, $\mathcal K$ is a certified certifiability object, $M_{\mathcal K}$ is the set of choices certified by $\mathcal K$, $\mathcal V$ is a preference or value functional defined only over $M_{\mathcal K}$, and $\mathcal C$ is a termination, quotient, update, or target-revision certificate defined only after the previous layers have been fixed.  The extension is sound when every selector that selects choices from $M_{\mathcal K}$ satisfies the original invariance objective whenever the certificate assumptions of $\mathcal K$ are true.
\end{definition}

The definition deliberately does not prescribe a unique representation of $\mathcal K$.  In the endpoint $\mathcal K$ is the Boolean set $V=\Viab(G)$.  In the relaxed constructions it may be a nested potential, an information state kernel, a moving envelope, or a family of local certificates.  The next theorems give sufficient conditions under which these representations remain sound.

\subsection{Retained fragment: semantic potentials}\label{sec:potential}

A Boolean kernel separates states into admissible and nonviable.  That separation is sometimes too coarse: in large symbolic spaces and continuous abstractions, a selector should feel almost no restriction deep inside the admissible region and should tighten only near the boundary.  The construction therefore promotes the margin from a tie-breaker into a primary object, while keeping hard invariance as the zero-level certificate.

\begin{definition}[Nested admissible continuation potential]\label{def:nested-potential}
Let $G\subseteq S$.  A nested admissible continuation potential is a map $p:S\to\mathbb R\cup\{-\infty\}$ together with a finite threshold set $\Theta\subseteq\mathbb R$ such that, for every $\theta\in\Theta$, the superlevel set
\[
P_\theta=\{s\in S:p(s)\geq \theta\}
\]
satisfies $P_\theta\subseteq G$ and is choice-refined invariant: for every $s\in P_\theta$ there exists $u\in A(s)$ with $T(s,u)\subseteq P_\theta$.  The induced margin domain map is
\[
M_\theta(s)=\{u\in A(s):T(s,u)\subseteq P_\theta\}.
\]
When $P_0=\Viab(G)$, the potential is exact at the hard endpoint.
\end{definition}

\begin{proposition}[Potential domain maps are sound and nested]\label{prop:potential-sound}
For every threshold $\theta\in\Theta$, any selector $\pi$ with $\pi(s)\in M_\theta(s)$ for all visited $s\in P_\theta$ keeps every play inside $P_\theta$ and hence inside $G$.  If $\theta_1\leq\theta_2$, then $P_{\theta_2}\subseteq P_{\theta_1}$ and $M_{\theta_2}(s)\subseteq M_{\theta_1}(s)$ on $P_{\theta_2}$.
\end{proposition}

\begin{proof}
The first claim follows by induction on time.  At time zero the state is in $P_\theta$.  If the current state is in $P_\theta$ and the selector selects a choice from $M_\theta$, every successor belongs to $P_\theta$ by definition of the domain map.  Thus all finite prefixes and all infinite plays remain in $P_\theta\subseteq G$.  The nesting of sets follows from the order of superlevel sets: $p(s)\geq\theta_2$ and $\theta_2\geq\theta_1$ imply $p(s)\geq\theta_1$.  The inclusion of domain maps follows because preserving the smaller set $P_{\theta_2}$ is a stronger condition than preserving the larger set $P_{\theta_1}$.
\end{proof}

\begin{definition}[Boundary-sensitive value]\label{def:boundary-value}
Let $\ell:\mathbb R\to\mathbb R_{\geq0}$ be nonincreasing and let $\lambda\geq0$.  For a fixed threshold $\theta$, define
\[
(\mathcal B_{\theta,\lambda}J)(s)=\min_{u\in M_\theta(s)}\left(c(s,u)+\lambda\ell(p(s))+\gamma\max_{s'\in T(s,u)}J(s')\right).
\]
The term $\lambda\ell(p(s))$ is a boundary pressure: it is small deep inside the potential and grows near the certified edge.
\end{definition}

\begin{theorem}[Boundary pressure does not weaken invariance]\label{thm:boundary-pressure}
For every fixed threshold $\theta$ and $\lambda\geq0$, the operator $\mathcal B_{\theta,\lambda}$ is a $\gamma$-contraction on $\ell_\infty(P_\theta)$ and has a unique fixed point.  Every greedy selector for that fixed point is admissible for $G$.
\end{theorem}

\begin{proof}
The added pressure term depends on the current state but not on the continuation value $J$.  Therefore the same maximum and minimum perturbation lemmas used for the endpoint Bellman operator give the contraction bound with factor $\gamma$.  The Banach fixed-point theorem gives uniqueness.  A greedy selector chooses only from $M_\theta$, and Proposition~\ref{prop:potential-sound} proves invariance.
\end{proof}

This construction is intentionally representation-free.  The potential may be graph distance to the losing attractor, a ranking function, an abstraction margin, a validated conservative score validated by superlevel certificates, or a domain-specific clearance measure.  The theorem uses only certified choice-refined invariance of the superlevel sets, not the origin of the potential.

\subsection{Retained fragment: risk envelopes between universal and probabilistic successors}\label{sec:risk-dial}

The endpoint model evaluates each certified choice by the universal successor.  This is correct for universal nondeterminism but too rigid when successor frequencies or confidence information are available.  The admissible way to introduce probability is not to blur the semantics; it is to insert a visible risk envelope after the certified domain has already been computed.

\begin{definition}[Successor risk envelope]\label{def:risk-envelope}
Assume that each pair $(s,u)$ has a nominal distribution $q_{s,u}$ on $T(s,u)$.  A successor risk envelope is a family $(\mathfrak Q_\beta)_{\beta\in[0,1]}$ such that $\mathfrak Q_\beta(s,u)$ is a nonempty closed convex set of distributions on $T(s,u)$, $q_{s,u}\in\mathfrak Q_\beta(s,u)$, and
\[
\mathfrak Q_0(s,u)=\{q_{s,u}\},
\qquad
\mathfrak Q_1(s,u)=\Delta(T(s,u)),
\qquad
\mathfrak Q_\beta(s,u)\subseteq\mathfrak Q_{\beta'}(s,u)	ext{ when }\beta\leq\beta'.
\]
The risk-evaluated admissible Bellman operator is
\[
(\mathcal B_V^\beta J)(s)=\min_{u\in A_V(s)}\left(c(s,u)+\gamma\sup_{q\in\mathfrak Q_\beta(s,u)}\sum_{s'\in T(s,u)}q(s')J(s')\right).
\]
\end{definition}

\begin{theorem}[Risk-envelope contraction and ordering]\label{thm:risk-contraction}
For every $\beta\in[0,1]$, $\mathcal B_V^\beta$ is a $\gamma$-contraction on $\ell_\infty(V)$ and has a unique fixed point $J_\beta^\star$.  If $\beta\leq\beta'$, then $J_\beta^\star\leq J_{\beta'}^\star$ pointwise.  Moreover $\beta=1$ recovers the universal successor maximum, while $\beta=0$ recovers the nominal probabilistic Bellman operator on the same certified domain.
\end{theorem}

\begin{proof}
For any fixed distribution $q$, the map $J\mapsto\sum q(s')J(s')$ is $1$-Lipschitz in the sup norm.  Taking the supremum over a nonempty family of such linear maps remains $1$-Lipschitz by the maximum perturbation argument, and taking the minimum over choices remains $1$-Lipschitz before multiplication by $\gamma$.  Hence $\mathcal B_V^\beta$ is a $\gamma$-contraction and has a unique fixed point.  If $\beta\leq\beta'$, then the supremum for $\beta$ is taken over a subset of the distributions used for $\beta'$, so $\mathcal B_V^\beta J\leq\mathcal B_V^{\beta'}J$ for all $J$.  The monotone contraction sandwich gives $J_\beta^\star\leq J_{\beta'}^\star$.  When $\beta=1$, the supremum over all distributions on a finite successor set is attained at a Dirac mass on a maximizer, so it equals $\max_{s'\in T(s,u)}J(s')$.  When $\beta=0$, the only distribution is $q_{s,u}$, giving expectation under the nominal model.
\end{proof}

The dial is not an additional logical operator.  It is a modeling parameter in the value layer.  Invariance remains choice-refined by $A_V$ or by one of the certified relaxed domain maps.  Moving the dial can change the value and the preferred certified choice, but it cannot authorize a choice that the logical layer excluded.

\subsection{Retained fragment: moving certified domains}\label{sec:moving-envelope}

A one-shot kernel assumes that the transition model is fixed.  In updated semantic construction, the model may be refined as observations arrive.  A moving boundary is admissible only if every newly admitted choice is certified against a conservative model before it is used.

\begin{definition}[Certified transition envelope]\label{def:certified-envelope}
Let $T^\dagger$ be the unknown true successor map.  A transition envelope at time $k$ is a map $\widehat T_k$ with $T^\dagger(s,u)\subseteq\widehat T_k(s,u)$ for all $s,u$.  A certified admissible envelope is a set $C_k\subseteq S$ such that $C_k\subseteq G$ and, for every admitted choice,
\[
M_k(s)=\{u\in A(s):\widehat T_k(s,u)\subseteq C_k\}
\]
is nonempty on every state in $C_k$.
\end{definition}

\begin{algorithm}[H]
\caption{Certified moving envelope}
\begin{algorithmic}[1]
\Require target $G$, conservative envelopes $\widehat T_k$, current state $s_k$, candidate $u_k^{prop}$
\Ensure executed choice $u_k$ with no transient invariance violation
\State compute or update an inner certificate $C_k\subseteq\Viab_{\widehat T_k}(G)$ containing $s_k$
\State compute $M_k(s_k)=\{u:\widehat T_k(s_k,u)\subseteq C_k\}$
\If{$u_k^{prop}\in M_k(s_k)$}
\State $u_k\gets u_k^{prop}$
\Else
\State choose any $u_k\in M_k(s_k)$, for example a Bellman-minimizing fallback
\EndIf
\State execute $u_k$ and record the next observation for the construction of $\widehat T_{k+1}$
\end{algorithmic}
\end{algorithm}

\begin{theorem}[Invariance under arbitrary certified boundary motion]\label{thm:moving-envelope}
Suppose that for every time $k$ the true successor map is contained in the current envelope, $T^\dagger\subseteq\widehat T_k$, and the selector executes only choices in $M_k(s_k)$ for a certified set $C_k\subseteq G$.  If $s_0\in C_0$ and every update supplies a certificate containing the realized next state, then every realized state belongs to $G$.  If in addition the envelopes are monotonically refined, $\widehat T_{k+1}(s,u)\subseteq\widehat T_k(s,u)$, then the exact kernels are monotone expanding: $\Viab_{\widehat T_k}(G)\subseteq\Viab_{\widehat T_{k+1}}(G)$.
\end{theorem}

\begin{proof}
At time $k$, the executed choice satisfies $\widehat T_k(s_k,u_k)\subseteq C_k$.  Since the true successor set is contained in the envelope, the realized next state lies in $C_k\subseteq G$.  The additional certificate condition for the next step ensures that the induction can continue.  Thus all realized states remain in $G$.  For the monotonicity statement, $\widehat T_{k+1}\subseteq\widehat T_k$ makes the resolver weaker.  The kernel-shrinking theorem applied in the opposite direction gives $\Viab_{\widehat T_k}(G)\subseteq\Viab_{\widehat T_{k+1}}(G)$.
\end{proof}

The theorem is stronger than recomputing the kernel after each update.  Recomputing is harmless only if the selector refuses to use a choice until the current envelope certifies it.  The boundary may breathe; the admissibility test may not.

\subsection{Retained fragment: approximate candidate maps under an exact certified domain}\label{sec:verified-shell}

Approximation is dangerous in the invariance layer but useful in the value layer.  The clean separation is to let an arbitrary candidate map rank choices while a proof-carrying domain map projects the candidate back into the certified set.

\begin{definition}[Certified-domain projection]\label{def:filter-projection}
Let $M(s)$ be any nonempty certified domain map on a region $K$.  Let $\mu:S\to U$ be an arbitrary candidate selector, and let $r(s,u)$ be any tie-breaking score on $M(s)$.  The certified-domain projection of $\mu$ is
\[
\Pi_M\mu(s)=
\begin{cases}
\mu(s), & \mu(s)\in M(s),\\
\argmin_{u\in M(s)} r(s,u), & \mu(s)\notin M(s).
\end{cases}
\]
\end{definition}

\begin{theorem}[Candidate-independent invariance and projection gap]\label{thm:filter-projection}
If $M$ is sound for $G$ on $K$, then $\Pi_M\mu$ is sound for $G$ on $K$ for every candidate selector $\mu$.  If $Q^\star$ is the optimal admissible state-choice value and the fallback minimizes $Q^\star$ on $M(s)$ whenever projection is needed, then the one-step loss at $s$ is zero after projection.  For an arbitrary fallback $\bar u(s)\in M(s)$, the local loss is bounded by
\[
Q^\star(s,\bar u(s))-\min_{u\in M(s)}Q^\star(s,u).
\]
\end{theorem}

\begin{proof}
The projected selector always chooses an element of $M(s)$ on $K$, so soundness follows directly from the soundness of the certified domain map.  If the fallback minimizes $Q^\star$, then the selected choice has the same admissible optimal one-step value as the best certified choice.  For an arbitrary fallback, the displayed expression is exactly the excess of the fallback's admissible state-choice value over the admissible optimum at $s$.
\end{proof}

This is the formal role of approximation in the construction.  The candidate map may be fast, generous, approximate, and data-driven.  The certified domain map is small, symbolic, and proof-carrying.  Approximation can improve the selected value only inside the choices already established by the certified domain map.

\subsection{Retained fragment: composing local cell constraints}\label{sec:local-constraints}

The most expensive object in the endpoint is the global kernel.  A scalable relaxation is to replace one global domain map by several local constraints whose certificates compose.  The following theorem is only a sufficient condition; it does not solve the general distributed case.  It gives a disciplined case in which no central authority needs to compute the whole game.

\begin{definition}[Local cell-constraint cover]\label{def:local-cover}
Let the global state set embed into a product $S\subseteq S_1\times\cdots\times S_r$, with projections $\pi_i:S\to S_i$.  For each component $i$, let $G_i\subseteq S_i$, let $V_i\subseteq G_i$, and let $M_i(x_i)\subseteq U_i$ be a local certified domain map satisfying
\[
T_i(x_i,u_i)\subseteq V_i
\quad\text{for all }x_i\in V_i,\ u_i\in M_i(x_i).
\]
A global transition relation $T$ is covered by the local constraints if every global choice $u=(u_1,\ldots,u_r)$ with $u_i\in M_i(\pi_i(s))$ satisfies
\[
\pi_i(s')\in T_i(\pi_i(s),u_i)
\quad\text{for every }s'\in T(s,u)\text{ and every }i.
\]
\end{definition}

\begin{theorem}[Local-to-global constraint soundness]\label{thm:local-global}
Under a local cell-constraint cover, the global set
\[
V_{loc}=S\cap\bigcap_{i=1}^r\pi_i^{-1}(V_i)
\]
is choice-refined invariant for the global target
\[
G_{loc}=S\cap\bigcap_{i=1}^r\pi_i^{-1}(G_i).
\]
Consequently any selector whose component choices satisfy $u_i\in M_i(\pi_i(s))$ keeps all plays inside $G_{loc}$ from every initial state in $V_{loc}$.
\end{theorem}

\begin{proof}
Let $s\in V_{loc}$ and choose a componentwise constrained choice.  For every successor $s'\in T(s,u)$ and every component $i$, the cover condition gives $\pi_i(s')\in T_i(\pi_i(s),u_i)$.  Since $\pi_i(s)\in V_i$ and $u_i\in M_i(\pi_i(s))$, local soundness gives $T_i(\pi_i(s),u_i)\subseteq V_i$.  Hence $\pi_i(s')\in V_i$ for all $i$, so $s'\in V_{loc}$.  Since every $V_i\subseteq G_i$, the global state stays in $G_{loc}$.
\end{proof}

\subsection{Retained fragment: target revision under an outer consistency operator}\label{sec:target-revision}

The outer relaxation allows target revision, while a higher layer checks that the proposed target is admissible before it can govern choice.  This is not the same problem as optimizing under a fixed target.  It requires an outer finite domain map over candidate formulas.

\begin{definition}[Finite target-revision interface]\label{def:target-revision}
Let $\Phi=\{\varphi_1,\ldots,\varphi_N\}$ be a finite set of candidate target formulas, with regions $G_i=\Sat(\varphi_i)$.  Let $\mathsf{Cons}\subseteq\Phi$ be an outer consistency predicate and let $\Rev(i,h)\subseteq\Phi$ be the set of targets that the current target $\varphi_i$ is allowed to revise to after finite history $h$.  A revision is certified at history $h$ if the proposed $\varphi_j$ satisfies $\varphi_j\in\mathsf{Cons}$, $\varphi_j\in\Rev(i,h)$, and the current information state belongs to the corresponding certified kernel.
\end{definition}

\begin{theorem}[Outer-operator soundness for finite target revision]\label{thm:target-revision}
Suppose every executed revision is certified in the sense of Definition~\ref{def:target-revision}, and after revision the selector selects choices only from the certified domain associated with the new target's certified kernel.  Then every play segment between two consecutive revisions satisfies the target active on that segment.  If the outer predicate also requires all candidate targets to imply a persistent base invariance formula $\psi_0$, then the entire play satisfies $\psi_0$ at every time.
\end{theorem}

\begin{proof}
At each revision time, certification places the current information state inside the certified kernel of the newly active target.  Between revision times, the selector chooses only choices from that target's certified domain, so the corresponding justification theorem keeps the segment inside the active target region.  If every active target implies the base invariance formula $\psi_0$, then membership in the active target region at every time implies satisfaction of $\psi_0$ at every time, including across revision boundaries.
\end{proof}

The statement is deliberately limited: target revision is admissible only when each proposed target is certified before it constrains subsequent choice.  Without that outer domain map, target generation would commute preference above truth and break the layer order.

\subsection{Retained note: semantic pressure and value tension}

The phrase semantic pressure denotes the loss of future admissibility caused by a present choice.  It is not a psychological notion and it is not a count.  Measuring it by the raw number of in-kernel successors conflates ``many successors'' with ``choice-forcing future'' and depends on the arbitrary normalizer $|V|$; the principled quantity is how deep inside the justifiably maintainable region a choice keeps the play, measured as distance to the boundary of the kernel.

\begin{definition}[Continuation image]
For $s\in V$ and $u\in A_V(s)$, the continuation image is $R_V(s,u)=T(s,u)$.  Admissibility of $u$ in $A_V(s)$ gives $R_V(s,u)\subseteq V$.
\end{definition}

\begin{definition}[Flattened reachability and the kernel boundary]
Write $s\to s'$ if $s'\in T(s,u)$ for some $u\in A(s)$.  The \emph{inner boundary} of $V$ is
\[
\partial V=\{\,s\in V:\exists\,s'\notin V\ \text{with}\ s\to s'\,\}.
\]
The \emph{escape distance} $\dist_V:S\to\mathbb N\cup\{\infty\}$ is the least fixed point of
\[
\dist_V(s)=0\ \ (s\notin V),\qquad
\dist_V(s)=1+\min_{s':\,s\to s'}\dist_V(s')\ \ (s\in V),
\]
with $\dist_V(s)=\infty$ when no flattened path from $s$ reaches $S\setminus V$.  Thus $\dist_V$ is the breadth-first layering of $V$ inward from its complement: $\dist_V(s)=1$ on $\partial V$ and increases with structural depth.
\end{definition}

\begin{definition}[Universal boundary margin]
For $s\in V$ and $u\in A_V(s)$, the \emph{universal boundary margin} of the choice is
\[
\marg_V(s,u)=\min_{s'\in T(s,u)}\dist_V(s').
\]
\end{definition}

Because $u\in A_V(s)$ forces $T(s,u)\subseteq V$, every successor has $\dist_V(s')\geq1$, so $\marg_V(s,u)\geq1$.  The minimum over successors is the worst case the resolver can impose, so $\marg_V$ is a universal margin: a choice with large margin keeps the play deep inside the kernel even against the least favourable successor, which is precisely the future-choice maintainability that the cardinality proxy only gestured at.

\begin{definition}[Semantic pressure as a margin penalty]\label{def:pressure}
Let $g:\mathbb N\cup\{\infty\}\to[0,1]$ be any nonincreasing function with $\inf_d g(d)=0$.  The \emph{semantic pressure} of an admissible choice is
\[
P_V(s,u)=g\bigl(\marg_V(s,u)\bigr).
\]
\end{definition}

\begin{remark}[The map from margin to penalty is a modeling parameter, not a constant]
Definition~\ref{def:pressure} fixes the \emph{structure} of pressure --- a nonincreasing function of the universal margin --- without fixing the function $g$.  The development below uses only two properties: that $P_V$ is valued in $[0,1]$, and, for the ordinal selector, that it is monotone in the margin.  Every $g$ meeting these is admissible, and no exchange rate, threshold, or normalizer is mandated by the construction; the choice of $g$ is left to the model rather than hardcoded here.  This is the deliberate replacement for $1-|R_V(s,u)|/|V|$, whose dependence on $|V|$ and on a count was an arbitrary commitment of exactly the kind the rest of the construction avoids.
\end{remark}

\begin{definition}[Penalty-regularized Bellman operator]
For a bounded penalty $m:\{(s,u):s\in V,\ u\in A_V(s)\}\to[0,1]$ and a weight $\eta\geq0$, define
\[
(\mathcal B_V^{\eta,m}J)(s)=\min_{u\in A_V(s)}\left[c(s,u)+\eta\,m(s,u)+\gamma\max_{s'\in T(s,u)}J(s')\right].
\]
The \emph{semantic-pressure operator} is the instance $\mathcal B_V^{\eta}=\mathcal B_V^{\eta,P_V}$.
\end{definition}

\begin{theorem}[Penalty regularization preserves contraction]
For every $\eta\geq0$ and every bounded $m$, the operator $\mathcal B_V^{\eta,m}$ has a unique fixed point and value iteration converges uniformly to it.
\end{theorem}

\begin{proof}
Let $J,K\in\ell_\infty(V)$ and fix $s\in V$.  For each $u\in A_V(s)$ the terms $c(s,u)$ and $\eta\,m(s,u)$ are identical in the two evaluations and cancel.  The maximum perturbation lemma gives
\[
\left|\gamma\max_{s'\in T(s,u)}J(s')-\gamma\max_{s'\in T(s,u)}K(s')\right|\leq \gamma\|J-K\|_\infty,
\]
and the minimum perturbation lemma over the finite nonempty set $A_V(s)$ transfers the bound to $|(\mathcal B_V^{\eta,m}J)(s)-(\mathcal B_V^{\eta,m}K)(s)|$.  Supremizing over $s$ yields a $\gamma$-contraction; Banach's theorem gives the unique fixed point and uniform convergence.
\end{proof}

\begin{remark}
Penalty regularization does not replace admissible continuation.  It ranks certified choices after the uncertified choices have already been removed, and no finite penalty can make a non-admissible choice logically admissible.  The order between the qualitative constraint and the quantitative penalty is what keeps the invariance guarantee exact.
\end{remark}

\begin{corollary}[Vanishing regularization]
Let $J^\star_{\eta,m}$ be the fixed point of $\mathcal B_V^{\eta,m}$ and $J^\star=J^\star_{0,m}$.  Since $m(s,u)\in[0,1]$,
\[
\|J^\star_{\eta,m}-J^\star\|_\infty\leq\frac{\eta}{1-\gamma},
\]
so $J^\star_{\eta,m}\to J^\star$ uniformly as $\eta\to0^+$; in particular this holds for the semantic-pressure penalty $m=P_V$.
\end{corollary}

\begin{proof}
$\mathcal B_V^{\eta,m}$ is the admissible Bellman operator for the cost $c_\eta(s,u)=c(s,u)+\eta\,m(s,u)$ on the same region $V$.  Then $\|c_\eta-c\|_\infty=\eta\|m\|_\infty\leq\eta$, and the Lipschitz dependence of the fixed point on the cost gives the bound.
\end{proof}

A scalarized penalty trades value against margin at the fixed weight $\eta$.  When the intent is only to break ties among choices that are already value-optimal, the margin can be used ordinally, paying no value at all and committing to no $g$.

\begin{proposition}[Secondary value refinement over the Bellman argmin set]\label{prop:secondary}
Let $Q^\star(s,u)=c(s,u)+\gamma\max_{s'\in T(s,u)}J^\star(s')$ and let
\[
O(s)=\Big\{u\in A_V(s):Q^\star(s,u)=\min_{v\in A_V(s)}Q^\star(s,v)\Big\}
\]
be the value-optimal certified choices at $s$.  Any selector with
\[
\pi_{\mathrm{lex}}(s)\in\argmax_{u\in O(s)}\marg_V(s,u)
\]
is admissible, attains the optimal value $J^{\pi_{\mathrm{lex}}}=J^\star$, and among all value-optimal certified selectors it maximizes the universal boundary margin at every state.  The selection is ordinal: it uses only the order of $\marg_V$ on $O(s)$, hence requires neither the function $g$ nor any numerical constant.
\end{proposition}

\begin{proof}
Each $O(s)\subseteq A_V(s)$, so every selected choice keeps all successors in $V$ and the selector is admissible.  Because $\pi_{\mathrm{lex}}(s)\in O(s)$ attains the Bellman minimum, $\mathcal B^{\pi_{\mathrm{lex}}}J^\star=\mathcal B_VJ^\star=J^\star$; thus $J^\star$ is the fixed point of the affine $\gamma$-contraction $\mathcal B^{\pi_{\mathrm{lex}}}$, and the selector-evaluation lemma gives $J^{\pi_{\mathrm{lex}}}=J^\star$.  Any value-optimal certified selector must choose at $s$ from $O(s)$, since a choice outside $O(s)$ is value-suboptimal there; maximizing $\marg_V$ over $O(s)$ therefore yields the largest universal margin attainable at $s$ without sacrificing optimality, state by state.  Only comparisons of $\marg_V$ values are used, so the rule is invariant under any strictly decreasing reparametrization and in particular does not depend on $g$.
\end{proof}

\begin{remark}
The two mechanisms occupy different regimes.  The scalarized operator $\mathcal B_V^{\eta}$ admits margin into the objective and pays for it continuously, with deviation bounded by the vanishing-regularization estimate.  The ordinal selector treats value as primary and the universal margin as a strict tie-breaker, so it never trades value for margin and is a secondary value refinement over the Bellman argmin set rather than a separate value result.  Both act only after the uncertified choices have been removed, preserving the separation between the logical constraint and the quantitative objective.
\end{remark}

\subsection{Retained note: proofs as certified choices}

The structure has a dual interpretation.  A choice is a proof obligation reducer because it transforms the current obligation into a set of successor obligations.  The admissible kernel consists of states from which the obligation can be reduced forever without leaving the target region.

\begin{definition}[Obligation transformer]
For $X\subseteq S$, define
\[
\Omega(X)=\{(s,u):u\in A(s),\ T(s,u)\subseteq X\}.
\]
\end{definition}

A pair $(s,u)$ belongs to $\Omega(X)$ exactly when choice $u$ transforms the obligation $X$ into successor obligations still inside $X$.  This is the operational content of choice-cell necessity.

\begin{proposition}[Obligation form of admissible continuation]
A state $s$ belongs to $\Viab(G)$ if and only if $s\in G$ and there exists an infinite proof tree rooted at $s$ whose internal choices are choices and whose children are all admissible successors.
\end{proposition}

\begin{proof}
Let $V=\Viab(G)$.  If $s\in V$, then $V=F_G(V)$ gives $s\in G$ and, for every node $x\in V$, at least one choice $u_x$ with $T(x,u_x)\subseteq V$.  Build a rooted tree by placing $s$ at the root, labelling each node $x$ by the selected choice $u_x$, and adding one child for every successor in $T(x,u_x)$.  The construction can be repeated at every child because all children remain in $V$.  Hence every node of the tree lies in $G$.

Conversely, suppose such a tree exists.  Let $W$ be the set of state labels that occur at roots of subtrees with the same property.  Every $x\in W$ lies in $G$, and the choice written at the root of the corresponding subtree has all its successor children again labelled by elements of $W$.  Therefore $W\subseteq G\cap\Pre_{\exists\Box}(W)=F_G(W)$.  Since $\Viab(G)$ is the greatest fixed point, every post-fixed point of $F_G$ is contained in it.  Thus $s\in W\subseteq\Viab(G)$.
\end{proof}

The proposition explains why stopping needs a certificate.  An infinite proof tree only proves continued certifiability.  A certificate proves authorized termination.

The tree of the previous proposition can be named as a single coinductive object, which is what makes the proofs-as-choices reading technical rather than figurative.

\begin{definition}[Invariance realizer]\label{def:realizer}
Fix a modal invariance target $\varphi$ and write $G=\Sat(\varphi)$.  The \emph{invariance-realizability} relation $\mathcal R\subseteq S\times\mathrm{Tree}$ is the greatest relation such that whenever $r$ realizes $s$ the object $r$ has the form $r=\langle u,(r_{s'})_{s'\in T(s,u)}\rangle$ with $s\in G$, $u\in A(s)$, and, for every successor $s'\in T(s,u)$, the component $r_{s'}$ realizes $s'$.  A \emph{realizer} at $s$ is any $r$ with $(s,r)\in\mathcal R$.  A realizer is \emph{regular} (finite-memory) if it is generated by a state-based map $\rho:V\to U$ via $r_x=\langle\rho(x),(r_{x'})_{x'\in T(x,\rho(x))}\rangle$, so that the chosen choice depends only on the current state.
\end{definition}

The greatest-relation clause is the coinductive counterpart of the greatest-fixpoint definition of the kernel: a realizer never terminates and certifies invariance at every node it unfolds.

\begin{proposition}[Adequacy of realizers]\label{prop:adequacy}
A state $s$ has an invariance realizer if and only if $s\in\Viab(\Sat\varphi)$.  Moreover $s$ has a \emph{regular} realizer if and only if $s\in\Viab(\Sat\varphi)$, and the regular realizers are exactly the memoryless justification selectors $\rho:V\to U$ with $\rho(x)\in A_V(x)$.
\end{proposition}

\begin{proof}
Let $V=\Viab(\Sat\varphi)$.  If $r$ realizes $s$, the set $W=\{x:\exists r',\ (x,r')\in\mathcal R\}$ satisfies, for each $x\in W$, both $x\in G$ and the existence of a choice whose successors all lie in $W$; hence $W\subseteq G\cap\Pre_{\exists\Box}(W)=F_G(W)$, and as $V$ is the greatest fixed point $W\subseteq V$, so $s\in V$.  Conversely, on $V$ the fixed-point equation supplies for each $x$ a choice $\rho(x)\in A_V(x)$; defining $r_x=\langle\rho(x),(r_{x'})_{x'\in T(x,\rho(x))}\rangle$ coinductively yields a realizer at every $s\in V$, and this realizer is regular by construction.  Thus realizability, regular realizability, and kernel membership coincide, and the witnessing regular realizers are precisely the memoryless state-based selectors valued in $A_V$, which are the memoryless admissible strategies of the invariance game.
\end{proof}

\begin{remark}[Realizers are admissible strategies in the model-checking game]
The kernel is defined by the choice-cell necessity formula $\nu Z.(\varphi\wedge[\exists]\Box Z)$, and a modal $\mu$-calculus formula is evaluated by a parity model-checking game in which the verifier maintains the formula's truth along the unfolding \cite{niwinski1996,emerson1991}.  For the present greatest-fixpoint-of-a-conjunction-with-$[\exists]\Box$ formula that game is an invariance game on the arena, and an admissible verifier strategy is exactly an invariance realizer in the sense of Definition~\ref{def:realizer}: the choices at the modal step are choices, and the greatest-fixpoint priority makes infinite plays admissible for the verifier.  Adequacy is therefore the model-checking-game soundness and completeness for this formula, read through the choice vocabulary.
\end{remark}

\begin{remark}[Synthesis as realizability, and the Curry--Howard reading]
Treating $\varphi$ as a specification to be implemented by a reactive selector against an universal environment is realizability in the reactive-semantic construction sense \cite{pnueli1989}: the selector realizes the invariance specification, and $\Viab(\Sat\varphi)$ is exactly the set of initial states from which the specification is realizable.  Under the propositions-as-types reading this is a coinductive instance of the formulas-as-types correspondence: the invariance formula is a greatest-fixpoint proposition, a realizer is its (non-well-founded) proof term, and a choice is the proof step that discharges the one-step obligation $[\exists]\Box$ into successor obligations.  This reading is the standard correspondence between proofs, strategies, and programs specialized to the admissible-predecessor fixpoint, not a new logical principle; its only role here is to make precise the sense in which a justification selector \emph{is} a proof of continued certifiability, which is what the dual view asserts and what the next section's semantic construction theorem then packages with its quantitative layer.
\end{remark}

\section{Appendix A: expanded proof of finite stabilization}

The finite stabilization lemma can be sharpened slightly.  Define $X_0=S$ and $X_{n+1}=F_G(X_n)$.  Since each $X_n$ is a subset of $S$, the sequence lives in a finite lattice.  Since the sequence is descending, every strict step removes at least one state.  Since there are only $|S|$ states, at most $|S|$ strict removals can occur.  Therefore stabilization occurs by step $|S|$.  The stabilized set is a fixed point by definition of stabilization.  Any other fixed point $Y$ of $F_G$ satisfies $Y\subseteq X_n$ for every $n$ by induction.  Hence $Y$ is contained in the stabilized set.  Therefore the stabilized set is the greatest fixed point.

\section{Appendix B: undiscounted drift example}

Discounting is used for contraction, not for invariance.  If $\gamma=1$, uniqueness of the value fixed point can fail.  Let $S=\{s\}$ with one choice and zero cost.  The Bellman equation becomes $J(s)=J(s)$.  Every real value is a solution.  The admissible set is still well-defined, but the value equation is not uniquely pinned down.  This example shows why the construction separates admissible continuation from optimal value.

\section{Appendix C: product construction for temporal objectives}

A temporal objective can be handled by taking a product between the choice-refined structure and a deterministic acceptance automaton.  Let $Q$ be the automaton state set and let $\delta:Q\times2^{\mathsf P}\to Q$ be its transition function.  The product state space is $S\times Q$.  If the original transition sends $s$ to $s'$, the product transition sends $(s,q)$ to $(s',\delta(q,L(s')))$.  The acceptable region $G$ is then defined by the accepting condition in $Q$.  Once this product is formed, the admissible continuation and Bellman constructions apply unchanged.  This is the standard automata-theoretic route from temporal specifications to state-based verification \cite{vardi1986}.

\section{Appendix D: relation to abstract interpretation}

The quotient theorem can be read as an exact abstraction theorem.  If the abstraction map sends states to choice-refined bisimulation classes, then abstract and concrete values agree exactly on classes.  If the abstraction is coarser than choice-refined bisimulation, then the result may become an over-approximation or may become unsound.  Abstract interpretation studies sound approximation of fixed points in a broad lattice-theoretic setting \cite{cousot1977}.  The present theorem is narrower because it gives equality rather than sound inclusion.

\section{Appendix E: bibliographic positioning}

The semantic framing is compared with modal $\mu$-calculus model checking, ATL/ATL$^*$, strategy logic, ATL with indistinguishability, imperfect-information games, and quantitative verification.  The fixed-point substrate is standard; the constructed object is the ordered combination of information-state commitment, certified choice, value refinement, observational indistinguishability under public interfaces, and certificates.  The comparison fixes the terminology: the same predecessor equation is not relabelled as a new modality, and observational indistinguishability is a preservation theorem for certified selectors.

Kripke semantics supplies the relational interpretation of modal formulas \cite{kripke1963}.  Tarski's fixed-point theorem supplies the lattice principle for the greatest admissible continuation region \cite{tarski1955}.  Kozen's modal $\mu$-calculus supplies the fixed-point language in which the kernel can be written as a formula \cite{kozen1983}.  Park and Hennessy--Milner supply the behavioural-invariance template refined later by costs and choice-refined choices \cite{park1981,hennessy1985}.

Pnueli's temporal logic, Clarke--Emerson branching-time verification, and Vardi--Wolper automata-theoretic verification provide the temporal-specification background \cite{pnueli1977,clarke1981,vardi1986}.  Alur--Henzinger--Kupferman and Pauly supply ATL and coalition logic, whose one-step operator $\langle\!\langle C\rangle\!\rangle\bigcirc$ matches choice-cell necessity \cite{alur2002,pauly2002}.  Strategy logic and ATL with indistinguishability mark stronger neighboring formalisms adjacent to the present fragment \cite{chatterjee2010,vdhw2003}.

Bellman, Blackwell, Puterman, Bertsekas, and Shapley provide the dynamic-programming and discounted-game layer \cite{bellman1957,blackwell1965,puterman1994,bertsekas2012,shapley1953}.  Iyengar and Nilim--El~Ghaoui supply state-choice rectangular value recursion \cite{iyengar2005}.  Baier--Katoen and the PRISM line give the probabilistic model-checking neighbour, where probabilities and rewards are explicit rather than replaced by universal successor choice \cite{baier2008,kwiatkowska2011,kwiatkowska2020}.

Bernet--Janin--Walukiewicz, Gr{\"a}del--Thomas--Wilke, and the alternating-refinement literature supply the permissive-strategy and finite-game vocabulary used around admissible continuation \cite{bernet2002,gradel2002,alur1998}.  Paige--Tarjan, Cousot--Cousot, and the alternating-refinement literature frame exact quotienting and abstraction \cite{paige1987,cousot1977,alur1998}.  Reif and Chatterjee--Doyen--Henzinger--Raskin supply the subset construction and game algorithms behind the information-state layer \cite{reif1984,cdhr2007}.

\section{Appendix F: failure taxonomy}

This appendix records the failure modes blocked by the separation between truth, admissible continuation, value, and certification.  The first failure is truth collapse: $s\models\varphi$ does not imply $s\in\Viab(\Sat(\varphi))$; the correction is to require membership in the admissible continuation kernel.  The second failure is cheap exit: a one-step cost minimizer over $A(s)$ need not be admissible; the correction is to minimize only over $A_V(s)$.  The third failure is value stopping: a small value does not imply membership in the certificate set $C$; the correction is a separate certificate predicate.  The fourth failure is quotient drift: label equivalence can preserve formulas while changing values; the correction is value-refined modal bisimulation.  The fifth failure is objective drift: free mode changes can produce endless switching; the correction is a positive switching cost or an explicit mode constraint.  The sixth failure is approximation drift: approximate value comparison may select a near-best non-certified choice; the correction is to keep the exact admissible-choice constraint and approximate only the value comparison inside that constraint.

\begin{proposition}[Invalid implication schema]
None of the following implications is valid in all finite modal choice structures:
\[
s\models\varphi\Rightarrow s\in\Viab(\Sat(\varphi)),
\]
\[
u\in\argmin_{a\in A(s)}c(s,a)\Rightarrow u\in A_{\Viab(G)}(s),
\]
\[
J^\star(s)\leq J^\star(t)\Rightarrow s\in C.
\]
\end{proposition}

\begin{proof}
The first implication fails when a state satisfies $\varphi$ but every certified choice has a successor violating $\varphi$.  The second implication fails in the cheap-exit example from the main text.  The third implication fails in the certificate counterexample.  All witnesses are finite, so the failure is semantic rather than topological.
\end{proof}

Reward-only encodings of formulas fail for the same reason.  A reward ranks alternatives, but it does not express universal successor preservation.  If violation is represented by a large finite penalty, a non-certified choice can still be selected when its immediate gain is large enough.  Infinite penalties emulate constraints, but then the operative object is again a restricted feasible set rather than an unconstrained objective.

\section{Appendix G: complexity of the finite construction}

Let
\[
n=|S|,
\qquad m=\sum_{s\in S}|A(s)|,
\qquad e=\sum_{s\in S}\sum_{u\in A(s)}|T(s,u)|.
\]
These quantities measure states, admissible state-choice pairs, and successor incidences.  The descending admissible continuation computation needs at most $n$ strict rounds.  A naive implementation checks every incidence in every round, giving time $O(ne)$.  A worklist implementation can reduce repeated checks by maintaining counters of successors that remain inside the current candidate region.  Partition refinement can then be applied after the admissible region is known when quotient compression is desired \cite{paige1987}.

\begin{proposition}[Naive admissible continuation complexity]
The naive descending computation of $\Viab(G)$ terminates in at most $n$ rounds and costs $O(ne)$ time.
\end{proposition}

\begin{proof}
There are at most $n$ strict removals of states.  In each round, the algorithm may inspect every successor incidence to decide whether some certified choice has all successors inside the current set.  Thus each round costs $O(e)$ time.  Multiplying by at most $n$ rounds gives $O(ne)$ time.
\end{proof}

The same incidence parameters give the model-checking cost of the choice-refined invariance formula.  Evaluating the base modal part of a fixed formula $\varphi$ takes time linear in the formula size times the transition incidences inspected by its modal clauses.  The outer choice-refined greatest fixed point then performs at most $n$ admissible-predecessor rounds.  Thus the kernel formula $\nu Z.(\varphi\wedge[\exists]\Box Z)$ has the same finite fixed-point complexity profile as the invariance-game computation: repeated predecessor evaluation over the explicit arena.

\begin{proposition}[Choice-refined invariance-fragment model checking]
For an explicitly represented finite structure, the formula $\nu Z.(\varphi\wedge[\exists]\Box Z)$ can be model checked by evaluating $\Sat(\varphi)$ once and then running the descending admissible-predecessor iteration.  With a naive predecessor implementation, the fixed-point part costs $O(ne)$ time and terminates in at most $n$ rounds.
\end{proposition}

\begin{proof}
The denotation of $\varphi$ is a fixed target set $G=\Sat(\varphi)$.  By the fixed-point identity theorem, the full formula denotes $\Viab(G)=\nu X.(G\cap\Pre_{\exists\Box}(X))$.  The descending iteration reaches this greatest fixed point in at most $n$ strict removals.  Each naive round can inspect every successor incidence to test whether a choice has all successors inside the current candidate set, so the fixed-point part costs $O(ne)$.
\end{proof}

Value iteration has a different complexity profile.  One Bellman update over $V$ costs $O(e_V)$, where $e_V$ is the number of successor incidences for certified choices.  The contraction factor choices the number of iterations needed for a prescribed error.  If the residual target is $\rho$, the number of iterations is logarithmic in the initial error scale and inverse-logarithmic in $\gamma$.  This is the standard computational behavior of discounted value iteration \cite{puterman1994}.

\begin{proposition}[Residual stopping rule for value iteration]
If value iteration stops when $r_n\leq(1-\gamma)\varepsilon$, then $\|J_n-J^\star\|_\infty\leq\varepsilon$.
\end{proposition}

\begin{proof}
The a posteriori residual estimate gives $\|J_n-J^\star\|_\infty\leq r_n/(1-\gamma)$.  Substituting the stopping condition yields $\|J_n-J^\star\|_\infty\leq\varepsilon$.
\end{proof}

\section{Appendix H: average-cost and total-cost boundary}

The discounted theorem is sharp at the boundary $\gamma=1$.  If discounting is removed, the Bellman operator is generally nonexpansive rather than contractive.  Nonexpansiveness may be enough for special structures, but uniqueness and convergence require additional hypotheses.  Average-cost dynamic programming typically needs recurrence or communicating assumptions.  Total-cost formulations need transience or properness assumptions.  These alternatives are important, but they are different theorems rather than minor edits of the discounted theorem.

\begin{example}[Multiple undiscounted fixed points]
Let $S=\{s\}$, let there be one choice, let $T(s,u)=\{s\}$, and let $c(s,u)=0$.  With $\gamma=1$, the Bellman equation is
\[
J(s)=J(s).
\]
Every real number solves this equation.  The admissible region is still exactly $\{s\}$.  Hence logical admissible continuation remains unique while the value fixed point loses uniqueness.
\end{example}

This boundary is conceptually important.  Invariance is an invariance property.  Discounted optimality is a contraction property.  The two properties share a transition graph, but they do not share the same mathematical reason for existence.

\section{Appendix I: semantic role of the terminology}

The title phrase ``value-refined modal fixed-point semantics with certified choice and public share-alike certificates'' names an order of semantic operations.  Modal semantics first determines which formulas are true at a state or information state.  The admissible predecessor then determines which choices preserve those commitments for every successor still possible after the choice.  The greatest fixed point of that transformer yields the admissible-continuation kernel: the region in which the local witness keeps its claim certified indefinitely.  Only after this kernel is computed does the Bellman operator compare costs.  The public share-alike fragment uses the same order rather than a separate notice: attribution is a preserved label, derivative release is a continuation relation, ShareAlike is closure under that relation, and the accepted release statement is a certificate inside the already defined calculus.

The central object is therefore the ordered pair
\[
\bigl(\Just(G),\mathcal B_{\Just(G)}\bigr),
\]
with $\Just(G)=\Viab(G)$ in the full-information endpoint and $\Just(G)=\Viab_{\obs}(G)$ on information states.  The first component is logical: it identifies where the selector has a continuing proof obligation it can discharge.  The second component is quantitative: it ranks certified futures by discounted universal value.  Removing the first component permits value refinement to leave the admissible region.  Removing the second component gives commitment preservation without preference.  Removing the certificate component permits the selector to mistake cheap continuation for authorized termination.

Observational indistinguishability is a separate semantic role.  It does not remove the commitment or the proof obligation.  It identifies only the public equivalence class of certified choices.  The observer can still check that the visible behavior is compatible with the selector's modal commitments; what remains quotient-internal is the tie-breaking rule, boundary-margin preference, or data-derived candidate that selected one representative of an equivalence class.

The phrase ``selector transition system'' is read on the same finite structure.  A selector may act under an observation map, so what it can condition on is its knowledge, formalized by the $S5$ modality $K$ over the indistinguishability relation, and what it can keep certified is computed on reachable information states.  The perfect-information collapse shows that when observation is the identity this knowledge layer degenerates: $K$ becomes truth and information states become states.  Thus the full-information construction is not a different theory; it is the singleton-information state instance of the observation-based one.

\section{Appendix J: theorem dependency ledger}

The admissible predecessor is monotone; the logical admissible continuation kernel is the greatest fixed point of the admissible continuation transformer; membership in the kernel is equivalent to an invariance-preserving state-based selector; and the same kernel is the admissible region of the induced turn-based invariance game.  The admissible continuation-restricted Bellman operator is a contraction under discounting, has a unique fixed point, and is solved either by value iteration or by proper finite selector iteration.  Local descent over all certified choices can exit the admissible region, so value refinement must be restricted to $A_V$.  Value-refined modal bisimulation preserves modal truth, modal-definable admissible continuation membership, and optimal admissible value.  Certificate-gated stopping separates termination from low value, and the certified objective is the reach-while-stay region dual to the admissible-continuation kernel.  Positive switching cost bounds the number of objective switches under bounded total cost.  Bellman residuals give an a posteriori approximate-choice guarantee.  The admissible-continuation constraint is largest among constraints that preserve the target.  The base modal language characterizes ordinary bisimulation on the flattened relation in the finite Hennessy--Milner sense, while the choice-refined $\mu$-calculus formula $\nu Z.(\varphi\wedge[\exists]\Box Z)$ characterizes the admissible continuation kernel.  The optimal value is Lipschitz in costs and in the discount factor, monotone under universal successor enlargement on a common admissible region, and admits a secondary, ordinal value refinement of the universal boundary margin over its Bellman argmin set without any change of value.  Beyond the exact equivalence, the value-refined modal bisimulation is the zero set of a canonical pseudometric, the unique fixed point of a Hausdorff-lifted choice-matching transformer over the certified choices, and the optimal value is $1$-Lipschitz with respect to it, so the exact value-preserving quotient is the distance-zero specialization and an $\varepsilon$-quotient preserves optimal value up to $\varepsilon$ (Section~\ref{sec:bisimulation-metric}).

Extensions to stochastic successors, distributed choice selection, average cost, or infinite state spaces change at least one hypothesis in this registry and therefore require separate theorems.  Partial observation is discharged here for the finite universal discounted case by the information-state kernel and information state Bellman theorem.  What remains outside the endpoint is a genuinely probabilistic or distributed observation game in which observations, probabilities, and strategy interaction are all primitive rather than compiled into a finite information state arena.

\end{document}